\pdfoutput =1

\documentclass[pra,onecolumn,superscriptaddress]{article}

\usepackage[T1]{fontenc}
\usepackage[utf8]{inputenc}
\usepackage{amsmath,amsfonts,amsthm,amssymb}
\usepackage{mathtools}
\usepackage{subeqnarray}
\usepackage{setspace}
\usepackage{graphics,graphicx,color}
\usepackage{url}
\usepackage{enumerate}
\usepackage{float} 
\usepackage{multirow}
\usepackage{hhline}
\usepackage[margin=1.0in]{geometry}
\usepackage{braket}
\usepackage{dsfont} 
\usepackage{multirow}
\usepackage{hhline}
\usepackage[makeroom]{cancel}
\usepackage{ifthen}
\usepackage{comment}
\usepackage{breakcites}

\definecolor{ddgreen}{rgb}{.05,.4,.05}
\definecolor{damethyst}{rgb}{0.4, 0.2, 0.6}
\usepackage[colorlinks, linkcolor=damethyst,citecolor=ddgreen]{hyperref}

\newtheorem{theorem}{Theorem}
\newtheorem{cor}{Corollary}
\newtheorem{definition}{Definition}

\newtheorem{lem}{Lemma}

\newtheorem{construction}{Construction}

\newtheorem{remark}{Remark}

\usepackage[scr=boondoxo,scrscaled=1.05]{mathalfa}

\newcommand{\beq}{\begin{eqnarray}}
\newcommand{\eeq}{\end{eqnarray}}

\newcommand{\Gen}{\mathsf{Gen}}

\newcommand{\sk}{\mathsf{sk}}
\newcommand{\pk}{\mathsf{pk}}
\newcommand{\poly}{\mathsf{poly}}
\newcommand{\negl}{\mathsf{negl}}

\newcommand{\niton}{\not\owns}
\newcommand{\comp}{\mathsf{CPhO}}
\newcommand{\X}{\mathcal{X}}
\newcommand{\Y}{\mathcal{Y}}
\newcommand{\K}{\mathcal{K}}
\newcommand{\R}{\mathcal{R}}
\newcommand{\A}{\mathcal{A}}

\newcommand{\supp}{\textsc{Supp}}
\newcommand{\genf}{\textsc{Gen}_{\mathcal{F}}}
\newcommand{\invf}{\textsc{Inv}_{\mathcal{F}}}
\newcommand{\chkf}{\textsc{Chk}_{\mathcal{F}}}
\newcommand{\sampf}{\textsc{Samp}_{\mathcal{F}}}
\newcommand{\geng}{\textsc{Gen}_{\mathcal{G}}}
\newcommand{\invg}{\textsc{Inv}_{\mathcal{G}}}
\newcommand{\chkg}{\textsc{Chk}_{\mathcal{G}}}
\newcommand{\sampg}{\textsc{Samp}_{\mathcal{G}}}

\DeclareMathOperator{\bw}{\mathbf{w}}
\DeclareMathOperator{\iO}{iO}

\usepackage{tikz}

\definecolor{amethyst}{rgb}{0.6, 0.4, 0.8}

\newcommand{\anote}[1]{\textcolor{blue}{[Andrea: #1]}}

\title{Deniable Encryption in a Quantum World}
\author{Andrea Coladangelo\thanks{UC Berkeley \& Simons Institute for the Theory of Computing \& qBraid. Email: \texttt{andrea.coladangelo@gmail.com}} \and Shafi Goldwasser\thanks{UC Berkeley \& Simons Institute for the Theory of Computing. Email: \texttt{shafi.goldwasser@gmail.com}} \and Umesh Vazirani\thanks{UC Berkeley \& Simons Institute for the Theory of Computing. Email: \texttt{vazirani@berkeley.edu}}}
\date{}
\begin{document}

\maketitle
\begin{abstract}

(Sender-)Deniable encryption provides a very strong privacy guarantee: a sender who is coerced by an attacker into ``opening'' their ciphertext after-the-fact is able to generate ``fake'' local random choices that are consistent with any plaintext of their choice. 

In this work, we study (sender-)deniable encryption in a setting where the encryption procedure is a quantum algorithm, but the ciphertext is classical. We show that quantum computation unlocks a fundamentally stronger form of deniable encryption, which we call \emph{perfect unexplainability}.
The primitive at the heart of unexplainability is a quantum computation for which there is provably no efficient way, such as exhibiting the ``history of the computation'', to establish that the output was indeed the result of the computation. 
We give a construction that is secure in the random oracle model, assuming the quantum hardness of LWE. Crucially, this notion implies a form of protection against coercion ``before-the-fact'', a property that is impossible to achieve classically.\footnote{Note: A previous version of this paper also described an alternative notion of quantum deniability called ``$\ell$-deniability'' 
and proposed a construction that was claimed to be ``1-deniable''. This construction was later found to contain a fatal mistake. Thus, we have removed all discussion of $\ell$-deniability from the current version. The authors thank Ashwin Nayak and Rahul Jain for pointing out this mistake. The mistake does not affect what we consider to be the main conceptual notion that we propose and study in this paper, namely ``unexplainability''.}
\end{abstract}

\tableofcontents

\section{Introduction}
This work is motivated by the following overarching question: do \emph{local} quantum computations alone provide an advantage in cryptography? In other words, is there any quantum advantage in the setting where honest parties can perform \emph{local} quantum computations but are restricted to sending and storing \emph{classical} information?


While many examples of quantum advantage are known that leverage quantum communication (or shared entanglement), e.g.\ key distribution \cite{bennett1984quantum, ekert1992quantum} and oblivious transfer \cite{crepeau1988achieving, bartusek2021one, grilo2021oblivious}, or that rely on storing quantum information, e.g.\ quantum money \cite{Wiesner83}, copy-protection \cite{aaronson2009quantum}, and various other unclonable primitives \cite{ben2016quantum, broadbent2019uncloneable, broadbent2020quantum, coladangelo2021hidden}, examples of quantum advantage in the setting where parties rely solely on local quantum computations are much more limited~\cite{brakerski2018cryptographic, liu2022beating, gheorghiu2022quantum}.

Here, we study the notion of \emph{deniable encryption} in this setting. Deniable encryption was introduced by Canetti et al. \cite{canetti1997deniable}. In a deniable encryption scheme, honest parties are able to generate a ``fake'' secret key (in the case of \emph{receiver} deniability) and ``fake'' randomness (in the case of \emph{sender} deniability) to claim that the public communication is consistent with any plaintext of their choice. This allows them to preserve the privacy of the true plaintext even if an adversary coerces them after-the-fact into disclosing their private information. Here, we restrict our attention to non-interactive public-key schemes, and we focus on \emph{sender-deniable} encryption, namely the setting in which we only protect the sender against coercion by an attacker. A bit more formally, a public-key encryption scheme is sender-deniable if there exists a ``faking'' algorithm that takes as input a pair of messages $m_0, m_1$, an encryption $c = \textsf{Enc}(m_0, r)$, and the randomness $r$ used in the encryption, and outputs some ``fake'' randomness $r'$, which should look consistent with a genuine encryption of $m_1$. More precisely, the view of an attacker who receives $m_1$, $c$, and the fake randomness $r'$, should be computationally indistinguishable from the view of an attacker who receives $m_1$, along with a genuine encryption of $m_1$, and the true randomness used. 


In their original paper \cite{canetti1997deniable}, Canetti et al.\ gave a construction of a deniable encryption scheme where the real and fake views are computationally indistinguishable up to inverse polynomial distinguishing advantage - we will call the distinguishing advantage the ``faking probability''.
More generally, they show that the size of the ciphertext grows with the inverse of the faking probability. In a breakthrough work \cite{sahai2014use}, Sahai and Waters gave the first construction with negligible faking probability and compact (i.e polynomial-size) ciphertexts under the assumption that  secure indistinguishability obfuscation ($\iO$) exists. 
A more recent work \cite{agrawal2021deniable}, achieves deniable encryption with compact ciphertexts and negligible faking probability based on polynomial-time hardness of LWE, albeit the running time of the encryption algorithm is non-polynomial.

Despite this progress, an important shortcoming of the classical notion of deniability is that it can only handle coercion \emph{after-the-fact} (i.e.\ the attacker approaches the sender \emph{after} she has sent her ciphertext demanding an explanation), but it fundamentally cannot handle coercion \emph{before-the-fact} (i.e.\ the attacker approaches the sender \emph{before} she sends her ciphertext - even before the public key is revealed - and prescribes the randomness to be used by the sender in their encryption: this clearly rules out any possibility of the sender later ``faking'' an explanation). Thus, in this work, we raise the following question:

\begin{center}
\emph{Is there a notion of deniability achievable in the quantum world \\
that can provide some protection against coercion before-the-fact?} 
\end{center}

We provide an affirmative answer to this question in the setting where encryption is a \emph{quantum} algorithm, and the ciphertext is \emph{classical}. For simplicity, we will refer to this as the \emph{quantum setting}. The notion that we propose, and the corresponding construction, illustrate that quantum computation provides a fundamentally new kind of advantage for deniability, which provides some form of protection against coercion \emph{before-the-fact}.

\paragraph{The difficulty in defining deniability in the quantum setting} When attempting to formulate a definition of deniability in the quantum setting that mimics the classical definition, one quickly runs into two basic obstacles:
\begin{itemize}
\item[1.] The classical definition of deniability is centered around the notion of ``input randomness''. However, the concept of ``input randomness'' is not well-defined in the quantum setting: quantum algorithms can sample randomness by making measurements! Moreover, they can also ``cover their tracks'' by repeatedly performing measurements in an incompatible basis. To illustrate this, one might consider a definition of deniability where the sender provides the attacker with all of her leftover quantum work registers. However, such a definition would be trivial to satisfy by slightly modifying any classical encryption algorithm, and running it on a quantum computer as follows. 
Run the classical encryption algorithm using some register to sample the randomness; at the very end, measure the register containing the randomness in the Hadamard basis, and append the measurement outcome to the ciphertext (and accordingly modify decryption so that it ignores this part of the ciphertext). What this has accomplished is that the input randomness has effectively been ``erased'' as part of the honest encryption algorithm, and the leftover quantum state trivially does not reveal anything to the attacker.
\item[2.] More generally, beyond input randomness, the notion of a \emph{transcript} of a computation is also not well-defined for a quantum algorithm, since ``observing'' the computation at any step disturbs the computation in general.
\end{itemize}

\subsection{Our Contributions}
\label{sec: our results}

\paragraph{Unexplainable encryption} In light of the above, we put forward the notion of \emph{unexplainability}, which takes a different viewpoint on the concept of input randomness: in the classical setting, the input randomness can be viewed as a proxy for a ``proof'' that the ciphertext is a valid encryption of a certain plaintext. The notion of unexplainability formalizes this notion of a proof, and achieves a single definition that provides a natural common view of deniability in the classical and quantum setting.

In an unexplainable encryption scheme, it is computationally intractable for a sender, except with negligible probability, to ``prove'' that they encrypted a particular plaintext (thus an attacker has simply no reason to bother coercing a sender into opening their ciphertext in the first place: the sender would simply not be able to provide a convincing proof, even if they wanted to). Notice that, for encryption schemes with perfect decryption (i.e.\ where there is a unique plaintext consistent with a given ciphertext), unexplainability is impossible to achieve classically: the input randomness used to generate the ciphertext, along with the plaintext $m$, constitutes a proof that the ciphertext is an encryption of $m$. This impossibility does not apply to quantum encryption algorithms, for which randomness may inherently be the result of a measurement.

The crux in formalizing the definition of unexplainability is to formalize what it means for a sender to ``prove'' that they encrypted a particular plaintext (without referring to any input randomness). We argue that the appropriate notion of a proof is akin to that of an \emph{argument}. We start by defining ``explainability'', and then define unexplainability by taking the contrapositive. The following definition is not restricted to encryption schemes with perfect decryption, so it is achievable classically (and in some sense, provides a unified view of classical and quantum deniability). 

\begin{definition}[Explainability (informal)]
\label{def: exp}
A public-key encryption scheme is
\emph{explainable} if there exists an efficient verification procedure $\mathsf{Verify}$, taking as input a tuple of public key, ciphertext, message, and alleged proof $(\mathsf{pk},c,m, w)$, such that $\mathsf{Verify}(\mathsf{pk},c,m,w) = 0$ if the triple $(\mathsf{pk},c,m)$ is inconsistent. Moreover, $\mathsf{Verify}$ should satisfy the following:
\begin{itemize}
    \item \emph{Completeness}: there exists an efficient procedure that, on input $m$, $\mathsf{pk}$, generates $c,w$ such that $\mathsf{Verify}(\mathsf{pk},c,m,w)$ outputs $1$ with high probability.
    \item \emph{Soundness}: no efficient procedure, on input $\mathsf{pk}$, $m, m'$ with $m\neq m'$, can generate $c, w, w'$ such that $\mathsf{Verify}(\mathsf{pk}, c,m,w)$ and $\mathsf{Verify}(\mathsf{pk}, c,m',w')$ both output $1$, except with negligible probability.
\end{itemize}
\end{definition}

By contrapositive, a scheme is \emph{unexplainable} if, for any such $\mathsf{Verify}$, one of completeness or soundness fails. Thus, the sense in which a classical deniable scheme is unexplainable is that a sender can never convincingly ``prove'' to an attacker that they encrypted a particular plaintext (even if the sender wishes to do so honestly): this is precisely because, by the classical definition of deniability, it is always possible for a sender to efficiently generate randomness consistent with \emph{any} plaintext of their choice. In other words, a classical deniable encryption scheme is unexplainable because the soundness condition above fails when one takes $\mathsf{Verify}$ to be the procedure that interprets $w$ as the randomness used in the encryption, and simply checks that $c = \mathsf{Enc}(\pk, m; w)$ (technically, unexplainability, as stated, requires the soundness condition to fail for all $\mathsf{Verify}$, though this distinction is not important for the intuition in this introduction - we refer the reader to Section~\ref{sec: den vs unexp} for a more detailed discussion of the relationship between deniability and unexplainability).

The notion of unexplainability can be thought of as a rephrasing of deniability from a different perspective. In fact, we show that the appropriate variation on the definition of unexplainability is \emph{equivalent} to deniability. However, unlike deniability, the notion of unexplainability has a very natural extension to the quantum setting, as it is not centered around randomness, but rather, more abstractly, around the notion of a proof: without modifications, the definition above makes sense even in the quantum setting (where one may choose to allow the ``proof'' $w$ to be a quantum state).



Notice that for an encryption scheme with \emph{perfect} decryption (i.e. one in which there is a unique plaintext consistent with a given ciphertext), the notion of a proof described above coincides with that of an $\mathsf{NP}$ (or $\mathsf{QMA}$) proof: since we require that $\mathsf{Verify}(\mathsf{pk},c,m,w) = 0$ if the triple $(\mathsf{pk},c,m)$ is inconsistent, then $\mathsf{Verify}$ is precisely an $\mathsf{NP}$-relation (or $\textsf{QMA}$-relation) for the language $$L = \{x = (m, c, \mathsf{pk}) : c \text{ is a valid encryption of }m \text{ under } \mathsf{pk}\}\,.$$
The definition then has the additional requirement of \emph{completeness}, which asks that there is an efficient procedure that takes as input $\pk, m$, and generates valid $c,w$.

As anticipated earlier, for encryption schemes with perfect decryption, unexplainability is impossible to achieve classically. First notice, that for a scheme with perfect decryption, the \emph{soundness} condition in Definition \ref{def: exp} is trivially satisfied. Thus, the scheme is unexplainable if and only if for all $\mathsf{Verify}$ (as in Definition \ref{def: exp}) the \emph{completeness} condition fails. However, the latter is contradicted by the following. Consider the $\mathsf{NP}$-relation $\mathsf{Verify}\left((m,c,\mathsf{pk}), w\right)$ where the witness is the randomness used to encrypt, i.e. $\mathsf{Verify}\left((m,c,\mathsf{pk}), w\right) = 1$ if $\mathsf{Enc}(\mathsf{pk}, m; w) = c$, and $\mathsf{Verify}\left((m,c,\mathsf{pk}), w\right) = 0$ otherwise. Then, simply consider the efficient procedure that encrypts honestly, and outputs the randomness as the witness.




In this work, we show the following.
\begin{theorem}[Informal]
\label{thm: 2 inf}
There exists an unexplainable public-key encryption scheme with perfect decryption, with security in the quantum random oracle model (QROM), assuming the quantum hardness of LWE.
\end{theorem}

The encryption scheme that makes the theorem true is simple, and is inspired by the use of trapdoor claw-free functions in \cite{brakerski2018cryptographic, mahadev2018classical}, and in the follow-up work~\cite{brakerski2020simpler}, which makes use of a random oracle.

The public key for the encryption scheme is a choice of trapdoor claw-free function pair $(f_{k,0}, f_{k,1})$. The encryption of a single bit $m$ is a triple $(z,d,y)$ such that $z = m \oplus d \cdot (x_0\oplus x_1) \oplus H(x_0) \oplus H(x_1) $, where $x_0, x_1$ are the pre-images of $y$. As mentioned earlier, while the encryption algorithm is quantum, the ciphertext is \emph{classical}. In  other words, the bit $d \cdot (x_0 \oplus x_1) \oplus H(x_0) \oplus H(x_1)$ is used as a \emph{one-time pad}. CPA security then follows straightforwardly from the fact that $d \cdot (x_0 \oplus x_1) \oplus H(x_0) \oplus H(x_1)$ is a hardcore bit, i.e.\ that it is computationally hard to guess given a uniformly random $d$ and $y$. The need for a random oracle stems from the need to obtain a more rigid characterization of the structure of algorithms that produce valid encryptions (we discuss at the end why obtaining a scheme from just LWE may be difficult).

We show that our scheme is unexplainable in the strongest sense: it is simply impossible, except with negligible probability, for an efficient quantum algorithm to produce both a valid encryption \emph{and} a proof of its validity. We emphasize that the latter guarantee holds for \emph{any} efficient quantum algorithm that attempts to produce both a valid (and hence classical) encryption and a (possibly quantum) proof, not just for algorithms that run \emph{honest} encryption. Using the terminology introduced earlier, we show that for any verification procedure $\mathsf{Verify}$ such that $\mathsf{Verify}(\mathsf{pk},c,m,w) = 0$ if the triple $(\mathsf{pk},c,m)$ is inconsistent, the \emph{completeness} condition fails. We call this \emph{perfect unexplainability}. Notice that perfect unexplainability is impossible to achieve classically, since one can always take $\mathsf{Verify}$ to be the procedure that interprets $w$ as the randomness in the encryption, and runs encryption forward to check consistency. 

Our work provides a new perspective on the notion of proofs in a quantum world by formalizing the notion of ``inability to prove’’: this is captured by a computational problem that can be solved efficiently, but for which it is impossible to concurrently provide an efficiently checkable proof of correctness.
Classically, the transcript of a computation serves as an efficiently checkable proof that the output is correctly reported. At the heart of the new property is the fact that quantum computations, in general, lack the notion of a \emph{transcript} or a ``trajectory'' of the computation. In this work, we identify a computational problem for which it is \emph{provably intractable} to find a solution while concurrently extracting any meaningful proof that the solution is correct. Using the terminology introduced earlier, a proof is a witness to an appropriate NP (or QMA)-relation $\mathsf{Verify}$.

\vspace{2mm}
The analysis in our security proof uses Zhandry's compressed oracle technique \cite{Zhandry-how}, and relies on two novel technical contributions.
\begin{itemize}
    \item[(i)] The first key step in the security proof is establishing that a strategy which produces valid encryptions must be close to a strategy that queries the oracle at a uniform superposition of the two pre-images (in a sense made more precise in the main text). In a bit more detail, let $H: \{0,1\}^n \rightarrow \{0,1\}$ be a uniformly random function, and let $x_0, x_1 \in \{0,1\}^n$. We prove a technical lemma that characterizes the structure of strategies that are successful at guessing $H(x_0) \oplus H(x_1)$. We use this lemma to derive a corresponding ``rigidity'' theorem for strategies that produce valid encryptions (equivalently, valid ``equations'' in the terminology of \cite{brakerski2018cryptographic}). 
     This rigidity theorem may find applications elsewhere, and may be of independent interest.
    \item[(ii)] The second technical contribution is an \emph{online} extraction argument which allows to extract a claw from any prover that produces a valid equation \emph{and} a proof of its validity, with non-negligible probability. The extraction is \emph{online} in the sense that no rewinding is required.
\end{itemize}

We point out that the scheme that makes Theorem \ref{thm: 2 inf} true can be instantiated with any 2-to-1 (noisy) trapdoor claw-free function pair (we do not require an \emph{adaptive} hardcore bit property, as in~\cite{brakerski2018cryptographic,mahadev2018classical}).

\begin{remark}
Although a perfectly unexplainable encryption scheme based solely on, say hardness of LWE without the use of a random oracles is desirable, we point out that it is unlikely that such a result can be achieved without any additional assumption. The reason is that it would imply a \emph{single-round} message-response proof of quantumness protocol \cite{brakerski2018cryptographic}. The following is the protocol:
\begin{itemize}
    \item The verifier samples $(\pk, \mathsf{sk})$ as in the encryption scheme, together with a message $m$ (say uniformly at random). The verifier sends $\pk$ and $m$ to the prover.
    \item The prover returns an encryption $c$ of $m$ under public key $\pk$.
    \item The verifier checks that $c$ decrypts to $m$.
\end{itemize}
Since perfect unexplainability is impossible to achieve classically, it must be that the encryption algorithm is quantum, and in particular that it cannot be replaced by a classical algorithm (otherwise there would be a way to ``explain'' by providing the input randomness). Now, as originally pointed out in \cite{brakerski2020simpler}, a single-round proof of quantumness immediately implies a separation of the sampling classes $\mathsf{BPP}$ and $\mathsf{BQP}$. Such a separation does not seem to be implied by the hardness of LWE, as the current state-of-the-art suggests that LWE is equally intractable for classical and quantum computers.
\end{remark}

\paragraph{Protection against coercion ``before-the-fact''} We find the new phenomenon to be quite striking, beyond its implications for deniable encryption, and we believe that it has the potential to find applications in other settings. As a first example, our encryption scheme provides protection against coercion \emph{before-the-fact}, in the sense that the attacker approaches the sender \emph{before} she sends her ciphertext - even before the public key is revealed. This type of coercion is a major concern for electronic elections with online encrypted votes. An encryption scheme providing this type of protection is impossible in the classical world --- a coercer who approaches a sender prior to sending an encrypted message can dictate which randomness shall be used by the sender when computing the encryption (and which plaintext shall be encrypted). In this way, the coercer can, at a later stage, check that the ciphertext submitted by the sender corresponds to an encryption of the desired plaintext with the prescribed randomness. This holds true even if the public key (or some other public information to be used in the encryption) is revealed \emph{after} coercion. Our unexplainable encryption scheme prevents coercion before-the-fact since the randomness cannot be controlled, even by the sender themselves: it is instead the result of a carefully chosen measurement. Thus, in a scenario where the public key (or some other public information to be used in the encryption) is revealed \emph{after} the coercion stage, a before-the-fact coercer would not be able to succeed. In particular, there is simply no way for an attacker to prescribe to the sender \emph{how} to encrypt in a way that it can later verify, because this would immediately violate perfect unexplainability: it would imply that there exists \emph{some} quantum algorithm that outputs ciphertexts \emph{and} proofs of their validity. We discuss coercion before-the-fact in more detail in Section \ref{sec: before-the-fact}. As a further practical motivation, the fact that it is not possible for a sender, even if they wanted to, to prove that their ciphertext is a valid encryption of their plaintext provides a way to protect against ``vote-selling''. 

\paragraph{Stronger protection against coercion with quantum ciphertexts} While this work focuses on quantum encryption algorithms with \emph{classical} ciphertexts, we briefly discuss (in Section~\ref{sec:quantum-ciphertexts}) how one can obtain an even stronger notion of protection against coercion before-the-fact by considering encryption schemes where part of the ciphertext is a quantum state. We have discussed above how our encryption scheme from Theorem~\ref{thm: 2 inf} provides protection against coercion before-the-fact when the public key has not yet been revealed. Now, what if an attacker knows the public key? Is there any way to protect against coercion before-the-fact? If the ciphertext is classical, then there is no hope: the attacker can just honestly compute a ciphertext, and prescribe to the sender that she should submit this ciphertext. Since the ciphertext is classical (and known to the attacker), the attacker can always verify that the sender submitted what was prescribed. However, if part of the ciphertext is allowed to be a \emph{quantum} state, then it becomes in principle possible that the attacker might not be able to verify!

In fact, a variation of the scheme of Theorem~\ref{thm: 2 inf} achieves the stronger notion of protection against coercion before-the-fact. The main idea is the following: recall that, in the scheme of Theorem~\ref{thm: 2 inf}, the encryption of a single bit $m$ is a triple $(z,d,y)$, where $z = m \oplus d \cdot (x_0 \oplus x_1) \oplus H(x_0) \oplus H(x_1)$. In the variation of the scheme, the encryption is instead a pair $(\ket{\psi}, y)$, where $$\ket{\psi} =  \frac{1}{\sqrt{2}} \ket{0} \ket{x_0} + (-1)^m \frac{1}{\sqrt{2}}\ket{1} \ket{x_1}\,,$$ 
(and the random oracle $H$ is no longer necessary).
Note that this ciphertext can be generated efficiently by applying a $Z^m$ gate to the first qubit of the uniform superposition of $x_0$ and $x_1$. Moreover, a ciphertext of this form has the following important property:
\begin{itemize}
\item[(i)] A sender that receives such a ciphertext from an attacker can ``obliviously'' flip the encoded plaintext by applying a $Z$ gate to the first qubit.
\item[(ii)] The attacker cannot tell the difference between the original ciphertext and the ``flipped'' one. It is not difficult to see that an efficient algorithm to distinguish between the two ciphertexts would imply an efficient algorithm to extract a claw (for this reduction to work, it is crucial that part of the ciphertext received from the attacker, namely the string $y$, is \emph{classical} - in the practical scenario of an election, this can be enforced, for example, by having the government require that $y$ be sent over a classical communication channel).
\end{itemize}



\subsection{Related Work and Concepts}
In the classical setting, \emph{receiver-deniable} encryption has also been studied. In the latter, an attacker coerces the \emph{receiver} after-the-fact into revealing their secret key. While non-interactive receiver-deniable encryption is impossible classically, there exists a generic transformation that compiles any sender-deniable encryption scheme into a receiver-deniable one, at the cost of one additional message \cite{canetti1997deniable}. While we do not formalize this explicitly, such a transformation also applies to the quantum setting.

In an interactive setting, one can also consider the notion of \emph{bideniable} encryption, i.e.\ an encryption scheme that is simultaneously sender and receiver-deniable. Classically, this setting has been considered by Canetti, Park and Poburinnaya \cite{canetti2020fully}, who show that bideniable encryption can be realized from iO. We leave it as an open question to improve this result (i.e. realize bideniable encryption from weaker assumptions) in the quantum setting using \emph{classical} communication. 

Although we do not formalize this, we remark that, if one allows for \emph{quantum} communication, deniable encryption (including bideniable) becomes immediate in the interactive setting (with information-theoretic security!). The reason is that sender and receiver can share an information-theoretically secure key by running a quantum key distribution protocol, and then use this key as a one-time pad. If the attacker approaches the parties \emph{before} the key distribution protocol is completed, then the parties can trivially deny since no information about the message has been revealed yet. If the attacker approaches the parties \emph{after} the key distribution protocol is completed, then the key is information-theoretically hidden from the attacker, and the parties can trivially find a key that is consistent with any message of their choice.

\subsection*{Acknowledgements}
This work was carried out while A.C. was a Quantum Postdoctoral Fellow at the Simons Institute for the Theory of Computing supported by NSF QLCI Grant No. 2016245. S.G. is supported by DARPA under agreement No. HR00112020023. U.V is supported by Vannevar Bush faculty fellowship N00014-17-1-3025, MURI Grant FA9550-18-1-0161, and NSF QLCI Grant No. 2016245. The authors thank Ashwin Nayak and Rahul Jain for pointing out a mistake in a previous version of the paper related to the alternative notion of ``$\ell$-deniability'' (distinct from the notion of unexplainability, which is the main notion that we propose and study in our paper). The mistake invalidates a construction of $1$-deniability presented in our earlier version.



\section{Technical Overview}
\label{sec: tech overview}
In this section, we provide an informal overview of the notion of \emph{unexplainability}, our construction, and its proof of security. 
\subsection{Notation}
In this overview, we let $\{(f_{k,0},f_{k,1})\}_k$ be a family of trapdoor claw-free function pairs. We assume that $f_{k,0},f_{k,1}$ are injective with identical range, and that they map $n$-bit strings to $m$-bit strings. In our actual scheme, we will instead use \emph{noisy} trapdoor claw-free functions, since these can be constructed from LWE. However, this distinction is immaterial for the purpose of this overview. 
\subsection{An unexplainable encryption scheme}
The construction is simple, and is inspired by the ``proof of quantumness'' construction in \cite{brakerski2020simpler}. To obtain an encryption scheme, the idea is to use the ``hardcore'' bit in their construction as a one-time pad. 

In more detail, the public key in the encryption scheme is a choice of trapdoor claw-free function $k$, and the secret key is a trapdoor $t_k$.
\begin{itemize}
    \item \emph{Encryption}: To encrypt a bit $m$, under public key $k$, Compute a triple $(z,d,y)$, where $d \in \{0,1\}^n$ and $y \in \mathcal{Y}$, and $z = d \cdot (x_0^y \oplus x_1^y) \oplus H(x_0^y) \oplus H(x_1^y)$. This can be done as follows:
    \begin{itemize}
        \item Create the uniform superposition over $n+1$ qubits 
        $\sum_{b\in\{0,1\} , x \in \{0,1\}^n} \ket{b}\ket{x}$. Then, compute $f_{k,0}$ and $f_{k,1}$ in superposition, controlled on the first qubit. The resulting state is
        $$ \sum_{b\in\{0,1\} , x \in X} \ket{b}\ket{x}\ket{f_{k,b}(x)} \,.$$
        \item Measure the image register, and let $y$ be the outcome. As a result, the state has collapsed to:
        $$ \frac{1}{\sqrt{2}} (\ket{0} \ket{x_0^y} + \ket{1} \ket{x_1^y}) \,.$$
        \item Query the phase oracle for $H$, to obtain:
        $$\frac{1}{\sqrt{2}} ((-1)^{H(x_0)}\ket{0} \ket{x_0} + (-1)^{H(x_1)}\ket{1} \ket{x_1}) \,. $$
        \item Apply the Hadamard gate to all $n+1$ registers, and measure. Parse the measurement outcome as $z || d$ where $z \in \{0,1\}$ and $d \in \{0,1\}^n$. 
    \end{itemize}
    The ciphertext is $c = (m \oplus z, d, y)$.
    \item \emph{Decryption}: On input $c = (\tilde{z},d,y)$, use the trapdoor $t_k$ to compute the pre-images $x_0^y, x_1^y$. Output $\tilde{z} \oplus d \cdot (x_0^y \oplus x_1^y) \oplus H(x_0^y) \oplus H(x_1^y)$.
\end{itemize}

The actual scheme will be a parallel repetition of this, i.e. the plaintext $m$ is encrypted many times using the single-shot scheme described here. However, for the purpose of this overview, we will just consider the single-shot scheme.


\subsection{Security}
It is straightforward to see that the scheme satisfies CPA security. This essentially follows from a regular ``hardcore bit'' property, satisfied by the trapdoor claw-free function family (more details in Section \ref{sec: security scheme 1}) In this overview, we focus on outlining how the scheme satisfies (a strong version of) unexplainability.

Our main result is that this scheme has the property that, although it is possible to encrypt, it is not possible to simultaneously produce both a valid encryption of a desired plaintext $m$ \emph{and} a ``proof'' or a ``certificate'' $w$ that the ciphertext indeed is an encryption of $m$. We refer to this as \emph{perfect unexplainability}.

More precisely, we show that, for any efficient algorithm $\mathsf{Verify}$ taking as input a tuple $(\mathsf{pk},c,m, w)$, such that $\mathsf{Verify}(\mathsf{pk},c,m,w) = 0$ if the triple $(\mathsf{pk},c,m)$ is inconsistent, the following holds: for any efficient algorithm $P$, for any $m$,
\begin{equation}\Pr[\mathsf{Verify}(\mathsf{pk}, c, m, w) = 1: (c,w) \gets P(\mathsf{pk}, m)] = \negl(n) \,. \label{eq: 1}
\end{equation}

Notice that perfect unexplainability is impossible to achieve with a classical encryption scheme, because one can always have $w$ play the role of the randomness in the encryption, and have $\mathsf{Verify}$ be the algorithm that simply outputs $1$ if $\mathsf{Enc}(\mathsf{pk}, m ; w) = c$, and $0$ otherwise.

A consequence of this observation is that perfect unexplainability is an even stronger property than a ``proof of quantumness'': if an encryption scheme is unexplainable, then it must be that producing valid encryptions (with high enough probability) is something that only a quantum computer can do (otherwise, there would exist a classical algorithm $\mathsf{Enc}$ that encrypts successfully, and, just like before, we can set $\mathsf{Verify}$ to be the procedure that interprets $w$ as randomness, runs $\mathsf{Enc}$ with this randomness, and checks consistency with the ciphertext).

For simplicity, in the rest of this discussion, we will take the plaintext to be $m=0$, so that producing a valid ciphertext is equivalent to producing a valid equation, i.e. a triple $(z,d,y)$, such that $z = d \cdot (x_0^y \oplus x_1^y) \oplus H(x_0^y) \oplus H(x_1^y)$.

So, what is the intuition for why the property above holds for our scheme?
At a high level, the intuition is that in order to produce a valid equation, one \emph{has to} query the oracle on a superposition of the two pre-images, and have the two branches interfere in just the right way, so as to produce $z,d$ which satisfy the equation. This delicate process of interference between the two branches essentially requires that nothing is ``left behind'' in the workspace. 

The difficulty with formalizing this intuition is that it assumes a global view of the entire computation. Whereas unlike in the classical setting, we cannot in general define the trajectory which in this case corresponds to the sequence of oracle queries that led to a particular configuration. To do so, we make use of Zhandry's compressed oracle technique, which leverages properties of uniformly random oracles. It turns out that the cryptographic assumptions interact very cleanly with the compressed oracle formalism. 
This provides a way of carrying key quantities from our classical intuition over into natural formalizations in the quantum context. 

This reduction involves two broad steps:
\begin{itemize}
    \item[(i)] We establish a \emph{rigidity} theorem which formalizes the intuition that an algorithm $P$ which is successful at producing valid equations must query the oracle on a superposition of both pre-images, in a sense that we will make more precise below. Concisely, we appeal to Zhandry's compressed oracle technique for ``recording queries'' \cite{Zhandry-how}, which formalizes the idea that, when the oracle is uniformly random, there is a meaningful way to record the queries made by the algorithm efficiently, in a way that is well-defined and does not disrupt the run of the algorithm. In a compressed oracle simulation, the quantum state of the algorithm at any point is in a superposition over \emph{databases} of queried inputs.
    What we establish is that, if $P$ produces a valid equation with high probability, then the quantum state of $P$ right before measurement of the equation must be close to a uniform superposition (with the appropriate phase) of two branches: one on which the first pre-image was queried, and one in which the second was queried. 
    \item[(ii)] We leverage (i) to construct an algorithm that extracts a claw. One crucial observation here is that, since Verify never accepts an inconsistent tuple $(\mathsf{pk},c,m)$, then it must be the case that whenever Verify accepts, it \emph{must} itself have queried at a superposition of both pre-images, in the same sense as in the rigidity theorem in point (i). At a high level, the extraction algorithm will eventually be the following:
    \begin{itemize}
        \item Run $P$, followed by a measurement to obtain $z,d,y$, and an appropriate measurement of the database register, hoping to find a pre-image of $y$;
        \item Then, run $\mathsf{Verify}$ on the leftover state, and conditioned on ``accept'', measure the database register, hoping to find the \emph{other} pre-image of $y$.
    \end{itemize}
    It is not a priori clear that this leads to successfully recovering a claw (in particular the measurement of the database right after running $P$, may disrupt things in a way that Verify no longer accepts). What we show is that if, to begin with, $\mathsf{Verify}$ accepts with high enough probability, then this extraction strategy works.
\end{itemize}

In the rest of the section, we discuss the elements of the security proof in a bit more detail. 

\subsubsection{Zhandry's compressed oracle technique}
\label{sec: comp oracles tech overview}

In this subsection, we give an exposition of Zhandry's compressed oracle technique, both because it is a building block in our proof of security, and also 
to encourage its broader use. A reader who is familiar with the technique should feel free to skip this subsection. 

Let $H: \{0,1\}^n \rightarrow \{0,1\}$ be a fixed function. For simplicity, in this overview we restrict ourselves to considering boolean functions (since this is also the relevant case for our scheme).

While classically it is always possible to record the queries of the algorithm, in a way that is undetectable to the algorithm itself, this is not possible in general in the quantum case. The issue arises because the quantum algorithm can \emph{query in superposition}. We illustrate this with an example.

Consider an algorithm that prepares the state $\frac{1}{\sqrt{2}}(\ket{x_0} + \ket{x_1})\ket{y}$, and then makes an oracle query to $H$. The state after the query is:
\begin{equation}
    \frac{1}{\sqrt{2}}\ket{x_0}\ket{y\oplus H(x_0)} + \frac{1}{\sqrt{2}}\ket{x_1}\ket{y\oplus H(x_1)}
\end{equation}

Suppose we additionally ``record'' the query made, i.e. we copy the queried input into a third register. Then the state becomes:
\begin{equation}
    \frac{1}{\sqrt{2}}\ket{x_0}\ket{y\oplus H(x_0)}\ket{x_0} + \frac{1}{\sqrt{2}}\ket{x_1}\ket{y\oplus H(x_1)}\ket{x_1}
\end{equation}

Now, suppose that $H(x_0) = H(x_1)$, then it is easy to see that, in the case where we didn't record queries, the state of the first register after the query is exactly $\frac{1}{\sqrt{2}}(\ket{x_0} + \ket{x_1})$. On the other hand, if we recorded the query, then the third register is now entangled with the first, and as a result the state of the first register is no longer $\frac{1}{\sqrt{2}}(\ket{x_0} + \ket{x_1})$ (it is instead a mixed state). Thus, recording queries is not possible in general without disturbing the state of the oracle algorithm.

Does this mean that all hope of recording queries is lost in the quantum setting? It turns out, perhaps surprisingly, that there is a way to record queries when $H$ is a \emph{uniformly random} oracle.

When thinking of an algorithm that queries a uniformly random oracle, it is useful to purify the quantum state of the algorithm via an oracle register (which keeps track of the function that is being queried). An oracle query is then a unitary that acts in the following way on a standard basis element of the query register (where we omit writing normalizing constants):
$$ \ket{x}\ket{y} \sum_{H} \ket{H} \mapsto \sum_{H} \ket{x}\ket{y\oplus H(x)} \ket{H}   \,.$$
It is well-known that, up to applying a Hadamard gate on the $y$ register before and after a query, this oracle is equivalent to a ``phase oracle'', which acts in the following way:
\begin{equation}
    \ket{x}\ket{y} \sum_{H} \ket{H} \mapsto \sum_{H} (-1)^{y\cdot H(x)} \ket{x}\ket{y} \ket{H} \label{eq: 45}
\end{equation} 

Now, to get a better sense of what is happening with each query, let's be more concrete about how we represent $H$ using the qubits in the oracle register.

A natural way to represent $H$ is to use $2^n$ qubits, with each qubit representing the output of the oracle at one input, where we take the inputs to be ordered lexicographically. In other words, if $\ket{H} = \ket{t}$, where $t \in \{0,1\}^{2^n}$, then this means that $H(x_i) = t_i$, where $x_i$ is the $i$-th $n$-bit string in lexicographic order. Using this representation, notice that 
$$ \frac{1}{\sqrt{2^n}} \sum_{H} \ket{H} = \ket{+}^{\otimes 2^n} \,.$$
Now, notice that we can write the RHS of \eqref{eq: 45} as 
$$ \ket{x}\ket{y} \sum_H (-1)^{y \cdot H(x)} \ket{H} \,,$$
i.e. we can equivalently think of the phase in a phase oracle query as being applied to the oracle register.

Thus, when a phase oracle query is made on a standard basis vector of the query register $\ket{x}\ket{y}$, all that happens is:
$$ \sum_{H} \ket{H} \mapsto \sum_{H} (-1)^{y\cdot H(x)} \ket{H} \,.$$

Notice that, using the representation for $H$ that we chose above, the latter transformation is:
\begin{itemize}
    \item When $y = 0$,  $$\ket{+}^{\otimes 2^n} \mapsto \ket{+}^{\otimes 2^n} \,. $$
    \item When $y = 1$,  $$\ket{+}^{\otimes 2^n} \mapsto \ket{+}\cdots \ket{+}_{i-1}\ket{-}_i\ket{+}_{i+1}\cdots \ket{+} \,,$$
    where $i$ is such that $x$ is the $i$-th string in lexicographic order.
\end{itemize}
In words, the query does not have any effect when $y=0$, and the query flips the appropriate $\ket{+}$ to a $\ket{-}$ when $y=1$. Then, when we query on a general state $\sum_{x,y} \alpha_{xy} \ket{x}\ket{y}$, the state after the query can be written as:
$$\sum_{x,y} \alpha_{xy} \ket{x}\ket{y} \ket{D_{xy}} \,,$$
where $D_{xy}$ is the all $\ket{+}$ state, except for a $\ket{-}$ corresponding to $x$ if $y=1$.

The crucial observation now is that all of these branches are \emph{orthogonal}, and thus it makes sense to talk about "the branch on which a particular query was made": the state of the oracle register reveals exactly the query that has been made on that branch. More generally, after $q$ queries, the state will be in a superposition of branches on which at most $q$ of the $\ket{+}$'s have been flipped to $\ket{-}$'s. These locations correspond exactly to the queries that have been made.

Moreover, the good news is that there is a way to keep track of the recorded queries \emph{efficiently}: one does not need to store all of the (exponentially many) $\ket{+}$'s, but it suffices to keep track only of the locations that have flipped to $\ket{-}$ (which is at most $q$). If we know that the oracle algorithm makes at most $q$ queries, then we need merely $n \cdot q$ qubits to store the points that have been queried. We will refer to the set of queried points as the \emph{database}. Formally, there is a well-defined isometry that maps a state on $2^n$ qubits where $q$ of them are in the $\ket{-}$ state, and the rest are $\ket{+}$, to a state on $n \cdot q$ qubits, which stores the $q$ points corresponding to the $\ket{-}$'s in lexicographic order.

Let $D$ denote an empty database of queried points. Then a query to a uniformly random oracle can be thought of as acting in the following way:
$$ \begin{cases} 
\ket{x}\ket{y} \ket{D} \mapsto \ket{x}\ket{y} \ket{D} \,, \textnormal{ if } y =0\\ 
\ket{x}\ket{y} \ket{D} \mapsto \ket{x}\ket{y} \ket{D \cup \{x\}} \,, \textnormal{ if } y=1 \,.
\end{cases}$$

Such a way of implementing a uniformly random oracle is referred to as a \emph{compressed phase oracle} simulation \cite{Zhandry-how}.
Formally, the fact that the original and the compressed oracle simulations are \emph{identical} from the point of view of the oracle algorithm (which does not have access to the oracle register) is because at any point in the execution of the algorithm, the states in the two simulations are both purifications of the same mixed state on the algorithm's registers.

We point out that there are two properties of a uniformly random oracle that make a compressed oracle simulation possible:
\begin{itemize}
    \item The query outputs at each point are independently distributed, which means that the state of the oracle register is always a product state across all of the $2^n$ qubits.
    \item Each query output is uniformly distributed. This is important because in general $\alpha \ket{0} + \beta\ket{1} \not \perp \alpha \ket{0} - \beta\ket{1}$ unless $|\alpha| = |\beta|$.
\end{itemize}

Notice that the above compressed oracle simulation does not explicitly keep track of the value of the function at the queried points (i.e. a database is just a set of queried points). In the following slight variation on the compressed oracle simulation, also from \cite{Zhandry-how}, a database is instead a set of pairs $(x,w)$ representing a queried point and the value of the function at that point. This variation will be more useful for our analysis.

Here $D$ is a database of pairs $(x,v)$, which is initially empty. A query acts as follows on a standard basis element $\ket{x}\ket{y} \ket{D}$:
\begin{itemize}
    \item If $y = 0$, do nothing.
    \item If $y = 1$, check if $D$ contains a pair of the form $(x,v)$ for some $v$. 
\begin{itemize}
    \item If it does not, add $(x, \ket{-})$ to the database, where by this we formally mean:
    $D \mapsto \sum_{v} (-1)^v \ket{D \cup (x,v)}$
    \item If it does, apply the unitary that removes $(x,\ket{-})$ from the database.
\end{itemize}
\end{itemize}

One way to understand this compressed simulation is that our database representation only keeps track of pairs $(x, \ket{-})$ (corresponding to the queried points), and it does not keep track of the other unqueried points, which in a fully explicit simulation would correspond to $\ket{+}$'s. One can think of the outputs at the unqueried points as being ``compressed'' in this succinct representation.

It is easy to see that the map above can be extended to a well-defined unitary. In the rest of this overview, we will take this to be our compressed phase oracle. For an oracle algorithm $A$, we will denote by $A^{\comp}$ the algorithm $A$ run with a compressed phase oracle.


\subsubsection{The structure of strategies that produce valid equations}
Let $P$ be an efficient prover, which takes as input a choice $k$ of claw-free function pair, and outputs an equation $(z,d,y)$. Suppose $P$ outputs a valid equation, i.e. $ z= d \cdot (x_0^y \oplus x_1^y) \oplus H(x_0^y) \oplus H(x_1^y)$, with high probability. Suppose we run a compressed phase oracle simulation $P^{\comp}$. We argue that the state of $P^{\comp}$ right before measurement, conditioned on output $(z,d,y)$, must be such that:
\begin{itemize}
    \item Almost all the weight is on databases containing \emph{exactly one} of the two pre-images of $y$.
    \item Up to a phase, the weights on databases containing $x_0^y$ and $x_1^y$ are approximately symmetrical.
\end{itemize}
In this subsection, we provide some intuition as to why this is true. 

In general, we can write the state of $P^{\comp}$ right before measurement of the output as:
$$ \sum_{z,d,y,w,D} \alpha_{z,d,y,w,D} \ket{z}\ket{d}\ket{y}\ket{w}\ket{D} \,,$$
for some $\alpha_{z,d,y,w,D}$, where the $z,d,y$ registers correspond to the output equation, the $w$ register is a work register (which includes the query registers), and the $D$ register is the database register. Here $D$ is a database of pairs $(x,v)$.

This expression can be simplified in a couple of ways. First, up to negligible weight, no database can contain \emph{both} pre-images, since otherwise this would yield an efficient algorithm to extract a claw. Hence, we can write the state as a superposition over $z,d,y$ and over databases that either: do not contain any pre-image $y$, they contain only $x_0^y$, or they contain only $x_1^y$:
\begin{align}
    &\sum_{\substack{z,d,y,w \\ D \niton x_0^y, x_1^y}} \alpha_{z,d,y,w,D,0} \ket{z}\ket{d}\ket{y}\ket{w} \sum_{v_0} (-1)^{v_0} \ket{D \cup (x_0^y, v_0)} \nonumber \\
    + &\sum_{\substack{z,d,y,w \\ D \niton x_0^y, x_1^y}} \alpha_{z,d,y,w,D,1} \ket{z}\ket{d}\ket{y}\ket{w} \sum_{v_1} (-1)^{v_1} \ket{D \cup (x_1^y, v_1)}  \nonumber \\
    + &\sum_{\substack{z,d,y,w \\ D \niton x_0^y, x_1^y}} \beta_{z,d,y,w,D} \ket{z}\ket{d}\ket{y}\ket{w}  \ket{D} \label{eq: 9}
\end{align} 
for some $\alpha_{z,d,y,w,D,b}, \beta_{z,d,y,w,D}$. (Recall that the reason for the presence of the phases $(-1)^{v_0}$ and $(-1)^{v_1}$ is that, by definition of the compressed oracle simulation, the output at each queried point in the database is in a $\ket{-}$ state).

Second, we expect intuitively that $P$ should not be able to produce valid equations if it does not query any pre-image at all. Thus, we expect that if $P$ produces a valid equation with high probability, then there should be only a small weight on the third branch in expression \eqref{eq: 9}, i.e. the $\beta_{z,d,y,w,D}$ coefficients should be small. In our security proof, we formalize this intuition, and we show that any weight on the third branch contributes precisely $1/2$ to the probability of producing a valid equation, i.e. any weight on the third branch amounts to guessing an equation uniformly at random. Note that it is not a priori clear that this is true, and a more delicate analysis is required to establish that, when calculating the probability of outputting a valid equation, there is no interference between the branches that do not contain any pre-image, and the ones that contain one. We refer the reader to the full proof in Section \ref{sec: unexp scheme} for more details.

So, suppose now for simplicity that $P$ produces a valid equation with probability $1$, then, from what we have discussed, up to negligible weight, the state just before measurement of the output is in a superposition of branches that contain exactly one of the two pre-images:
\begin{align}
    &\sum_{\substack{z,d,y,w \\ D \niton x_0^y, x_1^y}} \alpha_{z,d,y,w,D,0} \ket{z}\ket{d}\ket{y}\ket{w} \sum_{v_0 \in \{0,1\}} (-1)^{v_0} \ket{D \cup (x_0^y, v_0)} \nonumber \\
    + &\sum_{\substack{z,d,y,w \\ D \niton x_0^y, x_1^y}} \alpha_{z,d,y,w,D,1} \ket{z}\ket{d}\ket{y}\ket{w} \sum_{v_1 \in \{0,1\}} (-1)^{v_1} \ket{D \cup (x_1^y, v_1)}  \label{eq: 88}
\end{align} 

Finally, we wish to argue that the coefficients on the two branches are uniform, i.e. $|\alpha_{z,d,y,w,D,0}| = |\alpha_{z,d,y,w,D,1}|$ for all $z,d,y,w,D$. For this, we again appeal to the fact that $P$ produces a valid equation with probability $1$.

What does this probability correspond to in terms of expression \eqref{eq: 88}? In order to calculate this probability, we need to first \emph{decompress} the database at both pre-images. What we mean by decompressing at $x$ is the following:
\begin{itemize}
    \item If the database already contains $x$, then do nothing.
    \item If the database does not contain $x$, then add $(x,\ket{+})$ to it.
\end{itemize}

The reason why this makes sense is the following. Recall that points that are not present in the compressed database correspond to $\ket{+}$'s in the fully explicit ``uncompressed'' database. Since the condition for a valid equation depends on the value at both pre-images, we need to keep track of these values, and uncompress at the two pre-images in order to talk about the probability of a valid equation.

The state after decompressing at $x_0^y$ and $x_1^y$ is:
\begin{align}
    &\sum_{\substack{z,d,y,w \\ D \niton x_0^y, x_1^y}} \alpha_{z,d,y,w,D,0} \ket{z}\ket{d}\ket{y}\ket{w} \sum_{v_0, v_1} (-1)^{v_0} \ket{D \cup (x_0^y, v_0) \cup (x_1^y, v_1)} \nonumber \\
    + &\sum_{\substack{z,d,y,w \\ D \niton x_0^y, x_1^y}} \alpha_{z,d,y,w,D,1} \ket{z}\ket{d}\ket{y}\ket{w} \sum_{v_0,v_1} (-1)^{v_1} \ket{D \cup (x_0^y, v_0) \cup (x_1^y, v_1)}  \label{eq: 8}
\end{align} 

Now, by the equivalence of the regular and compressed oracle simulations, we have that the probability that $P$ outputs a valid equation is equal to:
$$ \Pr[z = d \cdot (x_0^y \oplus x_1^y) \oplus v_0 \oplus v_1]$$

In order for this probability to be close to $1$, it must be that, for any $z,d,y,w, D$, the amplitudes of the two branches $\alpha_{z,d,y,w, D, 0}$ and $\alpha_{z,d,y,w,D,1}$ interfere precisely constructively when $z = d \cdot (x_0^y \oplus x_1^y) \oplus v_0 \oplus v_1$ and interfere destructively otherwise. Hence, they have to be equal up to an appropriate phase.

\subsubsection{Extracting a claw}

We now provide some intuition about how the structure that we derived in the previous subsection can be leveraged to extract a claw.

Throughout this subsection, let $P$ be such that, on input a choice of trapdoor claw-free function pair $k$, it simultaneously produces valid equations and proofs which are accepted with high probability, i.e. it outputs $(z,d,y,w)$ such that $\mathsf{Verify}(k,z,d,y,w)= 1$ with high probability. For concreteness we take this probability to be $9/10$ in this overview.

The first key observation is that, because by definition $\mathsf{Verify}$ never accepts an invalid equation (or except with negligible probability), then it must be the case that whenever $\mathsf{Verify}$ accepts, it \emph{must} itself have queried at a superposition of both pre-images in the following more precise sense. 

For a choice of claw-free functions $k$, denote by $\mathsf{Verify}_k$ the algorithm $\mathsf{Verify}$ where we fix the choice of claw-free functions to be $k$. Let $A$ be \emph{any} efficient oracle algorithm that, on input a choice of claw-free functions, outputs tuples $(z,d,y,w)$. Suppose we run a compressed oracle simulation of $A$, followed by $\mathsf{Verify}$, i.e. we run $(\mathsf{Verify}_k \circ A(k))^{\comp}$. Let the state right before measurement of $\mathsf{Verify}$'s output be:
$$ \alpha \ket{0} \ket{\phi_0} + \beta \ket{1} \ket{\phi_1} \,, $$
where the first qubit is the output qubit of $\mathsf{Verify}$, and $\ket{\phi_0}$, $\ket{\phi_1}$ are some states on the remaining registers (including the database register). Then, except with negligible probability over the choice of $k$, the following holds: if $\beta$ is non-negligible, then $\ket{\phi_1}$ has weight only on databases containing either $x_0^y$ or $x_1^y$, and moreover, the weights for each pre-image are approximately equal. Such structure on the action of $\mathsf{Verify}$ follows from a similar argument as in the previous subsection, combined with the fact that $\mathsf{Verify}$ never accepts an invalid equation.

Even with this observation in hand, it is not a priori clear how we can extract a claw: ideally, we would like to say that, because of the observation above, if we were to run $(\mathsf{Verify}_k \circ P(k))^{\comp}$ \emph{twice} and measure the database register, there would be a noticeable chance of obtaining distinct pre-images of some $y$. However, the issue is that if we run the computation twice, nothing guarantees that we will obtain the same $y$ both times. In fact, if one thinks about the honest strategy for producing valid equations, each $y$ has only an exponentially small probability of being the outcome. 

To overcome this issue, the key observation is that, if $P$ is successful with high enough probability, then there is a way to extract a claw with noticeable probability in a single run! The extraction algorithm is the following:
\begin{itemize}
    \item[(i)] Run $P^{\comp}(k)$, and measure the output registers to obtain $z,d,y$. Moreover, check if the database register at this point contains a pre-image of $y$. If so, measure it. Denote this by $x_b^y$.
    \item[(ii)] Run $\mathsf{Verify}_k^{\comp}$ on the leftover state from the previous step, and measure the output register. Conditioned on ``accept'', measure the database register. If the database contains $x_{\bar{b}}^y$, output the claw $(x_0^y, x_1^y)$.
\end{itemize}

The idea behind this algorithm is that, thanks to the first observation, the state conditioned on obtaining ``accept'' in step (ii) has weight only on databases containing either $x_0^y$ or $x_1^y$, and moreover, the weights for each pre-image are approximately equal. This implies that, conditioned on observing ``accept'' in step (ii), the final measurement of the algorithm is guaranteed to produce one of the two pre-images approximately \emph{uniformly at random}. 

Now, notice that the algorithm already has a constant probability of obtaining one of the two pre-images in step (i) (assuming $P$ succeeds with probability $\frac{9}{10}$). To see this, notice that if this weren't the case, then $P$ would be producing valid equations at most with probability close to $\frac12$ (since this is the probability of producing a valid equation when the database register does not contain any of the two pre-images), and therefore the probability that $\mathsf{Verify}$ accepts $P$'s output would also be at most close to $\frac12$ (instead of being $\frac{9}{10}$). Thus, what is left to show is that the probability of obtaining ``accept'' in step (ii) \emph{conditioned on finding a pre-image in step} (i) is still noticeable. 

It is not a priori clear that this is the case. However, what we show is that if, to begin with, $\mathsf{Verify}$ accepts $P$'s output with high enough probability (in fact any probability non-negligibly higher than $1/2$), then this extraction strategy works.

We remark that for our actual construction, we will want that, for any efficient $P$, the probability of simultaneously producing a valid equation \emph{and} a proof is \emph{negligible} (not just smaller than $\frac12 +$ negligible). Thus, our encryption scheme will be a parallel repetition of the single-shot encryption scheme described in this overview, i.e. the encryption of a single bit will consist of many single-shot encryptions of that bit.

\section{Preliminaries}
\subsection{Notation}
We use the acronyms PPT and QPT for \emph{probabilistic polynomial time} and \emph{quantum polynomial time} respectively.

For a classical probabilistic algorithm $\mathcal{A}$, we write $\mathcal{A}(x ; r)$ to denote running $\mathcal{A}$ on input $x$, with input randomness $r$.

For a finite set $S$, we use $x \gets S$ to denote uniform sampling of $x$ from the set $S$. 
We denote $[n] = \{1, 2, \cdots, n\}$. We denote by $\mathsf{Bool}(n)$ the set of functions from $n$ bits to $1$.

\subsection{Deniable Encryption}
We recall the classical notion of sender-deniable encryption. 
\begin{definition}
\label{def: classical deniability}
A public-key encryption scheme $(\mathsf{Gen}, \mathsf{Enc}, \mathsf{Dec})$ is said to be \emph{deniable} if there exists a PPT algorithm $\mathsf{Fake}$ such that, for any messages $m_0, m_1$:
$$(\pk, \mathsf{Enc}(\pk, m_1;r), m_1, r)  \approx_c (\pk, \mathsf{Enc}(\pk, m_0;r), m_1, r') $$
where $r$ is uniformly random, $\pk$ is sampled according to $\mathsf{Gen}$, and $r' \gets \mathsf{Fake}(\pk, c, m_0, m_1, r)$.
\end{definition}

\subsection{Noisy Trapdoor Claw-Free Functions}

In this section we introduce the notion of noisy trapdoor claw-free functions (NTCFs). This section is taken almost verbatim from \cite{brakerski2020simpler}. Let $\X, \Y$ be finite sets and $\K$ a set of keys. For each $k \in K$ there should exist two (efficiently computable) injective functions $f_{k,0}, f_{k,1}$ that map $\X$ to $\Y$, together with a trapdoor $t_k$ that allows efficient inversion from $(b,y)\in \{0,1\} \times \Y$ to $f_{k,b}^{-1}(y)\in\X \cup \set{\perp}$. For security, we require that for a randomly chosen key $k$, no polynomial time adversary can efficiently compute $x_0, x_1 \in \X$ such that $f_{k,0}(x_0) = f_{k,1}(x_1)$ (such a pair $(x_0, x_1)$ is called a \emph{claw}). 

Unfortunately, we do not know how to construct such `clean' trapdoor claw-free functions. Hence, as in the previous works \cite{brakerski2018cryptographic, mahadev2018classical, brakerski2020simpler}, we will use `noisy' version of the above notion. For each $k \in \K$, there exist two functions $f_{k,0}, f_{k,1}$ that map $\X$ to a distribution over $\Y$.


\begin{definition}[NTCF family]\label{def: ntcfs}
	Let $\lambda$ be a security parameter. Let $\X$ and $\Y$ be finite sets. Let $\mathcal{D}_\mathcal{Y}$ be a distribution over $\mathcal{Y}$.
	 Let $\mathcal{K}_{\mathcal{F}}$ be a finite set of keys. A family of functions 
	$$\mathcal{F} \,=\, \big\{f_{k,b} : \X\rightarrow \mathcal{D}_{\Y} \big\}_{k\in \mathcal{K}_{\mathcal{F}},b\in\{0,1\}}$$
	is called a \emph{noisy trapdoor claw free (NTCF) family} if the following conditions hold:

	\begin{enumerate}
		\item{\textbf{Efficient Function Generation.}} There exists an efficient probabilistic algorithm $\genf$ which generates a key $k\in \mathcal{K}_{\mathcal{F}}$ together with a trapdoor $t_k$: 
		$$(k,t_k) \leftarrow \genf(1^\lambda)\;.$$

		\item{\textbf{Trapdoor Injective Pair.}}  
		\begin{enumerate}
			\item \textit{Trapdoor}: There exists an efficient deterministic algorithm $\invf$ such that with overwhelming probability over the choice of $(k,t_k) \gets \genf(1^{\lambda})$, the following holds: \\
			$$ \text{for all } b\in \{0,1\}, x\in \X \text{ and } y\in \supp(f_{k,b}(x)), \invf(t_k,b,y) = x. $$
			\item \textit{Injective pair}: For all keys $k\in \mathcal{K}_{\mathcal{F}}$, there exists a perfect matching $\R_k \subseteq \X \times \X$ such that $f_{k,0}(x_0) = f_{k,1}(x_1)$ if and only if $(x_0,x_1)\in \R_k$. 
		\end{enumerate}

		\item{\textbf{Efficient Range Superposition.}}
		For all keys $k\in \mathcal{K}_{\mathcal{F}}$ and $b\in \{0,1\}$ there exists a function $f'_{k,b}:\X\to \mathcal{D}_{\Y}$ such that the following hold.
		\begin{enumerate} 
			\item For all $(x_0,x_1)\in \mathcal{R}_k$ and $y\in \supp(f'_{k,b}(x_b))$, $\invf(t_k,b,y) = x_b$ and $\invf(t_k,b\oplus 1,y) = x_{b\oplus 1}$. 
			\item There exists an efficient deterministic procedure $\chkf$ that, on input $k$, $b\in \{0,1\}$, $x\in \X$ and $y\in \Y$, returns $1$ if $y\in \supp(f'_{k,b}(x))$ and $0$ otherwise. Note that $\chkf$ is not provided the trapdoor $t_k$. 
			\item For every $k$ and $b\in\{0,1\}$, 
			$$ \mathbb{E}_{x\leftarrow \X} \big[\,H^2(f_{k,b}(x),\,f'_{k,b}(x))\,\big] \,\leq\, \mu(\lambda) \;.$$
			 for some negligible function $\mu$. Here $H^2$ is the Hellinger distance. Moreover, there exists an efficient procedure  $\sampf$ that on input $k$ and $b\in\{0,1\}$ prepares the state
			\begin{equation}
			    \frac{1}{\sqrt{|\X|}}\sum_{x\in \X,y\in \Y}\sqrt{(f'_{k,b}(x))(y)}\ket{x}\ket{y}\;.
			\end{equation}
		\end{enumerate}

		\item{\textbf{Claw-Free Property.}}
		For any PPT adversary $\A$, there exists a negligible function $\negl(\cdot)$ such that the following holds: 
		\begin{align*}
			\Pr[(x_0, x_1) \in \R_k  : (k, t_k) \gets \genf(1^{\lambda}), (x_0, x_1) \gets \A(k)] \leq \negl(\lambda)
		\end{align*}

	\end{enumerate}

\end{definition} 

The next two definitions are taken from \cite{mahadev2018classical}.
\begin{definition}[Trapdoor Injective Function Family]
\label{def: trapdoor injective}
Let $\lambda$ be a security parameter. Let $\X$ and $\Y$ be finite sets. Let $\mathcal{D}_\mathcal{Y}$ be a distribution over $\mathcal{Y}$. Let $\mathcal{K}_{\mathcal{G}}$ be a finite set of keys. A family of functions 
$$\mathcal{G} \,=\, \big\{g_{k,b} : \X\rightarrow \mathcal{D}_{\Y} \big\}_{b\in\{0,1\},k\in \mathcal{K}_{\mathcal{G}}}$$
is called a \textbf{trapdoor injective family} if the following conditions hold:

\begin{enumerate}
\item{\textbf{Efficient Function Generation.}} There exists an efficient probabilistic algorithm $\geng$ which generates a key $k\in \mathcal{K}_{\mathcal{G}}$ together with a trapdoor $t_k$: 
$$(k,t_k) \leftarrow \textrm{GEN}_{\mathcal{G}}(1^\lambda)\;.$$

\item{\textbf{Disjoint Trapdoor Injective Pair.}} For all keys $k\in \mathcal{K}_{\mathcal{G}}$, for all $b, b'\in\{0,1\}$ and $x,x' \in \X$, if $(b,x)\neq (b',x')$, $\supp(g_{k,b}(x))\cap \supp(g_{k,b'}(x')) = \emptyset$. Moreover, there exists an efficient deterministic algorithm $\invf$ such that for all $b\in \{0,1\}$,  $x\in \X$ and $y\in \supp(g_{k,b}(x))$, $\invg (t_k,y) = (b,x)$.  

\item{\textbf{Efficient Range Superposition.}}
For all keys $k\in \mathcal{K}_{\mathcal{G}}$ and $b\in \{0,1\}$ 
\begin{enumerate} 
 
\item There exists an efficient deterministic procedure $\chkg$ that, on input $k$, $b\in \{0,1\}$, $x\in \X$ and $y\in \Y$, outputs $1$ if  $y\in \supp(g_{k,b}(x))$ and $0$ otherwise. Note that $\chkg$ is not provided the trapdoor $t_k$. 
\item There exists an efficient procedure  $\sampg$ that on input $k$ and $b\in\{0,1\}$ returns the state
\begin{equation}
    \frac{1}{\sqrt{|\X|}}\sum_{x\in \X,y\in \Y}\sqrt{(g_{k,b}(x))(y)}\ket{x}\ket{y}\;.
\end{equation}

\end{enumerate}
\end{enumerate}

\end{definition}

\begin{definition}[Injective Invariance]
A noisy trapdoor claw-free family $\mathcal{F}$ is \textbf{injective invariant} if there exists a trapdoor injective family $\mathcal{G}$ such that:
\begin{enumerate}
\item The algorithms $\chkf$ and $\sampf$ are the same as the algorithms $\chkg$ and $\sampg$.
\item For all quantum polynomial-time procedures $\mathcal{A}$, there exists a negligible function $\mu(\cdot)$ such that
\begin{equation}
\Big|\Pr_{(k,t_k)\leftarrow \textrm{GEN}_{\mathcal{F}}(1^{\lambda})}[\mathcal{A}(k) = 0] - \Pr_{(k,t_k)\leftarrow \textrm{GEN}_{\mathcal{G}}(1^{\lambda})}[\mathcal{A}(k) = 0]\Big|\leq \mu(\lambda)
\end{equation}
\end{enumerate}
\end{definition}

\begin{lem}[\cite{brakerski2018cryptographic}, \cite{mahadev2018classical}]
Assuming the quantum hardness of LWE, there exists an injective invariant NTCF family.
\end{lem}

\section{Unexplainable Encryption}
In this section, we formally introduce the notions of \emph{unexplainable} encryption, and \emph{perfectly unexplainable} encryption. In \ref{sec: before-the-fact}, we discuss the notion of coercion before-the-fact in more detail, and its relationship to perfect unexplainability.  In \ref{sec: den vs unexp}, we discuss in detail the relationship between unexplainable encryption and the standard definition of deniable encryption.
\subsection{Definition}
\label{sec: unexp def}
The following definition involves algorithms with the following syntax:
\begin{itemize}
\item $\mathsf{Explain}$ takes as input a public key $\mathsf{pk}$, a message $m$, and outputs a \emph{classical} ciphertext $c$ and a \emph{quantum} witness $w$; 
\item $\mathsf{Verify}$ takes as input a public key $\mathsf{pk}$, a classical ciphertext $c$, a message $m$, and a quantum witness $w$, and outputs a bit; \item $\mathsf{FakeExplain}$ takes as input a public key $\pk$, a pair of messages $m, m'$, and outputs a ciphertext $c$, and two (possibly entangled) witnesses $w,w'$.
\end{itemize}
In the rest of the section, we use the notation $\Pr\limits_{\pk}$ to mean that the probability is over sampling a public key $\mathsf{pk}$ from the generation algorithm $\mathsf{Gen}(1^{\lambda})$ (where the security parameter $\lambda$ is omitted from the notation).

\begin{definition}[Unexplainable Encryption]
\label{def: unexp}
We say that a public-key encryption scheme $(\mathsf{Gen}, \mathsf{Enc}, \mathsf{Dec})$ is  \emph{unexplainable} if the following holds. Let $\mathsf{Verify}$ and $\mathsf{Explain}$ be any (non-uniform) QPT algorithms (with the syntax described above) such that, for all $\pk, c,m,w$, $\mathsf{Verify}(\mathsf{pk}, c, m, w) = 0$, except with exponentially small probability, if $c$ is not in the support of the distribution of $\mathsf{Enc}(\mathsf{pk},m)$.

Suppose $\mathsf{Verify}, \mathsf{Explain}$ satisfy the following \emph{completeness} condition: there exists a non-negligible function $\gamma$, such that, for any $m$, for all $\lambda$, 
$$\Pr_{\pk}[\mathsf{Verify}(\pk, c, m, w) = 1: (c,w) \gets \mathsf{Explain}(\mathsf{pk},m)]  \geq \gamma(\lambda) \,.$$
Then, there exist a polynomial-time algorithm $\mathsf{FakeExplain}$ and a non-negligible function $\negl$, such that, for any distinct messages $m,m'$, for all $\lambda$:
\begin{align}
\label{eq: 76}
   \Big| &\Pr_{\pk}[\mathsf{Verify}\left(\pk, c, m, w\right) =1: c,w \gets \mathsf{Explain}(\pk,m)] \nonumber \\
   -  &\Pr_{\pk}[\mathsf{Verify}\left(\pk, c, m', w' \right) =1:  c,w \gets \mathsf{Explain}(\pk,m), \,\, w' \gets \mathsf{FakeExplain}(\pk,m,m',c,w)] \Big| = \negl(\lambda) 
\end{align}
\end{definition}

We believe that this definition naturally captures the ideal of privacy desired in a deniable setting. Moreover, since this definition does not explicitly refer to input randomness, it applies to the classical and quantum settings alike (up to making the corresponding algorithms classical or quantum).

We remark that in Definition \ref{def: unexp} we restricted to considering $\mathsf{Verify}$ such that $\mathsf{Verify}(\mathsf{pk}, c, m, w) = 0$, with probability $1$, whenever $c$ is not in the support of the distribution of $\mathsf{Enc}(\mathsf{pk},m)$. One could relax this requirement slightly (and thus obtain a stronger definition) by asking that $ \Pr[\mathsf{Verify}(\mathsf{pk}, c, m, w) = 0]$ be exponentially small (in the security parameter) whenever $c$ is not in the support of the distribution of $\mathsf{Enc}(\mathsf{pk},m)$. We will show that our scheme still satisfies this definition.

The following is a special case of unexplainable encryption. In words, an encryption scheme is \emph{perfectly} unexplainable if there does not exist a pair of efficient algorithms $\mathsf{Verify}$ and $\mathsf{Explain}$ (where $\mathsf{Verify}(\mathsf{pk}, c, m, w) = 0$ if $c$ is not in the range of $\mathsf{Enc}(\mathsf{pk},m; \cdot)$), for which the \emph{completeness} condition holds.
\begin{definition}[Perfectly Unexplainable Encryption]
\label{def: perfect unexp}
A public key encryption scheme $(\mathsf{Gen}, \mathsf{Enc}, \mathsf{Dec})$ is said to be \emph{perfectly unexplainable} if the following holds. Let $\mathsf{Explain}$ and $\mathsf{Verify}$ be any (non-uniform) QPT algorithms (with the syntax described above)such that, for all $\pk, c,m,w$, $\mathsf{Verify}(\mathsf{pk}, c, m, w) = 0$, except with exponentially small probability, if $c$ is not in the support of the distribution of $\mathsf{Enc}(\mathsf{pk},m)$. Then, for any $m$, there exists a negligible function $\negl$, such that for any $\lambda$, 
\begin{equation}
    \Pr_{\pk}[\mathsf{Verify}(\pk, c, m, w) = 1: (c,w) \gets \mathsf{Explain}(\mathsf{pk},m)]  = \negl(\lambda) \,. \label{eq: perfect unexp}
\end{equation}
\end{definition}

At first glance, perfect unexplainability seems unattainable. In fact, it is clear that a classical encryption scheme cannot be perfectly unexplainable: one can always take $\mathsf{Explain}$ to be the algorithm that encrypts honestly and outputs the ciphertext together with the randomness used (i.e. the witness is the randomness), and take $\mathsf{Verify}$ to be the algorithm that runs encryption forward and checks consistency. In contrast, the quantum encryption scheme that we will describe in Section \ref{sec: unexp scheme} will be perfectly unexplainable.

We emphasize that perfect unexplainability is the strongest sense in which a scheme can be unexplainable: condition \eqref{eq: 76} in Definition \ref{def: unexp} can be replaced by something as demanding as we desire, but this doesn't matter because, for a perfectly unexplainable encryption scheme, completeness is never satisfied, and thus any replacement of condition \eqref{eq: 76} would be vacuously satisfied. We point out that one could even relax the quantifier over $m$ in the completeness condition to ``$\exists m$''. Our scheme from Section \ref{sec: unexp scheme} would still satisfy this stronger definition.

\begin{remark}
It is crucial that the algorithm $\mathsf{Explain}$ is required to output a \emph{classical} ciphertext, along with a possibly quantum witness (i.e.\ equivalently, it outputs a classical-quantum state over the ciphertext-witness registers). The definition would be unachievable if $\mathsf{Explain}$ were allowed to output a \emph{quantum} ciphertext that is possibly entangled with the witness: in this case, one could simply consider an $\mathsf{Explain}$ algorithm that runs the honest encryption algorithm, but without measuring the register $\mathsf{C}$ containing the ciphertext. It then provides its work register $\mathsf{W}$ (which is possibly entangled with $\mathsf{C}$) as the witness. One can then define $\mathsf{Verify}$ to simply run the honest encryption unitary in reverse, and check that this recovers the initial (say, all zeros) state. This restriction on $\mathsf{Explain}$ is not only necessary, but arguably also natural. It captures the scenario where $\mathsf{Explain}$ tries to simultaneously submit a valid ciphertext (e.g.\ to a government in an online election) through a classical communication channel, as well as a witness (which could be a quantum state) that this ciphertext is correct (e.g.\ to an attacker who bought the vote).
\end{remark}

\subsection{Coercion before-the-fact}
\label{sec: before-the-fact}

While protecting against coercion before-the-fact is very desirable in practice, to the best of our knowledge, a corresponding notion has not been previously formalized
(most likely due to the fact that protecting against coercion before-the-fact is impossible to achieve classically in a cryptographic sense). Here, we propose a formal definition, and we observe that this is essentially equivalent to the notion of \emph{perfect unexplainability} from Definition \ref{def: perfect unexp}. 

In a coercion before-the-fact scenario, an attacker, who has in mind a message $m$, wishes to prescribe to the sender \emph{how} she should encrypt later in a way that:
\begin{itemize}
    \item The resulting ciphertext decrypts to $m$ with overwhelming probability.
    \item There is an efficient procedure for the attacker to verify that the sender's ciphertext (which the attacker obtains by intercepting) will decrypt to $m$ with overwhelming probability.
\end{itemize}
First, notice that there is no hope of protecting against coercion before-the-fact in a model where the attacker can approach the sender before she sends her ciphertext, and knows \emph{all} of the information that will be available to the sender at the time of encryption (e.g. the public key). In fact, in such a model, the attacker can simply generate a genuine encryption $c$ of the desired message $m$, and prescribe that the sender's ciphertext later be exactly $c$. So, instead we consider the scenario where the public key is \emph{not known} to the attacker at the time when he approaches the sender. Such a model captures an online election where citizens encrypt their votes before sending them to the government using a public key encryption scheme, and the public key is announced publicly only on election day. The attacker is allowed to approach the sender any time \emph{before} election day (i.e. any time before the public key is announced). More generally, one can consider a model where some additional information (not necessarily the public key) is revealed to the sender just before she encrypts. For simplicity, we restrict ourselves to formalizing the setting where the public key is the only information that is not known to the attacker at the time of coercion.

We first formally define the notion of a coercion before-the-fact \emph{attack}. An encryption scheme then protects against coercion before-the-fact if no such attack exists. 

\begin{definition}[Coercion before-the-fact attack]
\label{def: before-attack}
Let $(\mathsf{Gen}, \mathsf{Enc}, \mathsf{Dec})$ be a public-key encryption scheme with classical ciphertexts (where $\mathsf{Dec}$ is a classical deterministic algorithm). A coercion before-the-fact attack is a pair $(\mathsf{Enc}', \mathsf{Verify})$ where:
\begin{itemize}
    \item $\mathsf{Enc}'( \pk, m) \rightarrow c\,\,$ is a non-uniform QPT algorithm
    \item $\mathsf{Verify}(c, m) \rightarrow  \mathsf{accept}/\mathsf{reject}\,\,$ is a non-uniform QPT algorithm
\end{itemize}
They satisfy:
\begin{itemize}
    \item (Completeness of verification) For any $m, \lambda$, $$\Pr[\mathsf{Verify}(\mathsf{Enc}'(\pk, m), m) = \mathsf{accept}: (\sk,\pk) \gets \mathsf{Gen}(1^{\lambda})] = 1 \,.$$
    \item (Soundness of verification) There exists $C>0$, such that the following holds for all $c,m, \lambda$: $$ \Pr[\mathsf{Dec}(\sk, c) \neq m \,\land \, \mathsf{Verify}(\pk,c, m) = \mathsf{accept} :(\sk,\pk) \gets \mathsf{Gen}(1^{\lambda})] \leq 2^{-\lambda^C} \,. \footnote{The only reason why the upper bound in the soundness condition is exponentially small rather than being zero is that, for a scheme with imperfect decryption, the classical attack that prescribes the input randomness and checks consistency of the intercepted ciphertext with the prescribed randomness has completeness $1$ but exponentially small soundness, rather than $0$.}$$ 
\end{itemize}
\end{definition}

\begin{definition}
We say that an encryption scheme \emph{protects against coercion before-the-fact} if no coercion before-the-fact attacks exists.
\end{definition}
It is straightforward to see that a perfectly unexplainable encryption scheme protects against coercion before-the-fact. This is because a coercion before-the-fact attack provides a way to encrypt in a way that can be later verified.
\begin{theorem}
A perfectly unexplainable encryption scheme protects against coercion before-the-fact.
\end{theorem}
\begin{proof}
Let $(\mathsf{Gen}, \mathsf{Enc}, \mathsf{Dec})$ be a perfectly unexplainable encryption scheme (satisfying the non-uniform version of Definition \ref{def: perfect unexp}). Suppose for a contradiction that a coercion before-the-fact attack $(\mathsf{Enc}', \mathsf{Verify})$ existed.

Define $\mathsf{Verify}'$ to be the non-uniform algorithm that on input $\pk, c, m, w$, runs $\mathsf{Verify}(\pk,c,m)$ and ignores $w$. Let $\mathsf{Explain}$ be the non-uniform algorithm that, on input $\pk, m$, runs $c \gets \mathsf{Enc}'(\pk,m)$, and outputs $c, w^*$, for some fixed $w^*$. By the completeness and soundness of the coercion before-the-fact attack, it follows that the pair $(\mathsf{Explain}, \mathsf{Verify}')$ contradicts perfect unexplainability.
\end{proof}

Definition \ref{def: before-attack} only considers attacks where $\mathsf{Enc}’$ and $\mathsf{Verify}$ are non-uniform with \emph{classical} advice. More generally, one could also consider attacks where $\mathsf{Enc}’$ and $\mathsf{Verify}$ have \emph{quantum} advice. In particular, this advice could be in the form of an entangled state over two registers corresponding to the advice of $\mathsf{Enc}’$ and $\mathsf{Verify}$ respectively. This captures the scenario where an attacker coerces the sender before-the-fact by giving to the sender half of some entangled state, and prescribing what the encryption operation should be, i.e. $\mathsf{Enc}’$. The verification of the sender's intercepted ciphertext then makes use of the other half of the entangled state. 

It is not difficult to see that a slightly more general version of the definition of perfect unexplainability, where $\mathsf{Explain}$ and $\mathsf{Verify}$ are allowed to be non-uniform with entangled quantum advice, implies this more general notion of protection against coercion before-the-fact. Such a notion of perfect unexplainability holds assuming the quantum hardness of LWE against QPT algorithm with polynomial quantum advice. Our security proof in Section \ref{sec: unexp scheme} goes through unchanged (since our extraction algorithm does not involve any rewinding).

\subsection{The relationship between deniable and unexplainable encryption}
\label{sec: den vs unexp}
While the definition of unexplainable encryption that we introduced in Section \ref{sec: unexp def} seems related to deniability, its exact relationship to it is not entirely clear. The reason why the two definitions are difficult to formally compare is that the classical definition of deniability is, in some sense, centered around the input randomness, while the latter does not appear at all in unexplainability. In this section, we shed light on the relationship between unexplainability and deniability, by introducing a variation of unexplainability in the \emph{classical} setting, and showing that it is equivalent to deniability. Essentially, we show that deniability can be equivalently phrased from the point of view of unexplainability.  

The syntax of the algorithms is slightly different from that of Definition \ref{def: unexp}: $\mathsf{Explain}$ takes as input a public key $\mathsf{pk}$, a message $m$, randomness $r$, and outputs some witness $w$; 
$\mathsf{Verify}$ takes as input a public key $\mathsf{pk}$, a ciphertext $c$, a message $m$, a witness $w$, and outputs a bit; $\mathsf{FakeExplain}$ takes as input a public key $\pk$, a pair of messages $m, m'$, randomness $r$, and outputs a witness $w$. We emphasize that the following definition is entirely \emph{classical}.
\begin{definition}[Strongly Unexplainable Encryption]
\label{def: classical unexp}
A public key encryption scheme $(\mathsf{Gen}, \mathsf{Enc}, \mathsf{Dec})$ is said to be \emph{strongly unexplainable} if the following holds. For any PPT algorithm $\mathsf{Explain}$, there exists a PPT algorithm $\mathsf{FakeExplain}$ such that the following holds for any PPT algorithm $\mathsf{Verify}$ such that, for all $\mathsf{pk}, c, m, w$, $\mathsf{Verify}(\mathsf{pk}, c, m, w) = 0$ (with probability $1$) if $c$ is not in the range of $\mathsf{Enc}(\mathsf{pk},m; \cdot)$.

Suppose $\mathsf{Verify}, \mathsf{Explain}$ satisfy the following \emph{completeness} condition: there exists a non-negligible function $\gamma$, such that, for any $m$, for any $\lambda$, 
$$\Pr_{\pk,r}[\mathsf{Verify}(\pk, \mathsf{Enc}(\pk, m;r), m, w) = 1: w \gets \mathsf{Explain}(\mathsf{pk},m,r)]  = \gamma(\lambda) \,.$$
Then, for any distinct messages $m,m'$, for any $\lambda$, we have:
\begin{align}
\label{eq: 7}
   \Big| &\Pr_{\pk, r}[\mathsf{Verify}\left(\pk, \mathsf{Enc}(\pk, m;r), m, w\right) =1: w \gets \mathsf{Explain}(\pk,m,r)] \nonumber \\
   -  &\Pr_{\pk,r}[\mathsf{Verify}\left(\pk, \mathsf{Enc}(\pk, m';r), m, w\right) =1:  w \gets \mathsf{FakeExplain}(\pk,m,m',r)] \Big| = \negl(\lambda) 
\end{align}
\end{definition}

\begin{lem}
\label{lem: equivalence}
A classical encryption scheme is sender-deniable (as in Definition \ref{def: classical deniability}) if and only if it is strongly unexplainable (according to Definition \ref{def: classical unexp}).
\end{lem}

In light of this equivalence, one can view Definition \ref{def: unexp} from the previous section as a natural generalization of Definition \ref{def: classical unexp}, and hence of deniability, to the quantum setting, where no input randomness is a priori present. 

We remark that the order of quantification over algorithms is different in Definitions \ref{def: unexp} and \ref{def: classical unexp}:  in the former, we required that ``$\forall \,  \mathsf{Explain}, \mathsf{Verify}$ satisfying completeness, $\exists \, \mathsf{FakeExplain}$ such that \eqref{eq: 7} is satisfied.''; in the latter, we require that $\forall \, \mathsf{Explain}, \exists \, \mathsf{FakeExplain}$ such that $\forall \, \mathsf{Verify}$ (for which  $\mathsf{Verify}, \mathsf{Explain}$ satisfy completeness), \eqref{eq: 7} is satisfied. Although we find the former order of quantification to be more natural, the latter is what is necessary for the equivalence with the classical definition of deniability. We point out that the choice of order of quantification does not affect whether our scheme from Section \ref{sec: unexp scheme} satisfies the definition. This is because the order of quantification does not affect \emph{perfect} unexplainability.

\begin{proof}[Proof of Lemma \ref{lem: equivalence}]
``$\Rightarrow$'': Suppose $(\mathsf{Gen}, \mathsf{Enc}, \mathsf{Dec})$ is sender-deniable. Then there exists an algorithm $\mathsf{Fake}$ as in the definition of sender-deniability.

Let $\mathsf{Explain}$ be any QPT algorithm (with syntax as in the definition of unexplainability). We define the following algorithm $\mathsf{FakeExplain}$:
\begin{itemize}
    \item On input $(\mathsf{pk}, m, m',r)$, run $ r' \gets \mathsf{Fake}(\pk, m, m', r)$. 
    \item Run $w \gets \mathsf{Explain} (\pk, m',r')$. Output $w$.
\end{itemize}
Since the encryption scheme is sender-deniable, by definition, the following two distributions are computationally indistinguishable for any $m,m'$, where $\pk$ is sampled from $\mathsf{Gen}$, $r$ is uniformly random, and $r' \gets \mathsf{Fake}(\pk,m,m',r)$:
\begin{equation}
(\pk,\mathsf{Enc}(\pk,m;r), m, r) \approx_c (\pk,\mathsf{Enc}(\pk,m';r), m, r')  \label{eq: 3000}
\end{equation}
Consider then the following distinguishing algorithm $D$ for the two distributions:
\begin{itemize}
    \item On input $(\pk, c, m, r)$, run $w \gets \mathsf{Explain}(\pk,m,r)$.
    \item Then, run $\mathsf{Verify}(\pk,c,m, w)$, and output the outcome.
\end{itemize}
By \eqref{eq: 3000}, $$\big|\Pr[D(\pk,\mathsf{Enc}(\pk,m;r), m, r) = 1] - \Pr[D(\pk,\mathsf{Enc}(\pk,m';r), m, r') = 1] \big|= \negl(\lambda) \,.$$
Then, notice that the first term on the LHS is precisely equal to 
$$  \Pr_{\pk, r}[\mathsf{Verify}\left(\pk, \mathsf{Enc}(\pk, m;r), m, w\right) =1: w \gets \mathsf{Explain}(\pk,m,r)] \,,$$
while the second term, using the definition of $\mathsf{FakeExplain}$, is equal to
$$ \Pr_{\pk,r}[\mathsf{Verify}\left(\pk, \mathsf{Enc}(\pk, m';r), m, w\right) =1:  w \gets \mathsf{FakeExplain}(\pk,m,m',r)]\,.$$
Since this holds for any $m, m'$, this establishes precisely that the scheme is unexplainable.
\vspace{2mm}

``$\Leftarrow$'': Let $(\mathsf{Gen}, \mathsf{Enc}, \mathsf{Dec})$ be a strongly unexplainable encryption scheme. 
Now, consider the following algorithm $\mathsf{Explain}$:
\begin{itemize}
\item On input $(\pk, m, r)$ output $(c,w)$ where $c = \mathsf{Enc}(\pk, m; r)$ and $w = r$.
\end{itemize}
By definition of strong unexplainability, for any QPT $\mathsf{Explain}$, so in particular for $\mathsf{Explain}$ as defined above, there exists a QPT algorithm $\mathsf{FakeExplain}$ such that for any QPT $\mathsf{Verify}$ (such that the pair $\mathsf{Verify}, \mathsf{Explain}$ satisfies completeness), for any distinct messages $m,m'$, we have
\begin{align}
\label{eq: 38}
   \Big| &\Pr_{\pk, r}[\mathsf{Verify}\left(\pk, \mathsf{Enc}(\pk, m;r), m, w\right) =1: w \gets \mathsf{Explain}(\pk,m,r)] \nonumber \\
   -  &\Pr_{\pk,r}[\mathsf{Verify}\left(\pk, \mathsf{Enc}(\pk, m';r), m, w\right) =1:  w \gets \mathsf{FakeExplain}(\pk,m,m',r)] \Big| = \negl(\lambda)
\end{align}

Now, suppose for a contradiction that the encryption scheme is not deniable. Recall, that our particular choice of algorithm $\mathsf{Explain}$ outputs randomness. Then, since $\mathsf{FakeExplain}$ takes inputs of the form $(\pk,m,m',r)$ and also outputs some randomness $r'$, the latter has the correct syntax for a faking algorithm in a deniable encryption scheme. Since our scheme is not deniable by hypothesis, then there must exist a pair of distinct messages $m_*, m_*'$, and an efficient distinguisher $D$ for the two distributions $(\pk,\mathsf{Enc}(\pk,m_*;r), m_*, r)=1$ and $(\pk,\mathsf{Enc}(\pk,m_*';r), m_*, r')$ where $r' \gets \mathsf{Fake}(\pk, m_*, m_*', r)$. Without loss of generality, assume that, for all $\lambda$,
\begin{equation} 
\label{eq: 16}
\Pr_{\pk,r}[D(\pk,\mathsf{Enc}(\pk,m_*;r), m_*, r)=1] \geq  \Pr_{\pk, r}[D(\pk,\mathsf{Enc}(\pk,m_*';r), m_*, r')=1] + \delta(\lambda) \,.
\end{equation} 
where $\delta$ is a non-negligible function.


Now, consider the algorithm $\mathsf{Verify}$, defined as follows:
\begin{itemize}
    \item On input $\pk, c, m, w$, check that $\mathsf{Enc}(\pk, m; w) = c$. If this is not the case, output $0$. Otherwise, proceed to the next step.
    \item If $m \neq m_*$, output $1$, otherwise, run $D(\pk, c, m, w)$, and output the outcome.
\end{itemize}
First, we argue that $\mathsf{Verify}$ and $\mathsf{Explain}$ as defined above satisfy the completeness condition.

Notice that the probability on the LHS of \eqref{eq: 16} is precisely 
$$ \Pr_{\pk, r}[\mathsf{Verify}\left(\pk, \mathsf{Enc}(\pk, m;r), m, w\right) =1: w \gets \mathsf{Explain}(\pk,m,r)] \,,$$
and the probability on the RHS of \eqref{eq: 16} is precisely $$\Pr_{\pk,r}[\mathsf{Verify}\left(\pk, \mathsf{Enc}(\pk, m';r), m, w\right) =1:  w \gets \mathsf{FakeExplain}(\pk,m,m',r)] ]\,.$$
Hence, \eqref{eq: 16} contradicts \eqref{eq: 38}.
\end{proof}

\

\section{A Perfectly Unexplainable Encryption Scheme}
\label{sec: unexp scheme}
In this section, we describe a public-key encryption scheme, and prove that it is both CPA-secure and perfectly unexplainable (as in Definition \ref{def: perfect unexp}).

\subsection{Construction}

Let $\X,\Y,\K$ be finite sets. Let $\mathcal{F} \,=\, \big\{f_{k,b} : \X\rightarrow \mathcal{D}_{\Y} \big\}_{k\in \mathcal{K},b\in\{0,1\}}$ be a family of noisy trapdoor claw-free functions (which exists assuming the quantum hardness of LWE \cite{brakerski2018cryptographic}). Let $f'_{k,b}: \X \rightarrow \Y$ be functions satisfying the \emph{efficient range superposition property} of Definition \ref{def: ntcfs}. 

The scheme that we describe in this section is a parallel repeated version of the scheme described in the technical overview (Section \ref{sec: comp oracles tech overview}). Here, by parallel repetition we mean that the same plaintext $m$ is encrypted $L$ times (with the same public key), where $L$ is polynomial in the security parameter. Without parallel repetition, we are only able to show that the LHS of Equation \eqref{eq: perfect unexp} is upper bounded by $\frac12 + \negl$. With parallel repetition, we will be able to improve this to $\negl$.

\begin{construction} 
\label{cons: unexplainable scheme}
$\,$
Parameters: $L = \poly(\lambda)$.
\begin{itemize}
    \item $\mathsf{Gen}(1^{\lambda}) \rightarrow (\mathsf{pk}, \mathsf{sk})$:
    \begin{itemize}
        \item Run $(k,t_k) \gets \genf(1^\lambda)$. Output $(\mathsf{pk} , \mathsf{sk}) = (k, t_k)$.
    \end{itemize}
    \item $\mathsf{Enc}(m, \pk) \rightarrow c$:
    \begin{itemize}
        \item For $i \in [L]$, do the following:
     \begin{itemize}
        \item On input $m \in \{0,1\}$, and $\pk= k$, run $\sampf(k, \cdot)$ on a uniform superposition of $b$'s, to obtain the state
        $$ \frac{1}{\sqrt{|\mathcal{X}|}}\sum_{b\in\{0,1\} , x \in X} \sqrt{f_{k,b}'(x)(y)}\ket{b}\ket{x}\ket{y}\,,$$
        where we assume that $x$ and $y$ are represented by their bit decomposition. We assume without loss of generality that $\sampf$ that any auxiliary register is returned to the $\ket{0}$ state. \footnote{Since the output of $\sampf$ on the output registers is a pure state, one can always have $\sampf$ coherently ``uncompute'' on all registers except does containing the output.}
        \item Measure the image register, and let $y_i \in \Y$ be the outcome. As a result, the state has collapsed to:
        $$ \frac{1}{\sqrt{2}} (\ket{0} \ket{x_0} + \ket{1} \ket{x_1}) \,,$$
        where $x_0,x_1 \in \X$ are the unique elements such that $y_i$ is in the support of $f_{k,b}'(x_b)$.
        \item Query the phase oracle for $H$, to obtain:
        $$\frac{1}{\sqrt{2}} ((-1)^{H(\mathsf{BitDecomp}(x_0))}\ket{0} \ket{x_0} + (-1)^{H(\mathsf{BitDecomp}(x_1))}\ket{1} \ket{x_1}) \,. $$
        \item Let $n$ be the length of $\mathsf{BitDecomp}(x_0)$. Apply a Hadamard gate to all of the remaining registers, and measure. Parse the measurement outcome as $z_i || d_i$ where $z_i \in \{0,1\}$ and $d_i \in \{0,1\}^n$. 
        Let $z_i':= z_i \oplus m$.
    \end{itemize}
    \item Let $z' := z_1' \ldots z_L'$, $d := d_1 \ldots d_L$, $y = y_1 \ldots y_L$. Output $c = \left( z', d, y \right)$.
    \end{itemize}
 \item $\mathsf{Dec}(c,t_k) \rightarrow m$:
 \begin{itemize}
     \item Let $c = (z',\, d,\, y)$. Parse $z'$ as $z' = z'_1 \ldots z'_L$. Similarly for $d$ and $y$.
     \item For $i \in [L]$:
 \begin{itemize}
     \item For $b \in \{0,1\}$, run $\invf(t_k, b, y_i)$ to obtain pre-images $x_0^{y_i}$ and $x_1^{y_i}$.
     \item Let $m_i = z'_i \oplus \, d_i \cdot (\mathsf{BitDecomp}(x_0^{y_i}) \oplus \mathsf{BitDecomp}(x_1^{y_i})) \oplus H(\mathsf{BitDecomp}(x_0^{y_i})) \oplus H(\mathsf{BitDecomp}(x_1^{y_i}))$.
 \end{itemize}
 \item If $m_1 = \ldots =  m_L$, output $m_1$, otherwise output $\perp$.
 \end{itemize}
\end{itemize}

\end{construction}

\begin{theorem}
\label{thm: 2}
The scheme of Construction \ref{cons: unexplainable scheme} is a CPA-secure perfectly unexplainable encryption scheme (as in Definition \ref{def: perfect unexp}) in the quantum random oracle model (QROM), assuming the quantum hardness of LWE.
\end{theorem}
The rest of Section \ref{sec: unexp scheme} is dedicated to proving Theorem \ref{thm: 2}. For simplicity, we focus on the uniform version of Definition \ref{def: perfect unexp}. However, our construction also similarly satisfies the \emph{non-uniform} version, assuming the hardness of LWE against \emph{non-uniform} quantum adversaries.

CPA security is straightforward, and we show this in Section~\ref{sec:cpa}. From there on, we focus on showing perfect unexplainability. In Section \ref{sec: comp oracles}, we introduce Zhandry's compressed oracle technique. In Section \ref{sec: xor technical}, we prove technical results about the structure of strategies guessing the xor of two query outputs. In Section \ref{sec: the structure of strategies}, we specialize these results to the concrete setting of guessing ``equations''. In \ref{sec: extracting a claw}, we show how this structure can be used to extract a claw.

\subsection{CPA security}
\label{sec:cpa}
From now on, for ease of notation, when referring to the bit-decomposition of $x$, we simply write $x$ instead $\mathsf{BitDecomp}(x)$ when the context is clear. 

CPA security follows straightforwardly from the following ``hardcore bit'' property satisfied by $\mathcal{F}$, which says that any quantum polynomial-time adversary $\mathcal{A}$ has negligible advantage in the following game between a challenger and $\mathcal{A}$:

\paragraph{Hardcore bit property, Game 1:}
\begin{itemize}
    \item[(i)] The challenger samples $k$, and runs $\sampf(k, \cdot )$ on a uniform superposition of $b$'s to obtain the state
  $$ \frac{1}{\sqrt{|\mathcal{X}|}}\sum_{b\in\{0,1\} , x \in X} \sqrt{f_{k,b}'(x)(y)}\ket{b}\ket{x}\ket{y}\,,$$
    \item[(ii)] Measures the last register to get $y$. 
    \item[(iii)] Then, applies a Hadamard gate on the first two registers to get $z,d$ such that $z = d\cdot (x_0 \oplus x_1)$, where $x_0$ and $x_1$ are the pre-images of $y$. Sends $y,d$ to $\mathcal{A}$.
    \item[(iv)] $\mathcal{A}$ returns $z'$.
\end{itemize}
$\mathcal{A}$ wins if $z = z'$.

\vspace{2mm}
Suppose for a contradiction there was an adversary $\mathcal{A}$ that achieves non-negligible advantage in the above game. Then, the following algorithm $\mathcal{A}'$ recovers a claw. We use the notation $\mathcal{A}(y,d)$ to denote the output of $\mathcal{A}$ on input $y,d$.
\begin{itemize}
\item On input $k$, run $\sampf(k, \cdot )$ on a uniform superposition of $b$'s to obtain the state 
  $$ \frac{1}{\sqrt{|\mathcal{X}|}}\sum_{b\in\{0,1\} , x \in X} \sqrt{f_{k,b}'(x)(y)}\ket{b}\ket{x}\ket{y}\,,$$
  \item Measure the last two registers to obtain $x,y$ where $y \in \supp(f'_{k,b}(x))$ for some $b \in \{0,1\}$.
  \item Run the ``quantum Goldreich-Levin'' extraction algorithm using $\mathcal{A}(y, \cdot)$ as an oracle \cite{adcock2002quantum}. Recall that the ``quantum Goldreich-Levin'' extraction algorithm makes a single query to $\mathcal{A}(y, \cdot)$. Let $s$ be the output of this extraction algorithm.
  \item Output $(x, x \oplus s)$ as the claw.
\end{itemize}
Since $\mathcal{A}(y,d)$ guesses $d \cdot (x_0 \oplus x_1)$ with non-negligible advantage, the ``quantum Goldreich-Levin'' extraction algorithm outputs $s = x_0 \oplus x_1$ with non-negligible probability, which results in $\mathcal{A'}$ outputting a claw with non-negligible probability. Note that the reduction works even if $\mathcal{A}$ is non-uniform with quantum advice, since $\mathcal{A}'$ makes a single query to $\mathcal{A}$, so there is no issue with rewinding.

We now argue that a slight variation of the above hardcore bit property, which involves the random oracle $H$, also holds. We argue that no quantum polynomial-time adversary $\mathcal{A}^H$ (with access to the oracle $H$) has non-negligible advantage in the following game between a challenger and $\mathcal{A}^H$:
\paragraph{Hardcore bit property, Game 2:}
\begin{itemize}
    \item[(i)] The challenger samples $k$, and runs $\sampf(k, \cdot )$ on a uniform superposition of $b$'s to obtain the state
  $$ \frac{1}{\sqrt{|\mathcal{X}|}}\sum_{b\in\{0,1\} , x \in X} \sqrt{f_{k,b}'(x)(y)}\ket{b}\ket{x}\ket{y}\,,$$
    \item[(ii)] Measures the last register to get some outcome $y$. Then, queries the phase oracle for $H$ to create the state
    $$\frac{1}{\sqrt{2}} (-1)^{H(x_0)}\ket{0} \ket{x_0} + \frac{1}{\sqrt{2}}(-1)^{H(x_1)}\ket{1} \ket{x_1} \,,$$
    where $x_0$ and $x_1$ are the pre-images of $y$.
    \item[(iii)] Then, applies a Hadamard gate on the first two registers to get $z,d$ such that $z = d\cdot (x_0 \oplus x_1) \oplus H(x_0) \oplus H(x_1)$. Sends $y,d$ to $\mathcal{A}^H$.
    \item[(iv)] $\mathcal{A}^H$ returns $z'$.
\end{itemize}
$\mathcal{A}^H$ wins if $z = z'$.

\vspace{2mm}
To see that this ``hardcore bit'' property holds, suppose for a contradiction that there was a quantum polynomial-time $\mathcal{A}^H$ achieving non-negligible advantage. Notice that, trivially, there must also be an adversary $\mathcal{A}'^H$ that achieves non-negligible advantage at a slight variation of the game above, where the goal is to guess $d\cdot (x_0 \oplus x_1)$, instead of $d\cdot (x_0 \oplus x_1) \oplus H(x_0) \oplus H(x_1)$: either $\mathcal{A}^H$ already has this property, or the algorithm that outputs the opposite bit of $\mathcal{A}^H$ must have this property. However, notice that the new game no longer depends on $H$ since the distribution of $d,y$ is unaffected by the challenger's query to $H$ (which happens on the first qubit), and the winning condition does not depend on $H$. In fact, the new game is equivalent to Game 1. Thus, the adversary $\mathcal{A''}$ that simulates $\mathcal{A}'^H$ without making queries to $H$ (by using, for example, a $2q$-wise independent family of hash functions, where $q$ is the number of queries made by $\mathcal{A}'$) has identical (and thus non-negligible) advantage at Game 1, which is a contradiction. Hence, no quantum polynomial-time algorithm $\mathcal{A}^H$ can achieve non-negligible advantage at Game 2.

With the hardcore bit property of Game 2 in hand, CPA security of Construction~\ref{cons: unexplainable scheme} is straightforward. We argue CPA security when $L= 1$ in Construction~\ref{cons: unexplainable scheme}. CPA security for general polynomial $L$ follows by a standard hybrid argument. Notice that the distribution $(\pk, \mathsf{Enc}(\pk, 0) )$ is simply $(k, z,d,y )$, such that $z = d\cdot (x_0^y \oplus x_1^y) \oplus H(x_0) \oplus H(x_1)$, where $k,d,z,y$ are sampled exactly from the distribution of the hardcore bit property of Game 2. On the other hand, the distribution $(\pk, \mathsf{Enc}(\pk, 1))$ is $(k, z,d,y )$, such that $z \neq d\cdot (x_0^y \oplus x_1^y) \oplus H(x_0) \oplus H(x_1)$, where $k,z,d,y$ have the same distribution as before, except that $z$ is flipped. Clearly, distinguishing the two distributions is precisely equivalent to guessing the value of $d\cdot (x_0^y \oplus x_1^y) \oplus H(x_0) \oplus H(x_1)$, which is equivalent to Game 2.

\subsection{Compressed oracles}
\label{sec: comp oracles}
In this subsection, we formally introduce Zhandry's technique for recording queries \cite{Zhandry-how}. This section is loosely based on the explanation in \cite{Zhandry-how}. For a slightly more informal treatment, which carries most of the essence, we suggest starting from Section \ref{sec: comp oracles tech overview} in the technical overview.

\paragraph{Standard and Phase Oracles}
The quantum random oracle, which is the quantum analogue of the classical random oracle, is typically presented in one of two variations: as a \emph{standard} or as a \emph{phase} oracle. 

The standard oracle is a unitary acting on three registers: an $n$-qubit register representing the input to the function, an $m$-qubit register for writing the response, and a $m 2^n$ qubit register representing the truth table of the queried function $H: \{0,1\}^n \rightarrow \{0,1\}^m$. The algorithm that queries the standard oracle has access to the first two registers, while the third register, the oracle's state, is hidden from the algorithm except by making queries. The standard oracle unitary acts in the following way on standard basis states:
$$\ket{x}\ket{y}\ket{H} \mapsto \ket{x}\ket{y \oplus H(x)}\ket{H}\,.$$
For a uniformly random oracle, the oracle register is initialized in the uniform superposition $\frac{1}{\sqrt{m2^n}} \sum_{H} \ket{H}$.
This initialization is of course equivalent to having the oracle register be in a completely mixed state (i.e. a uniformly chosen $H$). This equivalence can be seen by just tracing out the oracle register. We denote the standard (uniformly random) oracle unitary by $\mathsf{StO}$. Moreover, for an oracle algorithm $A$, we will denote by $A^{\mathsf{StO}}$ the algorithm $A$ interacting with the standard oracle, implemented as above.

The phase oracle formally gives a different interface to the algorithm making the queries, but is equivalent to the standard oracle up to Hadamard gates. It again acts on three registers: an $n$-qubit register for the input, an $m$-qubit ``phase'' register, and a $m2^n$-qubit oracle register. It acts in the following way on standard basis states:
$$\ket{x}\ket{s}\ket{H} \mapsto (-1)^{s \cdot H(x)}\ket{x}\ket{s}\ket{H} \,.$$
For a uniformly random oracle, the oracle register is again initialized in the uniform superposition. One can easily see that the standard and phase oracles are equivalent up to applying a Hadamard gate on the phase register before and after a query. We denote the phase oracle unitary by $\mathsf{PhO}$. Moreover, for an oracle algorithm $A$, we will denote by $A^{\mathsf{PhO}}$ the algorithm $A$ interacting with the phase oracle.

\paragraph{Compressed oracle}
The \emph{compressed oracle} technique, introduced by Zhandry \cite{Zhandry-how}, is an equivalent way of implementing a quantum random oracle which (i) is efficiently implementable, and (ii) keeps track of the queried inputs in a meaningful way. This paragraph is loosely based on the explanation in \cite{Zhandry-how}.

In a compressed oracle, the oracle register does not represent the full truth table of the queried function. Instead, it represents a \emph{database} of queried inputs, and the values at those inputs. More precisely, if we have an upper bound $t$ on the number of queries, a database $D$ is represented as an element of the set $S^t$ where $S = (\{0,1\}^n \cup \{\perp\} ) \times \{0,1\}^m$. Each value in $S$ is a pair $(x,y)$: if $x \neq \perp$, then the pair means that the value of the function at $x$ is $y$, which we denote by $D(x) = y$; and if $x = \perp$, then the pair is not currently used, which we denote by $D(x) = \perp$. Concretely, let $l \leq t$. Then, for $x_1 < x_2 < \ldots < x_l$ and $y_1,\ldots, y_l$, the database representing $D(x_i) = y_i$ for $i \in [l]$, with the other $t-l$ points unspecified, is represented as
$$\Big( (x_1, y_1), (x_2,y_2), \ldots, (x_l, y_l), (\perp, 0^m), \ldots, (\perp, 0^m) \Big)$$
where the number of $(\perp, 0^m)$ pairs is $t-l$. We emphasize that in this database representation, the pairs are always ordered lexicographically according to the input value, and the $(\perp, 0^m)$ pairs are always at the end.

In order to define precisely the action of a compressed oracle query, we need to introduce some additional notation.

Let $|D|$ denote the number of pairs $(x,y)$ in database $D$ with $x \neq \perp$. Let $t$ be an upper bound on the number of queries. Then, for a database $D$ with $|D|<t$ and $D(x) = \perp$, we write $D \cup (x,y)$ to denote the new database obtained by deleting one of the $(\perp,0^m)$ pairs, and by adding the pair $(x,y)$ to $D$, inserted at appropriate location (to respect the lexicographic ordering of the input values). 

We also define a ``decompression'' procedure. For $x \in \{0,1\}^n$, $\mathsf{StdDecomp}_x$ is a unitary operation on the database register. If $D(x) = \perp$, it adds a uniform superposition over all pairs $(x,y)$ (i.e. it ``uncompressed'' at $x$). Otherwise, if $D$ is specified at $x$, and the corresponding $y$ register is in a uniform superposition, $\mathsf{StdDecomp}$ removes $x$ and the uniform superposition from $D$. If $D$ is specified at $x$, and the corresponding $y$ register is in a state orthogonal to the uniform superposition, then $\mathsf{StdDecomp}$ acts as the identity. More precisely,
\begin{itemize}
    \item For $D$ such that $D(x) = \perp$ and $|D| < t$, 
    $$ \mathsf{StdDecomp}_x \ket{D} = \frac{1}{\sqrt{2^m}} \sum_y \ket{D \cup (x,y)} $$
    \item For $D$ such that $D(x) = \perp$ and $D = t$, 
    $$\mathsf{StdDecomp}_x \ket{D} =\ket{D} $$.
    \item For $D$ such that $D(x) = \perp$ and $|D|<t$,
   \begin{equation}
    \mathsf{StdDecomp}_x \left(\frac{1}{\sqrt{2^n}}\sum_y (-1)^{z \cdot y} \ket{D \cup (x,y)} \right) = \begin{cases} \frac{1}{\sqrt{2^n}}\sum_y (-1)^{z \cdot y} \ket{D \cup (x,y)} \textnormal{ if } z \neq 0 \\
    \ket{D}   \textnormal{ if } z = 0
     \end{cases}
   \end{equation}
\end{itemize}
Note that we have specified the action of $\mathsf{StdDecomp}_x$ on an orthonormal basis of the database register (with a bound of $t$ on the size of the database). Moreover, it is straightforward to verify that $\mathsf{StdDecomp}_x$ maps this orthonormal basis to another orthonormal basis, and is thus a well-defined unitary. Note moreover that applying $\mathsf{StdDecomp}_x$ gives the identity. Let $\mathsf{StdDecomp}$ be the related unitary acting on all the registers $x,y, D$ which acts in the following way:
$$\mathsf{StdDecomp} \ket{x,y}\otimes \ket{D} =  \ket{x,y}\otimes \mathsf{StdDecomp}_x \ket{D} \,.$$

So far, we have considered a fixed upper bound on the number of queries. However, one of the advantages of the compressed oracle technique is that an upper bound on the number of queries does not need to be known in advance. To handle a number of queries that is not fixed, we defined the procedure $\mathsf{Increase}$ which simply increases the upper bound on the size of the database by initializing a new register in the state $\ket{(\perp, 0^n)}$, and appending it to the end. Formally, $\mathsf{Increase} \ket{x,y}\otimes \ket{D} \mapsto \ket{x,y}\otimes \ket{D}\ket{(\perp, 0^n)}$.

Now, define the unitaries $\mathsf{CStO}'$ and $\mathsf{CPhO}'$ acting in the following way:
\begin{align}
    \mathsf{CStO}' \ket{x,y} \otimes \ket{D} &= \ket{x,y\oplus D(x)} \otimes  \ket{D} \nonumber\\
    \mathsf{CPhO}' \ket{x,y} \otimes \ket{D} &= (-1)^{y \cdot D(x)}\ket{x,y} \otimes  \ket{D}
\end{align}
Finally, we define the compressed standard and phase oracles $\mathsf{CStO}$ and $\mathsf{CPhO}$ as:
\begin{align}
    \mathsf{CStO} &= \mathsf{StdDecomp} \circ \mathsf{CStO}' \circ \mathsf{StdDecomp} \circ \mathsf{Increase} \nonumber\\
    \mathsf{CPhO} &= \mathsf{StdDecomp} \circ \mathsf{CPhO}' \circ \mathsf{StdDecomp} \circ \mathsf{Increase}
\end{align}

For an oracle algorithm $A$, we denote by $A^{\mathsf{CStO}}$ (resp. $A^{\mathsf{CPhO}}$), the algorithm $A$ run with the compressed standard (resp. phase) oracle, implemented as described above. The following lemma establishes that regular and compressed oracles are equivalent.
\begin{lem}[\cite{Zhandry-how}]
\label{lem: compressed oracle equiv}
For any oracle algorithm $A$, and any input state $\ket{\psi}$, 
$\Pr[A^{\mathsf{StO}}(\ket{\psi}) = 1] = \Pr[A^{\mathsf{CStO}}(\ket{\psi}) = 1]$. Similarly, for any oracle algorithm $B$, $\Pr[B^{\mathsf{PhO}}(\ket{\psi}) = 1] = \Pr[B^{\mathsf{CPhO}}(\ket{\psi}) = 1]$.
\end{lem}

In the rest of the paper, we choose to work with \emph{phase} oracles and compressed \emph{phase} oracles. Moreover, to use a more suggestive name, we will denote the compressed phase oracle $\mathsf{CPhO}$ by $\comp$.

\subsection{Guessing the XOR of two query outputs}
\label{sec: xor technical}
In this section, we formalize the main technical ingredient that will allow us to characterize the structure of strategies that produce valid equations. We do this via a few related technical lemmas. Lemma \ref{lem: xor} is a general theorem which relates the probability of guessing the xor of the query outputs at two points, to the structure of a strategy that achieves such probability of guessing. We will not use Lemma \ref{lem: xor} directly in our proof of security, but we choose to include it for its generality, as it may be of independent interest. Lemmas \ref{lem: technical} and \ref{lem: structure xor high success} are more specialized, and they establish properties of strategies in certain regimes of guessing probability. We will employ the latter lemmas in our security proof.

Let $n, L \in \mathbb{N}$. In this section, we consider oracle algorithms querying functions in $\mathsf{Bool}(n)$ (i.e. $n$ bit input, and single bit output).

We introduce some notation. When running a compressed oracle simulation, we denote by $\mathsf{O}$ the database register. For convenience, we use the following notation for the database register: for a database $D$ (of pairs in $\{0,1\}^n \times \{0,1\}$), and a subset $X \subseteq \{0,1\}^n$, we write:
$$\ket{D \cup X} := \sum_{w_x \in \{0,1\} : x \in X} (-1)^{w_x} \ket{D \cup \{(x,w_x): x \in X\}} \,.$$
To simplify further, we write $\ket{D \cup x} = \ket{D \cup \{x\}}$. For $x\in \{0,1\}^n$, we concisely write $D \ni x$ (resp. $D \niton x$) to mean that $D$ contains (resp. does not contain) a pair $(x,v)$ for some $v$.

In the following lemma, $P$ is an oracle algorithm acting on an $L$-qubit register $\mathsf{T}$, the output register; $\mathsf{W}$, a work register; $\mathsf{X}$ and $\mathsf{E}$, the query registers (where $\mathsf{X}$ is an $n$-qubit register, and $\mathsf{E}$ is a $1$-qubit register). For ease of notation, we will write $\ket{t,x,e,w} := \ket{t}\ket{x}\ket{e}\ket{w}$.

\begin{lem}
\label{lem: xor}
Let $n, L \in \mathbb{N}$. For $i \in [L]$, $b \in \{0,1\}$, let $x_{b}^i \in \{0,1\}^n $. Let $P$ be an oracle algorithm, and $\ket{\psi}$ a state. Let $\epsilon >0$. Suppose
$$\Pr_{H \gets \mathsf{Bool}(n)}[t_i = H(x_0^i) \oplus H(x_1^i) \textnormal{ for all } i \in [L]: t \gets P^H \ket{\psi}] \geq  1 -\epsilon\,.$$
Let $\ket{\psi_{\textnormal{final}}}$ be the state of $P^{\comp} \ket{\psi}$ just before measurement of the output. Then, for any fixed $i$, there exists a (not necessarily normalized) $\ket{\psi'_{\textnormal{final}}}$ such that:
\begin{align}
\ket{\psi'_{\textnormal{final}}} = &\sum_{\substack{t \in \{0,1\}^L,x,e,w \\ D \niton x_0^i, x_1^i}}\alpha_{t,x,e,w,D}  \ket{t,x,e,w} \sum_{b \in \{0,1\}} (-1)^{b \cdot t_i}\ket{D \cup x_{b}^i} \\
+&\sum_{\substack{t \in \{0,1\}^L,x,e,w \\ D \niton x_0^i, x_1^i}} \beta_{t,x,e,w,D} \ket{t,x,e,w}\Big(\ket{D} + (-1)^{t_i} \ket{D \cup \{x_0^i,x_1^i\}} \Big)
\end{align}
for some $\alpha_{t,x,e,w,D}$ and $\beta_{t,x,e,w,D}$, and
$$ \| \ket{\psi_{\textnormal{final}}} - \ket{\psi'_{\textnormal{final}}} \|^2 \leq 2\epsilon \,.$$
\end{lem}
It is straightforward to see that the above also implies that, upon normalization, $\ket{\psi'_{\textnormal{final}}}$ is $O(\sqrt{\epsilon})$-close to $\ket{\psi_{\textnormal{final}}}$ in Euclidean distance.

\begin{proof}[Proof of Lemma \ref{lem: xor}]
Fix $i$. Then, we can write the state of $P^{\comp} \ket{\psi}$, just before measurement of the output as
\begin{align}
  \ket{\psi_{\textnormal{final}}} = & \sum_{\substack{t \in \{0,1\}^L,x,e,w \\ D \niton x_0^i, x_1^i}}\sum_{b \in \{0,1\}}\alpha_{t,x,e,w,D,b} \ket{t,x,e,w} \ket{D \cup x_{b}^i} \nonumber\\
+&\sum_{\substack{t \in \{0,1\}^L,x,e,w \\ D \niton x_0^i, x_1^i}} \beta_{t,x,e,w,D} \ket{t,x,e,w} \ket{D \cup \{x_0^i,x_1^i\}} \nonumber\\
+&\sum_{\substack{t \in \{0,1\}^L,x,e,w \\ D \niton x_0^i, x_1^i}} \gamma_{t,x,e,w,D} \ket{t,x,e,w} \ket{D} \label{eq: 17}
\end{align}
for some $\alpha_{t,x,e,w,D,b}$, $\beta_{t,x,e,w,D}$ and $\gamma_{t,x,e,w,D}$.

By the equivalence of the standard and compressed oracle simulations (Lemma \ref{lem: compressed oracle equiv}), we have that, for any $i$, $\Pr_{H \gets \mathsf{Bool}(n)}[t_i = H(x_0^i) \oplus H(x_1^i): t_i \gets P^H \ket{\psi}]$ is equal to $\Pr[t_i = D(x_0^i) \oplus D(x_1^i)]$ where $t_i, D(x_0^i),D(x_1^i)$ are distributed as in a compressed oracle simulation. In particular, this is precisely the same as $\Pr[t_i = v_0 \oplus v_1]$, where $t_i = v_0 \oplus v_1$ are distributed as in the following process:
\begin{itemize}
    \item Run $P^\comp \ket{\psi}$, and obtain $t$.
    \item Apply $\mathsf{StdDecomp}_{x_0^i} \circ \mathsf{StdDecomp}_{x_1^i}$ to the leftover state.
    \item Measure the database register to obtain the values $v_0$ and $v_1$ at $x_0^i$ and $x_1^i$ respectively. Output $t, v_0, v_1$.
\end{itemize}

Now, 
\begin{align}
    \mathsf{StdDecomp}_{x_0^i} \circ \mathsf{StdDecomp}_{x_1^i} \ket{\psi_{\textnormal{final}}} = & \frac12 \sum_{\substack{t \in \{0,1\}^L,x,e,w \\ D \niton x_0^i, x_1^i}}\sum_{b \in \{0,1\}}\alpha_{t,x,e,w,D,b}\sum_{v_0, v_1 \in \{0,1\}} (-1)^{v_b} \ket{t,x,e,w} \ket{D \cup \{(x_{0}^i, v_0), (x_1^i, v_1)\}} \nonumber\\
+& \frac12 \sum_{\substack{t \in \{0,1\}^L,x,e,w \\ D \niton x_0^i, x_1^i}} \beta_{t,x,e,w,D} \sum_{v_0, v_1 \in \{0,1\}} (-1)^{v_0+v_1} \ket{t,x,e,w} \ket{D \cup \{(x_{0}^i, v_0), (x_1^i, v_1)\}} \nonumber\\
+& \frac12 \sum_{\substack{t \in \{0,1\}^L,x,e,w \\ D \niton x_0^i, x_1^i}} \gamma_{t,x,e,w,D} \sum_{v_0, v_1 \in \{0,1\}} \ket{t,x,e,w} \ket{D \cup \{(x_{0}^i, v_0), (x_1^i, v_1)\}} \label{eq: 22}
\end{align}

Then, 
\begin{align}
  \epsilon &\geq \Pr_{H \gets \mathsf{Bool}(n)}[t_j \neq H(x_0^j) \oplus H(x_1^j) \textnormal{ for some } j \in [L]: t \gets P^H \ket{\psi}] \nonumber\\
  &\geq  \Pr_{H \gets \mathsf{Bool}(n)}[t_i \neq H(x_0^i) \oplus H(x_1^i): t \gets P^H \ket{\psi}]\nonumber \\
  &= \frac14 \sum_{\substack{t \in \{0,1\}^L,x,e,w \\ D \niton x_0^i, x_1^i}} \sum_{\substack{v_0, v_1: \\ v_0 \oplus v_1 \neq t_i}} \big| \sum_b (-1)^{v_b}  \alpha_{t,x,e,w,D,b} + (-1)^{v_0 + v_1}\beta_{t,x,e,w,D}  + \gamma_{t,x,e,w,D} \big|^2 \nonumber\\
  &=\frac14 \sum_{\substack{t \in \{0,1\}^L,x,e,w \\ D \niton x_0^i, x_1^i}} \sum_{\substack{v_0, v_1: \\ v_0 \oplus v_1 \neq t_i}} \big| \sum_b (-1)^{v_0 + b \cdot \bar{t_i}}  \alpha_{t,x,e,w,D,b} + (-1)^{\bar{t_i}} \beta_{t,x,e,w,D} + \gamma_{t,x,e,w,D} \big|^2 \,, \label{eq: 23}
\end{align}
where the first line is by the hypothesis of the lemma, and the third line follows from calculating $\Pr[t_i \neq v_0 \oplus v_1]$ in expression \eqref{eq: 22}.

We can further simplify expression \eqref{eq: 23} as:
\begin{align}
    \eqref{eq: 23} &= \frac12 \sum_{\substack{t \in \{0,1\}^L,x,e,w \\ D \niton x_0^i, x_1^i}} \big| \sum_b (-1)^{b \cdot \bar{t_i}}  \alpha_{t,x,e,w,D,b} \big |^2 \nonumber \\
    &+ \frac12 \sum_{\substack{t \in \{0,1\}^L,x,e,w \\ D \niton x_0^i, x_1^i}} \big|  (-1)^{ \bar{t_i}}  \beta_{t,x,e,w,D} + \gamma_{t,x,e,w,D} \big |^2 \nonumber \\
    &= \frac12 \sum_{\substack{t \in \{0,1\}^L,x,e,w \\ D \niton x_0^i, x_1^i}} \big| \alpha_{t,x,e,w,D,0} - (-1)^{t_i} \alpha_{t,x,e,w,D,1}\big |^2 \nonumber \\
    &+ \frac12 \sum_{\substack{t \in \{0,1\}^L,x,e,w \\ D \niton x_0^i, x_1^i}} \big| \gamma_{t,x,e,w,D} - (-1)^{t_i}  \beta_{t,x,e,w,D}  \big |^2 \,, \label{eq: 24}
\end{align}
where the first equality is a crucial part of the calculation, which follows from the following simple fact: for any complex numbers $\beta, \gamma$, 
$$|\beta + \gamma|^2 + |\beta - \gamma|^2 = 2|\beta|^2 + 2| \gamma|^2 \,.$$
Now, consider the (not necessarily normalized) state $\ket{\psi'_{\textnormal{final}}}$ which is defined as in \eqref{eq: 17}, except we additionally set $\alpha_{t,x,e,w,D,0} = (-1)^{t_i} \alpha_{t,x,e,w,D,1}$ for all $t,x,e,w,D$, and $\gamma_{t,x,e,w,D} = (-1)^{t_i} \beta_{t,x,e,w,D}$. Then, notice that
\begin{align}
    \|\ket{\psi_{\textnormal{final}}} - \ket{\psi'_{\textnormal{final}}}\|^2 &=  \sum_{\substack{t \in \{0,1\}^L,x,e,w \\ D \niton x_0^i, x_1^i}} |\alpha_{t,x,e,w,D,0} - (-1)^{t_i} \alpha_{t,x,e,w,D,1}  |^2  +
   \sum_{\substack{t \in \{0,1\}^L,x,e,w \\ D \niton x_0^i, x_1^i}} |\gamma_{t,x,e,w,D} - (-1)^{t_i} \beta_{t,x,e,w,D} |^2  \nonumber\\
    &\leq 2 \epsilon \,. \label{eq: 25}
\end{align}
where the last line follows from Equations \eqref{eq: 23} and \eqref{eq: 24}.

\end{proof}

Lemmas \ref{lem: technical} and \ref{lem: structure xor high success} below are technical lemmas, which will be useful in the security proof. In words, Lemma \ref{lem: technical} says the following. As earlier, for $i \in [L]$, let $(x_0^i, x_1^i) \in \{0,1\}^n \times \{0,1\}^n$. Consider a state resulting from a compressed oracle simulation of an algorithm acting on the registers $\mathsf{T}, \mathsf{X}, \mathsf{E}, \mathsf{W}$, as introduced earlier. Suppose that the state only has weight on databases such that, for at least $k$ of the indices $i$, the database contains neither $x_0^i$ nor $x_1^i$. Then, $\Pr[t_i = D(x_0^i) \oplus D(x_1^i) \textnormal{ for all } i] $ is at most $2^{-k}$. In other words, for each index  $i$ such that neither $x_0^i$ nor $x_1^i$ are present in the database, the probability of correctly guessing all of the xor's correctly, drops by a factor of $2$.

In order to state this lemma formally, in the form that will be useful later, we introduce some additional notation. Denote by $\mathcal{S}_{\mathsf{comp}}$ the set of all (normalized) states of the form:
$$ \sum_{t,x,e,w,X \subseteq \{0,1\}^n } \alpha_{t,x,e,w, X} \ket{t,x,e,w}\ket{D_X} \,.$$
These are states that can be reached by running a compressed oracle simulation. Next, for a set $X = \{x_b^i \in \{0,1\}^n: i\in [L], b \in \{0,1\} \}$, define the projector $\Pi_{\mathsf{no-claw}(X)} := \sum_{D: D \nsupseteq \{x_0^i, x_1^i\}} \ket{D}\bra{D}_{\mathsf{O}}$. 
Denote by $\mathcal{S}_{\mathsf{no-claw}(X)}$ the set of (normalized) states $\ket{\psi}$ on registers $\mathsf{TXEWO}$ such that $\Pi_{\mathsf{no-claw}(X)} \ket{\psi} =\ket{\psi} $. This is the set of states that have zero weight on databases containing a ``claw''. Further, for $l \in [L]$, define the projector 
$$\Pi_{X, \leq l} := \sum_{\substack{D: \\ |\{ i: D \cap \{x_0^i, x_1^i\} \neq \emptyset\} |\leq l  }} \ket{D}\bra{D}_{\mathsf{O}}\,.$$
Denote by $\mathcal{S}_{X, \leq l}$ the set of (normalized) states $\ket{\psi}$ on registers $\mathsf{TXEWO}$ such that $\Pi_{X, \leq l} \ket{\psi} =\ket{\psi}$. In words, this is the set of states with weight only on databases which, for at most $l$ of the indices in $[L]$, contain any of $x_0^i$ or $x_1^i$. 
Let $\mathsf{DecompAll}$ be the map that ``decompresses'' at every input, i.e.
$$ \mathsf{DecompAll} := \bigotimes_{x \in \{0,1\}^n} \mathsf{StdDecomp}_x \,.$$
Finally, define 
$$\Pi_{\mathsf{valid}} := \sum_{\substack{t, D:\\ \forall i, \,t_i = D(x_0^i) \oplus D(x_1^i)}} \ket{t}\bra{t} \otimes \ket{D}\bra{D}\,.$$ 

\begin{lem}
\label{lem: technical}
Let $n, L \in \mathbb{N}$. Let $X = \{x_b^i \in \{0,1\}^n: i\in [L], b \in \{0,1\} \}$. Let $\ket{\psi} \in \mathcal{S}_{\mathsf{comp}} \cap \mathcal{S}_{\mathsf{no-claw}(X)} \cap \mathcal{S}_{X, \leq l}$. Then, 
$$\| \Pi_{\mathsf{valid}} \mathsf{DecompAll} \ket{\psi}\|^2 \leq  2^{-L+l} \,. $$
\end{lem}

\begin{proof}
First, for a set $A$, we write $\mathsf{StdDecomp}_{A} := \bigotimes_{x \in A} \mathsf{StdDecomp}_{x}$. Then, 
\begin{align*}
\| \Pi_{\mathsf{valid}} \,\mathsf{DecompAll} \ket{\psi}\|^2 &= \| \Pi_{\mathsf{valid}} \circ (\mathsf{StdDecomp}_{\{0,1\}^n \setminus X} \otimes \mathsf{StdDecomp}_{X} ) \ket{\psi}\|^2 
\end{align*}
Now, notice that $\Pi_{\mathsf{valid}}$ and $\mathsf{StdDecomp}_{\{0,1\}^n \setminus X}$ commute, because $\Pi_{\mathsf{valid}}$ acts trivially on all database elements other than those in $X$. Thus, 
\begin{align*}
\| \Pi_{\mathsf{valid}} \,\mathsf{DecompAll} \ket{\psi}\|^2\| & = \| \mathsf{StdDecomp}_{\{0,1\}^n \setminus X} \circ \Pi_{\mathsf{valid}} \circ \mathsf{StdDecomp}_{X}  \ket{\psi}\|^2 \\
& = \| \Pi_{\mathsf{valid}} \,   \mathsf{StdDecomp}_{X}  \ket{\psi}\|^2
\end{align*}
where the last line is simply because $\mathsf{StdDecomp}_{\{0,1\}^n \setminus X}$ is a unitary.

Now, since $\ket{\psi} \in \mathcal{S}_{\mathsf{comp}} \cap \mathcal{S}_{\mathsf{no-claw}(X)}$, we can write 
\begin{equation*}
  \ket{\psi} =  \sum_{\substack{t \in \{0,1\}^L,x,e,w \\ D \cap X = \emptyset}} \sum_{\substack{I \subseteq [L] \\ b = (b_i: \, i \in I)}} \alpha_{t,x,e,w,D,I,b} \ket{t,x,e,w} \ket{D \cup \{x_{b_i}^i : i \in I\}} \,.
\end{equation*}
Then, 
$$ \mathsf{StdDecomp}_{X}  \ket{\psi} = \sum_{\substack{t \in \{0,1\}^L,x,e,w \\ D \cap X = \emptyset}} \sum_{\substack{I \subseteq [L] \\ b = (b_i: \, i \in I)}} \alpha_{t,x,e,w,D,I,b} \sum_{v_0^i, v_1^i \in \{0,1\}, \,i \in [L]} (-1)^{\sum_{i \in I} v^i_{b_i}}\ket{t,x,e,w} \ket{D \cup \{(x_0^i, v_0^i), (x_1^i, v_1^i) : i \in I\}} \,. $$
We can then calculate
\begin{align}
    \| \Pi_{\mathsf{valid}} \,   \mathsf{StdDecomp}_{X}  \ket{\psi}\|^2 &= \sum_{\substack{t \in \{0,1\}^L,x,e,w \\ D \cap X = \emptyset}} \sum_{\substack{v_0^i, v_1^i \in \{0,1\}, \,i \in [L]:\\ \forall i, \, t_i = v_0^i \oplus v_1^i  }} \big| \sum_{\substack{I \subseteq [L] \\ b = (b_i: \, i \in I)}} (-1)^{\sum_{i \in I} v^i_{b_i}} \alpha_{t,x,e,w,D,I,b}  \big|^2  \nonumber\\
    &= \sum_{\substack{t \in \{0,1\}^L,x,e,w \\ D \cap X = \emptyset}} \sum_{\substack{v_0^i, v_1^i \in \{0,1\}, \,i \in [L]:\\ \forall i, \, t_i = v_0^i \oplus v_1^i  }} \big| \sum_{\substack{I \subseteq [L] \\ b = (b_i: \, i \in I)}} (-1)^{\sum_{i \in I} v^i_{0} + \sum_{i \in I} b_i \cdot t_i} \alpha_{t,x,e,w,D,I,b}  \big|^2  \label{eq: 280}
\end{align}
We can rewrite \eqref{eq: 280} as 
\begin{align}
    \sum_{\substack{t \in \{0,1\}^L,x,e,w \\ D \cap X = \emptyset}}\, \sum_{v_0^i \in \{0,1\}, \,i \in [L]} \big| \sum_{I \subseteq [L]} (-1)^{\sum_{i \in I} v^i_{0}} \sum_{b = (b_i: \, i \in I)} (-1)^{\sum_{i \in I} b_i \cdot t_i} \alpha_{t,x,e,w,D,I,b}  \big|^2 \nonumber\\
    =  \frac{1}{2^{2L}} \sum_{\substack{t \in \{0,1\}^L,x,e,w \\ D \cap X = \emptyset}} 2^L  \sum_{I \subseteq [L]} \big| \sum_{b = (b_i: \, i \in I)} (-1)^{\sum_{i \in I} b_i \cdot t_i} \alpha_{t,x,e,w,D,I,b}  \big|^2 \,, \label{eq: 290}
\end{align}
where, similarly to the proof of Lemma \ref{lem: xor}, the last equality is the crucial part of the calculation, and it follows from the following generalization of the fact about complex numbers used there. Let $N,M \in \mathbb{N}$. For $i \in [M]$, let $w_{1,i} \ldots, w_{N,i} \in \{0,1\}$, be such that: for any $s\neq t \in [N]$, $|\{i \in [M]: w_{s,i} = w_{t,i} \}| = M/2$. Then, for any complex numbers $\beta_1, \ldots, \beta_N$,
$$ \sum_{i \in [M]} |\sum_{j \in [N]} (-1)^{w_{j,i}} \beta_j|^2  = M \cdot (|\beta_1|^2 + \cdots + |\beta_N|^2) \,.$$
Verifying this fact is a straightforward calculation.
Then, to be explicit, to obtain \eqref{eq: 290}, we apply the latter fact for each fixed $t,x,e,w$, with $\sum_{i \in [M]}$ corresponding to $\sum_{v_0^i \in \{0,1\}, \,i \in [L]}$; $\sum_{j \in [N]}$ corresponding to $\sum_{I \subseteq [L]}$; for $I \subseteq [L]$, $\beta_I = \sum_{b = (b_i: \, i \in I)} (-1)^{\sum_{i \in I} b_i \cdot t_i} \alpha_{t,x,e,w,D,I,b}$; and the $w$'s are naturally determined by the $v_0^i$.

Now, it is not difficult to see that the RHS of \eqref{eq: 290} is upper bounded by
\begin{align}
&\frac{1}{2^{2L}} \sum_{\substack{t \in \{0,1\}^L,x,e,w \\ D \cap X = \emptyset}} 2^L \cdot 2^{l} \sum_{I \subseteq [L]} \sum_{b = (b_i: \, i \in I)}  \big| \alpha_{t,x,e,w,D,I,b}  \big|^2 \nonumber\\
= & \frac{1}{2^{L-l}} \sum_{I \subseteq [L]}  \sum_{\substack{t \in \{0,1\}^L,x,e,w \\ D \cap X = \emptyset}} \sum_{b = (b_i: \, i \in I)}  \big| \alpha_{t,x,e,w,D,I,b}  \big|^2 \,. \nonumber\\
= &\frac{1}{2^{L-l}} \cdot \| \ket{\psi}\|^2  = \frac{1}{2^{L-l}}\,.
\end{align}
Putting everything together, we have shown
$$ \|\Pi_{\mathsf{valid}} \mathsf{DecompAll} \ket{\psi} \|^2 \leq   2^{-L+l} \,,$$
as desired.
\end{proof}

The next lemma is essentially an abstraction of Lemma \ref{lem: xor}, stated in a form that will be useful for our security proof later.

\begin{lem}
\label{lem: structure xor high success}
Let $n, L \in \mathbb{N}$. For $i \in [L]$, $b \in \{0,1\}$, let $x_{b}^i \in \{0,1\}^n $. Let $\ket{\psi} \in \mathcal{S}_{\mathsf{comp}}$. Suppose 
$$\| \Pi_{\mathsf{valid}} \circ \mathsf{DecompAll} \ket{\psi} \|^2 \geq 1 - \epsilon \,,$$
Then, for each $i$, there exists a (not necessarily normalized) state \begin{align}
\ket{\psi'}= & \sum_{\substack{t,x,,y,x,e,w \\ D \niton x_0^{i}, x_1^{i}}}\alpha_{t,x,e,w,D} \ket{t,x,e,w} \Big(\ket{D \cup x_{0}^{i}} + (-1)^{t_i} \ket{D \cup x_1^{i}}\Big) \nonumber\\
+&\sum_{\substack{t,x,e,w  \\ D \niton x_0^{i}, x_1^{i}}} \beta_{t,x,e,w,D} \ket{t,x,e,w} \Big( \ket{D} + (-1)^{t_i} \ket{D \cup \{x_0^{i},x_1^{i}\}}\Big) \,,
\end{align}
for some coefficients $\alpha_{t,x,e,w,D}$, $\beta_{t,x,e,w,D}$, such that
$$\| \ket{\psi} - \ket{\psi}' \|^2 \leq 2 \epsilon \,.$$
\end{lem}

\begin{proof}
In general, for any $i$, we can write the state $\ket{\psi}$ as  
\begin{align}
    \ket{\psi}= & \sum_{\substack{t,x,e,w \\ D \niton x_0^{i}, x_1^{i}}}\sum_{b \in \{0,1\}}\alpha_{t,x,e,w,D,b} \ket{t,x,e,w} \ket{D \cup x_{b}^i} \nonumber\\
+&\sum_{\substack{t,x,e,w  \\ D \niton x_0^{i}, x_1^{i}}} \beta_{t,x,e,w,D} \ket{t,x,e,w} \ket{D \cup \{x_0^{i},x_1^{i}\}} \nonumber\\
+&\sum_{\substack{t,x,e,w  \\ D \niton x_0^{i}, x_1^{i}}} \gamma_{t,x,e,w,D} \ket{t,x,e,w} \ket{D}
\end{align}
for some $\alpha_{t,x,e,w,D,b}, \beta_{t,x,e,w,D},  \gamma_{t,x,e,w,D}$.
Let
\begin{align}
    \ket{\psi'}= & \sum_{\substack{t,x,e,w \\ D \niton x_0^{i}, x_1^{i}}}\alpha_{t,x,e,w,D} \ket{t,x,e,w} \Big(\ket{D \cup x_{0}^{i}} + (-1)^{t_i} \ket{D \cup x_1^{i}}\Big) \nonumber\\
+&\sum_{\substack{t,x,e,w  \\ D \niton x_0^{i}, x_1^{i}}} \beta_{t,x,e,w,D} \ket{t,x,e,w} \Big( \ket{D} + (-1)^{t_i} \ket{D \cup \{x_0^{i},x_1^{i}\}}\Big) \,,
\end{align}
where $\alpha_{t,x,e,w,D} := \alpha_{t,x,e,w,D,0}$.

By applying the exact same analysis as in the proof of Lemma \ref{lem: xor}, we obtain that
$$\|\ket{\psi} - \ket{\psi'} \|^2 \leq 2\epsilon \,.$$

\end{proof}

\subsection{The structure of strategies that produce valid equations}
\label{sec: the structure of strategies}
In this section, we characterize the structure of strategies that produce valid equations. The results in this section essentially follow from the results of the previous section (about the structure of strategies that guess the xor of query outputs), specialized to the setting where databases do not contain a claw. We first state and prove a somewhat clean and general lemma (Lemma \ref{lem: clean structure}), for the case of guessing a single equation, relating the probability of success to the structure of the strategy. Although we will not be using the latter lemma directly, we include it as it modular enough that it may be of independent interest. Then, we state two corollaries of lemmas \ref{lem: technical} and \ref{lem: structure xor high success}, which we will use directly in the security proof later. 

Throughout this section, let $\X,\Y,\K$ be finite sets. Let $\mathcal{F} \,=\, \big\{f_{k,b} : \X\rightarrow \mathcal{D}_{\Y} \big\}_{k\in \mathcal{K},b\in\{0,1\}}$ be a family of noisy trapdoor claw-free functions. Let $f'_{k,b}: \X \rightarrow \Y$ be functions satisfying the \emph{efficient range superposition property} of Definition \ref{def: ntcfs}.

Let $P$ be a prover. Denote the registers that $P$ acts on as: $\mathsf{Z}, \mathsf{D}, \mathsf{Y}$, corresponding to the equation; $\mathsf{X}, \mathsf{E}$, corresponding to the query registers; $\mathsf{W}$, corresponding to the witness. For ease of notation, we omit keeping track of $P$'s input registers, and any auxiliary work register. Assume that $x \in \mathsf{X}$ is represented in register $\mathsf{X}$ via its bit decomposition. Let $n$ denote the size of the bit decomposition.
For ease of notation, for $x \in \X$, we omit writing $\mathsf{BitDecomp}(x)$ in this section, when it is clear from the context.

All parameters implicitly depend on the security parameter $\lambda$, which we omit writing. 
\begin{lem}
\label{lem: clean structure}
Let $P$ be any non-uniform quantum polynomial-time oracle algorithm taking as input a choice of claw-free function pair $k$, and outputting a state on registers $\mathsf{ZDYXEW}$. Suppose there exists a function $\delta$ (of the security parameter) such that :
$$\Pr[z = d \cdot (x_0^y \oplus x_1^y) \oplus H(x_0^y) \oplus H(x_1^y): k\gets \mathcal{F}(1^{\lambda}),  (z, d, y) \gets P^H(k), H \gets \mathsf{Bool}(n)] \geq  \frac12 + \delta(\lambda) \,.$$

Let the state of $P^{\comp}(k)$ just before measurement of the output be \begin{align}
  \ket{\psi^k_{\textnormal{final}}} = & \sum_{\substack{z,d,y,x,e,w \\ D \niton x_0^y, x_1^y}}\sum_{b \in \{0,1\}}\alpha^k_{z,d,y,x,e,w,D,b} \ket{z,d,y,x,e,w} \ket{D \cup x_{b}^y} \nonumber\\
+&\sum_{\substack{z,d,y,x,e,w \\ D \niton x_0^y, x_1^y}} \beta^k_{z,d,y,x,e,w,D} \ket{z,d,y,x,e,w} \ket{D \cup \{x_0^y,x_1^y\}} \nonumber\\
+&\sum_{\substack{z,d,y,x,e,w \\ D \niton x_0^y, x_1^y}} \gamma^k_{z,d,y,x,e,w,D} \ket{z,d,y,x,e,w} \ket{D} \label{eq: 17}
\end{align}
Then, there exists a negligible function $\negl$ such that, for all $\lambda$:
\begin{itemize}
    \item[(i)] \begin{equation*}\mathbb{E}_{k \gets \mathcal{F}} \sum_{\substack{z,d,y,x,e,w \\ D \niton x_0^y, x_1^y}}|\beta^k_{z,d,y,x,e,w,D}|^2 = \negl(\lambda) \,.\end{equation*}
    \item[(ii)] \begin{equation*}\mathbb{E}_{k \gets \mathcal{F}} \sum_{\substack{z,d,y,x,e,w \\ D \niton x_0^y, x_1^y}} \sum_{b \in \{0,1\}} |\alpha^k_{z,d,y,x,e,w,D,b}|^2 \geq 2 \delta(\lambda) - \negl(\lambda) \,.\end{equation*}
    \item[(iii)] \begin{align*}&\mathbb{E}_{k \gets \mathcal{F}} \sum_{\substack{z,d,y,x,e,w \\ D \niton x_0^y, x_1^y}}|\alpha^k_{z,d,y,x,e,w,D,0} - (-1)^{z + d \cdot (x_0^y \oplus x_1^y)} \alpha^k_{z,d,y,x,e,w,D,1}|^2 \\  \leq  \,&\mathbb{E}_{k \gets \mathcal{F}}\sum_{\substack{z,d,y,x,e,w \\ D \niton x_0^y, x_1^y}}  \sum_{b \in \{0,1\}}|\alpha^k_{z,d,y,x,e,w,D,b}|^2 - 2\delta(\lambda) + \negl(\lambda) \,. \nonumber \end{align*}
\end{itemize}
\end{lem}
In words, this lemma establishes the following:
\begin{itemize}
    \item The state of the oracle algorithm after a compressed oracle simulation has at most negligible weight on databases containing a claw. This is point (i).
    \item In the regime where $\delta$ is non-negligible, there must be non-negligible weight on $y,D$ such that $D$ contains at least one of the two pre-images of $y$. Moreover, on average over $k$, there is non-negligible weight on $z,d,y,x,e,w,D$, such that $\alpha^k_{z,d,y,x,e,w,D, 0}$ is a non-negligible fraction of $\alpha^k_{z,d,y,x,e,w,D, 1}$ and viceversa.
    \item In the regime where the winning probability is close to $1$, almost all the weight must be on $y,D$ such that $D$ contains at least one of the two pre-images of $y$. Moreover, it must be that, approximately, $\alpha^k_{z,d,y,x,e,w,D, 0} = (-1)^{z + d \cdot (x_0^y \oplus x_1^y)} \alpha^k_{z,d,y,x,e,w,D, 1}$.
\end{itemize}

The proof uses ideas encountered in the previous section.
\begin{proof}[Proof of Lemma \ref{lem: clean structure}]
By an analogous calculation as in the proof of Lemma \ref{lem: xor}, we have that, for all $\lambda$,
\begin{align}
& \frac12 - \delta(\lambda) \nonumber \\
    &\geq \Pr[z \neq d \cdot (x_0^y \oplus x_1^y) \oplus H(x_0^y) \oplus H(x_1^y): k\gets \mathcal{F}(1^{\lambda}),  (z, d, y) \gets P^H(k), H \gets \mathsf{Bool}(n)] \nonumber \\
    &= \,\mathbb{E}_{k} \bigg(\frac12 \sum_{\substack{z,d,y,x,e,w \\ D \niton x_0^y, x_1^y}} \big| \alpha^k_{z,d,y,x,e,w,D,0} - (-1)^{z + d \cdot (x_0^y \oplus x_1^y)} \alpha_{z,d,y,x,e,w,D,1}\big |^2 \nonumber \\
    &\,\,\,\,+ \frac12 \sum_{\substack{z,d,y,x,e,w \\ D \niton x_0^y, x_1^y}} \big| \gamma^k_{z,d,y,x,e,w,D} - (-1)^{z + d \cdot (x_0^y \oplus x_1^y)}  \beta^k_{z,d,y,x,e,w,D}  \big |^2 \bigg) \,. \label{eq: 410}
\end{align}
Now, notice that there exists a negligible function $\negl$, such that, for all $\lambda$,
\begin{equation}
    \mathbb{E}_{k}\sum_{\substack{z,d,y,x,e,w \\ D \niton x_0^y, x_1^y}}|\beta^k_{z,d,y,x,e,w,D}|^2 = \negl(\lambda) \,, \label{eq: 430}
\end{equation}
otherwise this would immediately yield an efficient algorithm to extract a claw. This, combined with \eqref{eq: 410}, straightforwardly implies that, for all $\lambda$,
\begin{align}
\frac12 - \delta(\lambda) +\negl'(\lambda)
    &\geq
    \,\mathbb{E}_{k} \bigg(\frac12 \sum_{\substack{z,d,y,x,e,w \\ D \niton x_0^y, x_1^y}} \big| \alpha^k_{z,d,y,x,e,w,D,0} - (-1)^{z + d \cdot (x_0^y \oplus x_1^y)} \alpha_{z,d,y,x,e,w,D,1}\big |^2 \nonumber \\
    &\,\,\,\,\,\,\,\,\,\,\,\,\,\,\,\,\,\,+ \frac12 \sum_{\substack{z,d,y,x,e,w \\ D \niton x_0^y, x_1^y}} \big| \gamma^k_{z,d,y,x,e,w,D}\big |^2 \bigg) \,, \label{eq: 420}
\end{align}
for some other negligible function $\negl'$. Equation \eqref{eq: 430} also implies that, for all $\lambda$,
\begin{equation}
\label{eq: 450}
   \mathbb{E}_k \sum_{\substack{z,d,y,x,e,w \\ D \niton x_0^y, x_1^y}} \big| \gamma^k_{z,d,y,x,e,w,D}\big |^2 \geq  1 - \mathbb{E}_k \sum_{\substack{z,d,y,x,e,w \\ D \niton x_0^y, x_1^y}} \sum_{b \in \{0,1\}} \big| \alpha^k_{z,d,y,x,e,w,D,b}\big |^2 - \negl(\lambda) 
\end{equation}

Plugging this into \eqref{eq: 420}, gives 
\begin{align}
\frac12 - \delta(\lambda) +\negl'(\lambda)
    &\geq
    \,\mathbb{E}_{k} \bigg(\frac12 \sum_{\substack{z,d,y,x,e,w \\ D \niton x_0^y, x_1^y}} \big| \alpha^k_{z,d,y,x,e,w,D,0} - (-1)^{z + d \cdot (x_0^y \oplus x_1^y)} \alpha_{z,d,y,x,e,w,D,1}\big |^2 \nonumber \\
    &\,\,\,\,\,\,\,\,\,\,\,\,\,\,\,\,\,\,+ \frac12 - \frac12 \mathbb{E}_k \sum_{\substack{z,d,y,x,e,w \\ D \niton x_0^y, x_1^y}} \sum_{b \in \{0,1\}} \big| \alpha^k_{z,d,y,x,e,w,D,b}\big |^2 - \negl(\lambda) \bigg) \,, \label{eq: 440}
\end{align}

Now, to prove (ii), suppose for a contradiction that there exists some polynomial $p$ such that, for infinitely many $\lambda$,
\begin{equation*}\mathbb{E}_{k} \sum_{\substack{z,d,y,x,e,w \\ D \niton x_0^y, x_1^y}} \sum_{b \in \{0,1\}} |\alpha^k_{z,d,y,x,e,w,D,b}|^2 < 2 \delta(\lambda) - p(\lambda) \,.
\end{equation*} 
Then, for any such $\lambda$, plugging this into \eqref{eq: 440}, and rearranging, gives
\begin{equation} 
\negl'(\lambda) + \negl(\lambda) \geq \frac12 \mathbb{E}_{k}  \sum_{\substack{z,d,y,x,e,w \\ D \niton x_0^y, x_1^y}} \big| \alpha^k_{z,d,y,x,e,w,D,0} - (-1)^{z + d \cdot (x_0^y \oplus x_1^y)} \alpha_{z,d,y,x,e,w,D,1}\big |^2 + 2p(\lambda) \,,
\end{equation}
which is a contradiction. 

Finally, (iii) follows similarly from \eqref{eq: 440}.
\end{proof} 

Next, we proceed to stating and proving two corollaries of Lemmas \ref{lem: technical} and \ref{lem: structure xor high success} from the previous section. In this case, we do not restrict ourselves to the case of guessing a single equation. So, registers $\mathsf{Z}, \mathsf{D}, \mathsf{Y}$, each consist of $L$ sub-registers, for some $L \in \mathbb{N}$. So $\mathsf{Y} = \mathsf{Y}_1\ldots\mathsf{Y}_L$, and similarly for $\mathsf{Z}$ and $\mathsf{D}$. We will write $y = y_1 \ldots y_L$ where $y_i$ is the $i$-th image; $z = z_1\ldots z_L$, where $z_i \in \{0,1\}$; $d = d_1 \ldots d_L$, where $d_i \in \{0,1\}^n$.

We introduce some additional notation. For a subset $X \subseteq \X$, we define the following state of the database register $\mathsf{O}$,  $$D_{X} := \sum_{w_x \in \{0,1\}: x \in X} (-1)^{\sum_x w_x} \ket{\{(x, w_x): x \in X\}}\,.$$ 

Analogously to the previous section, let $\mathcal{S}_{\mathsf{comp}}$ be the set of states, on registers $\mathsf{ZDYXEWO}$, that can be reached by running a compressed oracle simulation. For an NTCF key $k$, let $\mathcal{S}_{k, \mathsf{no-claw}}$ be the set of states, on registers $\mathsf{ZDYXEWO}$, that have zero weight on $y, D$ such that $D$ contains a claw for $y_i$ for some $i$. 

Further, for $l \in [L]$, define the projector 
$$\Pi_{k, \leq l} := \sum_{\substack{y, D: \\ |\{ i: D \cap \{x_0^{y_i}, x_1^{y_i}\} \neq \emptyset\} |\leq l  }} \ket{y}\bra{y}_{\mathsf{Y}} \otimes \ket{D}\bra{D}_{\mathsf{O}}\,.$$
Denote by $\mathcal{S}_{k, \leq l}$ the set of states $\ket{\psi}$ on registers $\mathsf{ZDYXEWO}$ such that $\Pi_{k, \leq l} \ket{\psi} =\ket{\psi}$.


Finally, define $$\Pi_{\mathsf{valid}, k} := \sum_{\substack{z,d,y,D:\\  \forall i, \,\, z_i \oplus d_i \cdot (x_0^{y_i} \oplus x_1^{y_i}) = D(x_0^{y_i}) \oplus D(x_1^{y_i}) }} \ket{z,d,y}\bra{z,d,y} \otimes \ket{D}\bra{D} \,,$$
where the pre-images are with respect to key $k$. In words, $\Pi_{\mathsf{valid},k}$ is the projector onto states for which all $L$ equations are valid, with respect to key $k$.

\begin{cor}
\label{cor: 1}
Let $k$ be an NTCF key. Let $\ket{\psi} \in \mathcal{S}_{\mathsf{comp}} \cap \mathcal{S}_{k, \mathsf{no-claw}} \cap \mathcal{S}_{k, \leq l}$. Then, 
$$\| \Pi_{\mathsf{valid},k} \mathsf{DecompAll} \ket{\psi}\|^2 \leq  2^{-L+l} \,. $$
\end{cor}

\begin{proof}
For a fixed $y$, let 
$$\Pi_{\mathsf{valid}, k,y} := \sum_{\substack{z,d,D:\\ \forall i, \,\, z_i \oplus d_i \cdot (x_0^{y_i} \oplus x_1^{y_i}) = D(x_0^{y_i}) \oplus D(x_1^{y_i})  }} \ket{z,d}\bra{z,d} \otimes \ket{D}\bra{D} \,.$$

Then, we can write $\Pi_{\mathsf{valid},k}$ as 
\begin{equation*}
    \Pi_{\mathsf{valid},k} = \sum_{y} \ket{y}\bra{y} \otimes \Pi_{\mathsf{valid},k,y}
\end{equation*}

Now, we can write $\ket{\psi} = \sum_y \alpha_y \ket{y} \otimes \ket{\phi_y}$, for some $\alpha_y$, $\ket{\phi_y}$. Then,
\begin{align}
\| \Pi_{\mathsf{valid}, k} \mathsf{DecompAll} \ket{\psi}\|^2  &= \sum_y |\alpha_y|^2 \| \Pi_{\mathsf{valid}, k, y} \mathsf{DecompAll} \ket{\phi_y} \|^2 \nonumber\\
&\leq (\sum_y |\alpha_y|^2 ) \cdot 2^{-L+l}  \nonumber\\
&= \|\ket{\psi}\|^2 \cdot  2^{-L+l} = 2^{-L+l}\,, \nonumber
\end{align}
where the second line follows by an application of Lemma \ref{lem: technical}, using the fact that, for each $y$, $\ket{\phi_y}$ inherits the property that the database register contains no claws, and there are at most $l$ indices such that the database contains some pre-image corresponding to each of these indices.
\end{proof}

The following is a corollary of Lemma \ref{lem: structure xor high success}
\begin{cor}
\label{lem: structure}
Let $k$ be a choice of claw-free function pair. Let $\ket{\psi} \in \mathcal{S}_{\mathsf{comp}} \cap \mathcal{S}_{k, \mathsf{no-claw}}$. Suppose 
$$\| \Pi_{\mathsf{valid}, k} \circ \mathsf{DecompAll} \ket{\psi} \|^2 \geq 1 - \epsilon \,,$$
Then, for each $i$, there exists a (not necessarily normalized) state \begin{align}
\ket{\psi'}= & \sum_{\substack{z,d,y,x,e,w \\ D \niton x_0^{y_i}, x_1^{y_i}}}\alpha_{z,d,y,x,e,w,D} \ket{z,d,y,x,e,w} \Big(\ket{D \cup x_{0}^{y_i}} + (-1)^{ z + d\cdot (x_0^y \oplus x_1^y)} \ket{D \cup x_1^{y_i}}\Big) \nonumber\\
+&\sum_{\substack{z,d,y,x,e,w  \\ D \niton x_0^{y_i}, x_1^{y_i}}} \beta_{z,d,y,x,e,w,D} \ket{z,d,y,x,e,w} \Big( \ket{D} + (-1)^{z + d\cdot (x_0^y \oplus x_1^y)} \ket{D \cup \{x_0^{y_i},x_1^{y_i}\}}\Big) \,,
\end{align}
for some coefficients $\alpha_{z,d,y,x,e,w,D}$, $\beta_{z,d,y,x,e,w,D}$, such that
$$\| \ket{\psi} - \ket{\psi}' \|^2 \leq 2 \epsilon \,.$$
\end{cor}
Notice that the above straightforwardly implies that the normalization of $\ket{\psi'}$ is $O(\sqrt{\epsilon})$-close to $\ket{\psi}$ in Euclidean distance.

\begin{proof}
Define $\Pi_{\mathsf{invalid},k} := I - \Pi_{\mathsf{valid},k}$. For a fixed $y$, let 
$$\Pi_{\mathsf{invalid}, k,y} := \sum_{\substack{z,d,D:\\ \forall i, \,\, z_i \oplus d_i \cdot (x_0^{y_i} \oplus x_1^{y_i}) = D(x_0^{y_i}) \oplus D(x_1^{y_i})  }} \ket{z,d}\bra{z,D} \otimes \ket{D}\bra{D} \,.$$

Then, we can write $\Pi_{\mathsf{invalid},k}$ as 
\begin{equation*}
    \Pi_{\mathsf{invalid},k} = \sum_{y} \ket{y}\bra{y} \otimes \Pi_{\mathsf{invalid},k,y}
\end{equation*}

Now, we can write $\ket{\psi} = \sum_y \gamma_y \ket{y} \otimes \ket{\phi_y}$, for some $\gamma_y$, $\ket{\phi_y}$. Then, 
\begin{equation}
\epsilon \geq \| \Pi_{\mathsf{invalid}, k} \circ \mathsf{DecompAll} \ket{\psi} \|^2 = \sum_y |\gamma_y|^2 \| \Pi_{\mathsf{invalid},k,y} \circ \mathsf{DecompAll} \ket{\phi_y}\|^2 \,. \label{eq: 300}
\end{equation}

Since $\ket{\psi} \in S_{\mathsf{comp}}$, then, for each $y$, the state $\ket{\phi_y}$ inherits the same property (as a state on all other registers except the $\mathsf{Y}$ register). Thus, we can apply Lemma \ref{lem: structure xor high success}, and deduce that, for each $y$, for each $i$, there exists a (not necessarily normalized) state 
\begin{align}
\ket{\phi_y'}= & \sum_{\substack{z,d,x,e,w \\ D \niton x_0^{y_i}, x_1^{y_i}}}\alpha_{z,d,y,x,e,w,D} \ket{z,d,x,e,w} \Big(\ket{D \cup x_{0}^{y_i}} + (-1)^{z_i + d_i\cdot (x_0^{y_i} \oplus x_1^{y_i})} \ket{D \cup x_1^{y_i}}\Big) \nonumber\\
+&\sum_{\substack{z,d,x,e,w  \\ D \niton x_0^{y_i}, x_1^{y_i}}} \beta_{z,d,y,x,e,w,D} \ket{z,d,x,e,w} \Big( \ket{D} + (-1)^{z_i + d_i\cdot (x_0^{y_i} \oplus x_1^{y_i})} \ket{D \cup \{x_0^{y_i},x_1^{y_i}\}}\Big) \,,
\end{align} such that 
\begin{equation}
  \|  \ket{\phi_y} - \ket{\phi_y'} \|^2 \leq 2 \cdot \|\Pi_{\mathsf{invalid},k,y} \circ \mathsf{DecompAll} \ket{\phi_y}\|^2 \,.
\end{equation}
Formally, in order to match the syntax of Lemma \ref{lem: structure xor high success}, we append an output register $\mathsf{T}$ to $\ket{\phi_y}$, and apply the unitary that maps $\ket{0}_{\mathsf{T}}\ket{z,d,y} \mapsto \ket{ z_i \oplus d_i \cdot (x_0^{y_i} \oplus x_1^{y_i})}_{\mathsf{T}}\ket{z,d,y} $. Then, applying Lemma \ref{lem: structure xor high success} gives precisely the above.

Now, let $\ket{\psi'} = \sum_y \gamma_y \ket{y} \ket{\phi_y'}$. Then, plugging this into \eqref{eq: 300}, gives 
$$ \|\ket{\psi} - \ket{\psi'} \|^2 \leq 2\epsilon \,.$$
Notice that $\ket{\psi'}$ has the desired form.
\end{proof}

\subsection{Extracting a claw}
\label{sec: extracting a claw}
In this section we complete the proof of perfect unexplainability of Construction \ref{cons: unexplainable scheme}, building on the machinery developed in Section \ref{sec: the structure of strategies}.

Suppose for a contradiction that there exist algorithms $\mathsf{Verify}$ and $\mathsf{Explain}$ which violate Definition \ref{def: perfect unexp}. Then, there exists $m$ such that 
$$\Pr_{\pk}[\mathsf{Verify}(\pk, c, m, w) = 1: (c,w) \gets \mathsf{Explain}(\mathsf{pk},m)]$$
is non-negligible.

Moreover, recall that $\mathsf{Verify}(\mathsf{pk}, c, m, w) = 0$ if $c$ is not in the support of the distribution of $\mathsf{Enc}(\mathsf{pk},m)$ (we say that the tuple $(\pk,c,m)$ is inconsistent in this case). In this section, we will prove a slightly stronger result, and we will only assume that $\mathsf{Verify}(\mathsf{pk}, c, m, w) = 0$ with exponentially small probability (in the security parameter) if $c$ is not in the support of the distribution of $\mathsf{Enc}(\mathsf{pk},m)$.

Without loss of generality, let us assume that $m=0$, so that a valid ciphertext $c = (z,d,y)$ simply corresponds to $L$ valid ``equations'', i.e. $z_i = d_i \cdot (x_0^{y_i} \oplus x_1^{y_i})\oplus H(x_0^{y_i}) \oplus H(x_1^{y_i})$ for all $i \in [L]$.

At a high level, the proof proceeds in two steps.

\begin{itemize}
    \item[(i)] The first key observation is that, because by definition $\mathsf{Verify}$ never accepts an inconsistent tuple  (or except with exponentially small probability), then it must be the case that whenever $\mathsf{Verify}$ accepts, it \emph{must} itself have queried at a superposition of both pre-images in the following more precise sense. 

    For an NTCF key $k$, denote by $\mathsf{Verify}_k$ the algorithm $\mathsf{Verify}$ where we fix the choice of claw-free functions to be $k$. Let $A$ be \emph{any} efficient oracle algorithm that, on input a choice of claw-free functions, outputs tuples $(z,d,y,w)$. Suppose we run a compressed oracle simulation of $A$, followed by $\mathsf{Verify}$, i.e. we run $(\mathsf{Verify}_k \circ A(k))^{\comp}$. Let the state right before measurement of $\mathsf{Verify}$'s output be:
    $$ \alpha \ket{0} \ket{\phi_0} + \beta \ket{1} \ket{\phi_1} \,, $$
    where the first qubit is the output qubit of $\mathsf{Verify}$, and $\ket{\phi_0}$, $\ket{\phi_1}$ are some states on the remaining registers (including the database register). Then, except with negligible probability over the choice of $k$, the following holds: if $\beta$ is non-negligible, then $\ket{\phi_1}$ has weight only on databases containing, for each $i$, either $x_0^{y_i}$ or $x_1^{y_i}$, and moreover, the weights for each pre-image are approximately equal. A bit more precisely, we first argue that, for any state $\ket{\psi}$ which results from a compressed oracle simulation,
    $$ \| (\ket{1}\bra{1}_{\mathsf{Q}} \otimes (I-\Pi_{\mathsf{valid}})]) \mathsf{DecompAll} \circ \mathsf{Verify}_k^{\comp} \ket{\psi} \|^2$$ 
    must be \emph{exponentially small}. Otherwise, this would violate the fact that $\mathsf{Verify}$ accepts inconsistent tuples only with exponentially small probability. This is the content of Lemma \ref{lem: 3}.
    
    Since $\ket{\phi_1}$ is a state resulting from a compressed oracle simulation, then, if $\beta$ is non-negligible, it must be that 
     $$ \| (\ket{1}\bra{1}_{\mathsf{Q}} \otimes \Pi_{\mathsf{valid}}) \mathsf{DecompAll} \circ \mathsf{Verify}_k^{\comp} \ket{\phi_1} \|^2 \geq 1 - O(2^{-\lambda^c}) \,,$$
     for some $c>0$.
    
    We can then appeal to Corollary \ref{lem: structure} to deduce that $\ket{\phi_1}$ has the right structure, i.e. weight only on databases containing, for each $i$, either $x_0^{y_i}$ or $x_1^{y_i}$, with the weights for each pre-image approximately equal.
    \item[(ii)] We then construct an extraction algorithm which succeeds at extracting a claw with non-negligible probability (thus giving a contradiction). To provide some intuition, let us first describe a simplified version of the algorithm in the case $L=1$. This is only sufficient to extract a claw if $\mathsf{Explain}$ succeeds at convincing $\mathsf{Verify}$ with high probability, e.g. $9/10$ (as opposed to merely non-negligible).
\begin{itemize}
    \item[(a)] Run $\mathsf{Explain}^{\comp}(k)$, and measure the output registers to obtain $z,d,y$. Moreover, check if the database register at this point contains a pre-image of $y$. If so, measure it. Denote this by $x_b^y$.
    \item[(b)] Run $\mathsf{Verify}_k^{\comp}$ on the leftover state from the previous step, and measure the output register. Conditioned on ``accept'', measure the database register. If the database contains $x_{\bar{b}}^y$, output the claw $(x_0^y, x_1^y)$.
\end{itemize}

As we also described in the technical overview, the idea behind this algorithm is that, thanks to the first observation, the state conditioned on obtaining ``accept'' in step (b) has weight only on databases containing either $x_0^y$ or $x_1^y$, and moreover, the weights for each pre-image are approximately equal. This implies that, conditioned on observing ``accept'' in step (b), the final measurement of the algorithm is guaranteed to produce one of the two pre-images approximately \emph{uniformly at random}. 

Now, notice that the algorithm already has a constant probability of obtaining one of the two pre-images in step (a) (assuming $\mathsf{Explain}$ succeeds with probability $\frac{9}{10}$). To see this, notice that if this weren't the case, then $\mathsf{Explain}$ would be producing valid equations at most with probability close to $\frac12$ (since this is the probability of producing a valid equation when the database register does not contain any of the two pre-images), and therefore the probability that $\mathsf{Verify}$ accepts $\mathsf{Explain}$'s output would also be at most close to $\frac12$ (instead of being $\frac{9}{10}$). Thus, what is left to show is that the probability of obtaining ``accept'' in step (ii) \emph{conditioned on finding a pre-image in step} (i) is still noticeable (in fact, constant).

To see that this is the case, notice that, using the notation introduced earlier, the fact that $\mathsf{Explain}$ succeeds with probability $\frac{9}{10}$ implies that the coefficient $\beta$ is such that $|\beta|^2 > \frac{9}{10}$. The question is then: if we look at the state after running $\mathsf{Explain}^{\comp}(k)$, which databases are contributing to the weight on $\beta$? The key is that databases that do not contain any pre-image can at most contribute weight $\frac12$ (this is the content of Lemma \ref{cor: 1} from Section \ref{sec: the structure of strategies}, for the special case of $L=1$.). Thus at least $\frac{4}{10}$ of the weight in $|\beta|^2$ must come from databases containing one of the two pre-images. Thus, conditioned on obtaining a pre-image after step (a), the algorithm still has a constant probability of obtaining the other pre-image after the measurement in step (b). 

Things are a bit more tricky if we wish to show that $\mathsf{Explain}$ succeeds at most with negligible probability. Then, the argument above no longer works. The first step is to modify the encryption algorithm to be a parallel repetition of the single-shot algorithm, as in Construction \ref{cons: unexplainable scheme}. Let $L$ denote the number of repetitions. Let us try to carry out a similar argument as before. We can again deduce that $|\beta|^2$ must be non-negligible if $\mathsf{Explain}$ succeeds with non-negligible probability. Now, paralleling our analysis in the $L=1$ case, what we would like to establish is that there exists some (efficiently implementable) measurement $M$ of the database register which satisfies the following:
\begin{itemize}
    \item[(a)] The outcome of $M$ is a string $x$ (potentially a pre-image).
    \item[(b)] Conditioned on outcome $x$, running $\mathsf{Verify}$ on the post-measurement state still leads to acceptance with non-negligible probability.
\end{itemize}
While condition (b) was straightforwardly true in the $L=1$ case, it need not be in the case of a general $L$. The reason is that for a given $y = y_1 \ldots y_L$, there are now exponentially many possible subsets of pre-images of the $y_i$. It could be the case that a very favorable interference of all of these exponentially many branches leads to a \emph{non-negligible} $|\beta|^2$, but conditioned on obtaining any particular $x$, $\mathsf{Verify}$ may only accept with \emph{exponentially small} probability. Our approach to overcome this barrier is different than earlier, and is a bit more involved.

We show that there exists a sufficiently small (polynomial-size) set of measurements such that the following guarantee holds: at least one of these measurements produces a pre-image of some $y_i$ with non-negligible probability. Moreover, the post-measurement state conditioned on producing a pre-image leads to $\mathsf{Verify}$ accepting with non-negligible probability. This is essentially the content of Lemma \ref{lem: last}. We provide some high-level intuition before presenting the formal proof. 

Consider the first image $y_1$. Consider the measurement $M_1$ that checks whether a database contains a pre-image of $y_1$ or not. We can write the state before this measurement as 
$$ \ket{\phi_\text{yes}} + \ket{\phi_{\text{no}}} \,,$$
where $\ket{\phi_\text{yes}}$ is the projection onto the former outcome, and $\ket{\phi_{\text{no}}}$ on the latter.

One of the two \emph{unnormalized} post-measurement states must contribute \emph{non-negligible} weight to acceptance (meaning that running $\mathsf{Verify}$ on it will result in non-negligible amplitude on the accept outcome). In fact, suppose for a contradiction that both contributions were exponentially small, then, no matter how constructive the interference is between the two branches, the resulting weight on accept would still be exponentially small (since there are only two branches here).

Now, if we are lucky and it is  $\ket{\phi_\text{yes}}$ that leads to a non-negligible weight on accept, then we are almost done. We simply perform measurement $M_1$. Then with non-negligible probability we will collapse to $\ket{\phi_\text{yes}}$. Now, $\ket{\phi_\text{yes}}$ is made up of two components: one with databases containing $x_0^{y_1}$, and one with databases containing $x_1^{y_1}$. One of the two (unnormalized) components must contribute non-negligible amplitude to accept. Then, if we measure to extract a pre-image, we collapse to this component with non-negligible probability. 

Suppose instead that we are unlucky, and it is only $\ket{\phi_\text{no}}$ that contributed non-negligible weight to accept. For simplicity, let's take this non-negligible weight to be $1/p$. Then, we can sort of iterate the argument, going to the next pre-image $y_2$. We can define the 2-bit outcome measurement $M_2$ that checks whether the database contains a pre-image of $y_1$ (this is the first bit of the outcome) and whether it contains a pre-image of $y_2$ (this is the second bit).

Notice that we can write:
$$\ket{\phi_\text{no}} = \ket{\phi_\text{no, yes}} + \ket{\phi_\text{no, no}} \,,$$
where $\ket{\phi_\text{no, yes}}$ is the component that has weight only on databases that contain a pre-image of $y_2$ (and no pre-image of $y_1$), and $\ket{\phi_\text{no, no}}$ has weight on databases that do not contain a pre-image of $y_2$ (and still no pre-image of $y_1$). If the (un-normalized) state $\ket{\phi_\text{no, yes}}$ contributes non-negligible weight to accept (upon applying $\mathsf{Verify}$) then we are in good shape, since $M_2$ will collapse to $\ket{\phi_\text{no, yes}}$ with non-negligible probability. Otherwise, almost all of the $1/p$ contribution (except a negligible amount) must come from $\ket{\phi_\text{no, no}}$. 

We can continue to iterate this say $L/2$ times. If we are unlucky every time, it means that $\ket{\phi_\text{no, no, \ldots, no}}$ (with $L/2$ no's) contributes $1/p$ weight (minus a negligible quantity) to acceptance. The key observation now is that states with databases that have a constant fraction of pre-images missing can only lead to acceptance with \emph{exponentially} small probability (this was formalized in Corollary \ref{cor: 1}). So this is a contradiction. 

Thus, along the way, some $\ket{\phi_\text{no, \ldots, no, yes}}$ must be such that $\|\ket{\phi_\text{no, \ldots, no, yes}} \|$ is non-negligible, and moreover $\mathsf{Verify} \ket{\phi_\text{no, \ldots, no, yes}}$ has non-negligible weight on accept.

Note, importantly that we do \emph{not} perform all of the $L/2$ measurements that we defined simultaneously. Instead, only \emph{one} measurement (picked uniformly at random) is performed. The \emph{existence} of some good $\ket{\phi_\text{no, \ldots, no, yes}}$ guarantees that with $\frac{1}{L/2}$ probability we will be making a measurement which results in a collapse to the good $\ket{\phi_\text{no, \ldots, no, yes}}$ with non-negligible probability.

\end{itemize}

We now proceed to the formal proof. Since our construction and security are in the random oracle model, $\mathsf{Verify}$ and $\mathsf{Explain}$ are oracle algorithms. 

As in the previous section, we denote the registers that $\mathsf{Explain}$ acts on as: $\mathsf{Z}, \mathsf{D}, \mathsf{Y}$, corresponding to the equation; $\mathsf{X}, \mathsf{E}$, corresponding to the query registers; $\mathsf{W}$, corresponding to the witness. For ease of notation, we omit keeping track of the input registers and any auxiliary work register. Moreover, recall that we are taking $m=0$ without loss of generality, so we omit keeping track of an $\mathsf{M}$ register.

$\mathsf{Verify}$ acts on registers: $\mathsf{Z, D, Y, W}$, the input registers, and $\mathsf{Q}$ an output register. Again, we omit keeping track of an input register containing the public key. 

When running a compressed oracle simulation, we denote by $\mathsf{O}$ the database register.

For ease of notation we denote $Q := \mathsf{Verify}$ and $P:= \mathsf{Explain}$.

Denote by $Q_k$ the algorithm $Q$ with fixed input the NTCF key $k$, and similarly denote by $Q_{k,z,d,y}$ the algorithm $Q$ with fixed inputs $k,z,d,y$. For $y \in \{0,1\}^{m \times L}$, we write $y = y_1 \ldots y_L$, where $y_i \in \{0,1\}^m$ for each $i$ (here $m$ is the size of the bit-decomposition of an element in the co-domain of $f_{k,0}, f_{k,1}$).

Recall that we denote by $\mathsf{PhO}$ and $\comp$ respectively the phase oracle and the compressed phase oracle (as defined in Section \ref{sec: comp oracles}). Finally, using a similar notation as in the previous section, we denote by $\mathcal{S}_{\mathsf{comp}}$ the set of (normalized) states on registers $\mathsf{ZDYXEWO}$ of the form:
$$ \sum_{z,d,y,x,e,w,X \subseteq \{0,1\}^n } \alpha_{z,d,y,x,e,w, X} \ket{z,d,y,x,e,w}\ket{D_X} \,.$$
These are states that can be reached by running a compressed oracle simulation.

For an NTCF key $k$, let $\mathcal{S}_{k, \mathsf{no-claw}}$ be the set of (normalized) states, on registers $\mathsf{ZDYXEWO}$, that have zero weight on $y, D$ such that $D$ contains a claw for $y_i$ for some $i$. Let $\Pi_{\mathsf{valid},k}$, acting on registers $\mathsf{ZDYXEO}$, be the projector onto states for which all $L$ equations are valid, with respect to key $k$. We will omit writing the subscript $k$ when it is clear from the context.

\begin{lem}
\label{lem: 3}
There exists a constant $c>0$ such that, for all $\lambda$, for all $\ket{\psi} \in \mathcal{S}_{\mathsf{comp}}$, for all $k$:
$$ \| (\ket{1}\bra{1}_{\mathsf{Q}} \otimes (I-\Pi_{\mathsf{valid}})]) \mathsf{DecompAll} \circ Q_k^{\comp} \ket{\psi} \|^2 \leq 2^{-\lambda^c} \,.$$
\end{lem}
\begin{proof}
Suppose for a contradiction that, for all $c>0$, there exists $\lambda$,  $\ket{\psi} \in \mathcal{S}_{\mathsf{comp}}$, and $k$, such that:
\begin{equation} 
\label{eq: 250}
\| (\ket{1}\bra{1}_{\mathsf{Q}} \otimes (I-\Pi_{\mathsf{valid}})]) \mathsf{DecompAll} \circ Q_k^{\comp} \ket{\psi} \|^2 > 2^{-\lambda^c} \,. \end{equation}

Then, we will show that, for all $c>0$, there exists $\lambda$, $z,d,y, H$, and a state $\ket{\phi_H}$ on registers $\mathsf{XEW}$ (plus some purifying registers) such that $z \neq  d \cdot (x_0^y \oplus x_1^y) \oplus H(x_0^y) \oplus H(x_1^y)$, but 
$$ \| \ket{1}\bra{1}_{\mathsf{Q}} Q^H_{k} (\ket{z,d,y} \otimes \ket{\phi_H}   )\|^2 > 2^{-\lambda^c}\,.$$
which would contradict the fact the fact that $Q$ accepts inconsistent answers with exponentially small probability.

For a function $H \in \mathsf{Bool}(n)$, let $\Pi_{H}$ be the projector onto the database corresponding to the truth table of $H$. Then, 
$$ I-\Pi_{\mathsf{valid}} = \sum_{z,d,y} \sum_{\substack{H: \\z \neq d \cdot (x_0^y \oplus x_1^y) \oplus H(x_0^y) \oplus H(x_1^y) }} \ket{z,d,y}\bra{z,d,y} \otimes \Pi_H \,.$$

We can thus rewrite Equation \eqref{eq: 250} as saying that, for all $c$, there exists $\lambda$, and a state $\ket{\psi} = \sum_{z,d,y} \ket{z,d,y} \otimes \ket{\phi_{z,d,y}} \in \mathcal{S}_{\mathsf{comp}}$, such that
\begin{equation}
\label{eq: 251}
\sum_{z,d,y}  \sum_{\substack{H: \\z \neq d \cdot (x_0^y \oplus x_1^y) \oplus H(x_0^y) \oplus H(x_1^y) }} \| (\ket{1}\bra{1}_{\mathsf{Q}} \otimes \Pi_H ) \,  \mathsf{DecompAll} \circ Q^{\comp}_{k} (\ket{z,d,y} \otimes \ket{\phi_{z,d,y}}) \|^2 > 2^{-\lambda^c} \,.
\end{equation}

This implies that there exists $d, z, y$ such that
\begin{equation}
\label{eq: 251}
\sum_{\substack{H: \\z \neq d \cdot (x_0^y \oplus x_1^y) \oplus H(x_0^y) \oplus H(x_1^y) }} \| (\ket{1}\bra{1}_{\mathsf{Q}} \otimes \Pi_H ) \,  \mathsf{DecompAll} \circ Q^{\comp}_{k} \ket{z,d,y} \otimes \ket{\phi_{z,d,y}}_{\mathsf{XEWO}} \|^2 > 2^{-\lambda^c} \,.
\end{equation}

Let $\ket{\phi} : = \ket{\phi_{z,d,y}}$ from now on, where $z,d,y$ satisfy the above equation.

Notice that $ \mathsf{DecompAll} \circ Q_{k, z,d,y}^{\comp} \ket{\psi} =   Q_{k,m}^{\mathsf{PhO}} \circ \mathsf{DecompAll} \ket{\psi}$ for any $\ket{\psi} \in \mathcal{S}_{\mathsf{comp}}$. Hence, the LHS of \eqref{eq: 251} is equal to
\begin{equation}
\label{eq: 252}
     \sum_{\substack{H: \\z \neq d \cdot (x_0^y \oplus x_1^y) \oplus H(x_0^y) \oplus H(x_1^y) }} \| ( \ket{1}\bra{1}_{\mathsf{Q}} \otimes  \Pi_H )\, Q^{\mathsf{PhO}}_{k} \circ \mathsf{DecompAll} \ket{ z,d,y} \otimes \ket{\phi}\|^2  \,.
\end{equation}
Now, we can write $ \mathsf{DecompAll} \ket{\phi} = \sum_H \ket{\phi_H}_{\mathsf{XEW}} \otimes \ket{H}_{\mathsf{O}} $ for some $\ket{\phi_H}$, where $\ket{H}$ denote the database corresponding to the truth table of the function $H$. Then, \eqref{eq: 252} is equal to 
\begin{align}
\label{eq: 253}
    &\sum_{\substack{H: \\z \neq  d \cdot (x_0^y \oplus x_1^y) \oplus H(x_0^y) \oplus H(x_1^y) }} \| ( \ket{1}\bra{1} \otimes  \Pi_H )\, Q^{\mathsf{PhO}}_{k} \ket{z,d,y} \otimes \ket{\phi_{H}} \otimes \ket{H}  \|^2  \\
    =&  \sum_{\substack{H: \\z \neq  d \cdot (x_0^y \oplus x_1^y) \oplus H(x_0^y) \oplus H(x_1^y) }} \| \ket{1}\bra{1}_{\mathsf{Q}} \, Q^{\mathsf{PhO}}_{k}\ket{z,d,y} \otimes \ket{\phi_{H}} \otimes \ket{H}  \|^2 
\end{align}
where the last line uses the fact that $Q^{\mathsf{PhO}}_{k}$ commutes with $\Pi_H$. It follows straightforwardly that there exists $H$ such that $z \neq m \oplus d \cdot (x_0^y \oplus x_1^y) \oplus H(x_0^y) \oplus H(x_1^y)$ and 
$$ \| \ket{1}\bra{1}_{\mathsf{Q}} \, Q^{\mathsf{PhO}}_{k}\ket{z,d,y} \otimes \ket{\phi_{H}} \otimes \ket{H}  \|^2 > 2 \cdot 2^{-\lambda^c} \,.$$
The latter is equivalent to
$$  \| \ket{1}\bra{1}_{\mathsf{Q}} \, Q^{H}_{k} \ket{z,d,y}\otimes \ket{\phi_{H}}  \|^2 > 2 \cdot 2^{-\lambda^c}\,,$$
as desired.

\end{proof}


Now, suppose that we knew the following lemma were true. Then we will show that this yields an algorithm to extract a claw. Let $I \subseteq [L]$, and $b_I = (b_i: i \in I)$ where $b_i \in \{0,1, \perp\}$. Let $\Pi^{b_I}_{y,I}$ be the projector (acting on register $\mathsf{O}$) onto databases $D_X$ such that: for $i \in I$, if $b_i = \perp$, then $X \niton x_0^{y_i}, x_1^{y_i}$, otherwise $X$ contains $x_{b_i}^{y_i}$. Notice that for any $y$ and $I \subseteq \{0,1\}^L$, the projective measurement $M_{y,I} := \{\Pi^{b_I}_{y,I}\}_{b_I}$ can be performed efficiently.

In the following lemma, let $\ket{\psi_k}$ be the state of $P^{\comp}(k)$ just before measurement of the output.

\begin{lem}
\label{lem: last}
There exists a non-negligible function $g$, a sequence of sets $\mathcal{J}_{\lambda} \subseteq \K$, with $\frac{|\mathcal{J}_{\lambda}|}{|\K|} \geq g(\lambda)$, an efficiently generatable sequence of sets $\mathcal{S}_{good, \lambda}$ of pairs $(I, b_I)$, where $I \subseteq [L]$, and $b_I = (b_i: i \in I)$ where $b_i \in \{0,1, \perp\}$, such that the following holds for all $\lambda$:
\begin{itemize}
\item[(a)] $|\mathcal{S}_{good, \lambda}| = \poly(\lambda)$.
    \item[(b)] $b_I \neq \perp^{|I|}$ for all $(I,b_I) \in \mathcal{S}_{good, \lambda}$.
    \item[(c)] For $k\in \mathcal{J}_{\lambda}$, there exists $(I,b_I) \in \mathcal{S}_{good, \lambda}$ such that $\left \| \ket{1}\bra{1}_{\mathsf{Q}} Q_k^{\comp} (\sum_y \ket{y}\bra{y} \otimes \Pi_{y,I}^{b_I}) \ket{\psi_k} \right\|^2 \geq g(\lambda)$.
\end{itemize}
\end{lem}

With this lemma in hand, we can show that the following algorithm $\mathsf{Ext}$ succeeds at extracting a claw. 
\begin{itemize}
    \item[(i)] On input $k$, run $P^{\comp}(k)$ and obtain $z,d,y$. 
    \item[(ii)] Sample $(I,s_I) \gets \mathcal{S}_{good, \lambda}$, and perform the measurement $M_{y,I}$ on the database register of the leftover state from step (i). If the outcome is $s_I$, then the database contains at least a pre-image of $y_i$ for some $i \in I$. Let $x_b^{y_i}$ be this pre-image. 
    \item[(iii)] Run $Q_k^{\comp}$, and measure the database register. If the latter contains $x_{\bar{b}}^{y_i}$ output $(x_0^{y_i}, x_1^{y_i})$.
\end{itemize}

Let $\mathcal{J}_{\lambda}$ and $\mathcal{S}_{good, \lambda}$ be as in Lemma \ref{lem: last}, and let $(I, b_I) \in \mathcal{S}_{good, \lambda}$ be such that it satisfies condition (c) in Lemma \ref{lem: last}.

The above algorithm samples this pair with probability $\frac{1}{|\mathcal{S}_{good, \lambda}|} = \frac{1}{\poly(\lambda)}$ in step (ii). There exists some negligible function $\negl$, such that, for any $k \in \mathcal{J}_{\lambda}$, up to $\negl(\lambda)$ distance, we have
\begin{equation}
     Q_k^{\comp} (\sum_y \ket{y}\bra{y} \otimes \Pi_{y,I}^{b_I}) \ket{\psi_k} = \alpha \ket{0} \ket{\phi_0} + \beta \ket{1} \ket{\phi_1}
     \label{eq: 510}
\end{equation}
for some $\alpha, \beta$ such that $ |\beta| \geq g(\lambda) $, and for some  $\ket{\phi_0},\ket{\phi_1} \in \mathcal{S}_{\mathsf{comp}} \cap \mathcal{S}_{k, \mathsf{no-claw}}$. To see that $\ket{\phi_0},\ket{\phi_1} \in \mathcal{S}_{\mathsf{comp}}$, notice that $\mathcal{H}_\mathsf{Q}\otimes S_{\mathsf{comp}}$ is invariant under the action of $\Pi_{y,I}^{b_I}$, and is invariant under the action of $M \otimes \tilde{O}$, where $M$ is any linear operator on $H_{\mathsf{QZDYW}}$, and $\tilde{O} = \mathsf{StdDecomp}^{-1} \circ O \circ \mathsf{StdDecomp}$. Moreover, \eqref{eq: 510} must be true for some $\ket{\phi_0},\ket{\phi_1} \in  \mathcal{S}_{k, \mathsf{no-claw}} $, otherwise the following would be an efficient algorithm to obtain a claw: run $P^{\comp}(k)$; measure the $\mathsf{Y}$ register to obtain $y$; perform the measurement $M_{y,I}$ to obtain $b_I$ followed by a standard basis measurement of the database register. With non-negligible probability the database would contain a claw for $y_i$ for some $i$.

For the same reason as above, since $\ket{\psi_k} \in \mathcal{S}_{\mathsf{comp}}$, it is also the case that $(\sum_y \ket{y}\bra{y} \otimes \Pi_{y,I}^{b_I}) \ket{\psi_k} \in \mathcal{S}_{\mathsf{comp}}$. Hence, we can apply Lemma \ref{lem: 3} to deduce that
$$ \| (\ket{1}\bra{1}_{\mathsf{Q}} \otimes (I-\Pi_{\mathsf{valid}})]) \mathsf{DecompAll} \circ Q_k^{\comp} (\sum_y \ket{y}\bra{y} \otimes \Pi_{y,I}^{b_I}) \ket{\psi_k}\|^2 \leq 2^{-\lambda^c} \,.$$
for some $c>0$.

Plugging \eqref{eq: 510} into the above, gives

$$ |\beta|^2 \cdot \| (I-\Pi_{\mathsf{valid}}) \mathsf{DecompAll} \ket{\phi_1} \|^2 \leq 2^{-\lambda^c} \,.$$

Hence, 
$$ |\beta|^2 \cdot \| (\Pi_{\mathsf{valid}}) \mathsf{DecompAll} \ket{\phi_1} \|^2 \geq |\beta|^2 - 2^{-\lambda^c} \,. $$



Since $\beta$ is non-negligible, then there exists a polynomial $p$, such that $\beta(\lambda) > \frac{1}{p(\lambda)}$ for infinitely many $\lambda$. Let $\Lambda$ denote this set. Then, we have that, for $\lambda \in \Lambda$, 
$$\| (\Pi_{\mathsf{valid}}) \mathsf{DecompAll} \ket{\phi_1} \|^2 \geq  1- O(2^{-\lambda^{c'}})\,. $$
for some $c'>0$.

We can apply Corollary \ref{lem: structure} to deduce that, for each $i$, $\ket{\phi_1}$ is $O(2^{-\lambda^{c'}})$-close to a state of the form 
\begin{align}
    & \sum_{\substack{z,d,y,x,e,w \\ D \niton x_0^{y_i}, x_1^{y_i}}}\alpha_{z,d,y,x,e,w,D} \ket{z,d,y,x,e,w} \Big(\ket{D \cup \{x_{0}^{y_i}\}} + (-1)^{d\cdot(x_0^{y_i} + x_1^{y_i})} \ket{D \cup \{x_1^{y_i}\}}\Big) \nonumber\\
+&\sum_{\substack{z,d,y,x,e,w  \\ D \niton x_0^{y_i}, x_1^{y_i}}} \beta_{z,d,y,x,e,w,D} \ket{z,d,y,x,e,w} \Big( \ket{D} + (-1)^{d\cdot(x_0^{y_i} + x_1^{y_i})} \ket{D \cup \{x_0^{y_i},x_1^{y_i}\}}\Big) \,,
\end{align}    
for some $\alpha_{z,d,y,x,e,w,D}$ and $\beta_{z,d,y,x,e,w,D}$.

Now, notice that, for all but a negligible fraction of $k$, the weight on the second branch must be negligible in $\lambda$, otherwise this would yield an algorithm to recover a claw. Thus, for any polynomial $q$, for sufficiently large $\lambda \in \Lambda$ (and for all but negligible fraction of $k \in \mathcal{J}$), we have that 
$\ket{\phi_1}$ is $\frac{1}{q(\lambda)}$-close to a state of the form 
\begin{equation}
    \sum_{\substack{z,d,y,x,e,w \\ D \niton x_0^{y_i}, x_1^{y_i}}}\alpha_{z,d,y,x,e,w,D} \ket{z,d,y,x,e,w} \Big(\ket{D \cup \{x_{0}^{y_i}\}} + (-1)^{m+d\cdot(x_0^{y_i} + x_1^{y_i})} \ket{D \cup \{x_1^{y_i}\}}\Big) \,.
\end{equation}

Then, for any polynomial $q$, and for sufficiently large $\lambda \in \Lambda$, conditioned on $\mathsf{Ext}$ having sampled a pair $(I,s_I)$ satisfying condition (c) (which happens with probability $\frac{1}{\poly(\lambda)}$), the algorithm $\mathsf{Ext}$ obtains a different pre-image at step (iii) (for the same $y_i$) than the one obtained at step (ii) with probability $1/2 - \frac{1}{q(\lambda)}$.

Thus, all in all, $\mathsf{Ext}$ outputs a claw with non-negligible probability.

To conclude the proof of Theorem \ref{thm: 2}, we are left to prove Lemma \ref{lem: last}.

\begin{proof}[Proof of Lemma \ref{lem: last}]
For $i \in \{1,\ldots,L/2\}$, define $I_i := [i]$, and for $s \in \{0,1\}$, define $b_{i, s} := (\perp^{i-1}, s)$.

By hypothesis, $P$ returns a valid equation and proof with non-negligible probability. This implies that there exists a polynomial $p$, and an infinite sized set $\Lambda$ such that the following holds for all $\lambda \in \Lambda$: there exists a set $\mathcal{J}_{\lambda} \subseteq \K$ with $\frac{|\mathcal{J}_{\lambda}|}{|\K|} \geq 1/p(\lambda)$ such that, for $k \in \mathcal{J}_{\lambda}$, $\Pr[ \textnormal{Verifier accepts}|k] \geq \frac{1}{p(\lambda)}$. For the rest of the proof, we restrict our attention to $\lambda \in \Lambda$ and $k \in \mathcal{J}_{\lambda}$.

Denote by $\ket{\psi_k}$ the state of $P^{\comp}(k)$ just before measurement of the output. We argue that it must be the case that for any polynomial $q$, for all large enough $\lambda$, for all but a negligible fraction of $k \in \mathcal{J}_{\lambda}$,
$$ \left \| \ket{1}\bra{1}_{\mathsf{Q}} Q_k^{\comp} (\sum_y \ket{y}\bra{y} \otimes \Pi_{y,I_{L/2}}^{(\perp)^{\frac{L}{2}}}) \ket{\psi_k} \right\|^2 \leq \frac{1}{q(\lambda)}\,.$$

Suppose not, then there is some polynomial $q$, and a polynomial $r$ such that for infinitely many $\lambda \in \Lambda$, there exists some subset $\mathcal{J}'_{\lambda} \subseteq \mathcal{J}_{\lambda} $, with $|\mathcal{J}'_{\lambda}| \geq \frac{1}{r(\lambda)} | \mathcal{J}_{\lambda}| \geq \frac{1}{r(\lambda)p(\lambda)}$, such that
\begin{align}
    &\left \| \ket{1}\bra{1}_{\mathsf{Q}} (Q_k^{\mathsf{PhO}} \circ \mathsf{DecompAll}) (\sum_y \ket{y}\bra{y} \otimes \Pi_{y,I_{L/2}}^{(\perp)^{\frac{L}{2}}}) \ket{\psi_k} \right\|^2 \nonumber \\
    = &\left \| \ket{1}\bra{1}_{\mathsf{Q}} Q_k^{\comp} (\sum_y \ket{y}\bra{y} \otimes \Pi_{y,I_{L/2}}^{(\perp)^{\frac{L}{2}}}) \ket{\psi_k} \right\|^2 \nonumber \\
    >& \frac{1}{q(\lambda)} 
    \label{eq: 560}
\end{align}
where the first equality follows from the fact that, for $\phi \in \mathcal{S}_{\mathsf{comp}}$, $Q_k^{\mathsf{PhO}} \circ \mathsf{DecompAll} = \mathsf{DecompAll} \circ Q_k^{\comp}$, and that $\ket{1}\bra{1}_{\mathsf{Q}}$ commutes with $\mathsf{DecompAll}$. Let $\Lambda' \subseteq \Lambda$ be this infinitely sized set of $\lambda$. 

For convenience, denote $\ket{\psi'_k} := \left(\sum_y \ket{y}\bra{y} \otimes  \Pi_{y,I_{L/2}}^{(\perp)^{\frac{L}{2}}} \right)\ket{\psi_k}$. Let $\Pi_{k,y,\mathsf{claw}}$ denote a projection onto databases that contain a claw for $y_i$ for some $i$.

Then, for all $\lambda \in \Lambda'$, for all but a negligible fraction of $k \in \mathcal{J}'_{\lambda}$, 
$$ \big\| (\sum_y \ket{y}\bra{y} \otimes  \Pi_{k,y, \mathsf{claw}}) \ket{\psi'_k} \big\|^2 = \negl(\lambda) \,.$$
Otherwise, we immediately obtain a strategy to recover a claw with non-negligible probability.

Thus, for all $\lambda \in \Lambda'$, for all but a negligible fraction of $k \in \mathcal{J}'_{\lambda}$, we have that the (unnormalized) state $\ket{\psi'_k}$ is negligibly close to a state $C \cdot \ket{\psi_k''}$, for some $0 \leq C \leq 1$ and $\ket{\psi_k''} \in \mathcal{S}_{\mathsf{no-claw}}$.

Applying Corollary \ref{cor: 1}, we have that, for all $\lambda \in \Lambda'$, for all but a negligible fraction of  $k \in \mathcal{J}'_{\lambda}$,
\begin{equation}
    \big\| \Pi_{\mathsf{valid}} \mathsf{DecompAll} \ket{\psi'_k} \big\|^2 \leq 2^{-L/2} +\negl(\lambda) \,. \label{eq: 570}
\end{equation}

Combining \eqref{eq: 560} and \eqref{eq: 570}, we get that, for all $\lambda \in \Lambda'$, for all but a negligible fraction of $k \in \mathcal{J}'_{\lambda}$,
\begin{equation}
    \big\| (\ket{1}\bra{1}_{\mathsf{Q}} (Q_k^{\mathsf{PhO}}) \circ  (I-\Pi_{\mathsf{valid}}) \circ \mathsf{DecompAll} \ket{\psi'} \big\|^2  > \frac{1}{q(\lambda)}(1-2^{L/2}) - \negl(\lambda)\,. 
\end{equation}

Now, we can write 
$$ \mathsf{DecompAll} \ket{\psi'_k}= \frac{1}{\sqrt{2^{n+1}}}\sum_{H} \ket{\psi_{k,H}} \ket{H} \,,$$
for some $\ket{\psi_{k,H}}$, where here we are abusing notation slightly and equating a full database to the function it defines (also note that $2^{n+1}$ is the number of functions from $n$ bits to $1$ bit).

Then, for all $\lambda \in \Lambda'$, for all but a negligible fraction of $k \in \mathcal{J}'_{\lambda}$, we have 
\begin{equation}
    \frac{1}{2^{n+1}} \sum_H \big\| (\ket{1}\bra{1}_{\mathsf{Q}} (Q_k^{\mathsf{PhO}}) \circ  (I-\Pi_{\mathsf{valid}})\ket{\psi_{k,H}} \ket{H} \big\|^2  > \frac{1}{q(\lambda)}(1-2^{L/2}) - \negl(\lambda)\,. \label{eq: 30}
\end{equation}

Then \eqref{eq: 30} implies that there exists $H$ such that:
\begin{equation}
    \big\| (\ket{1}\bra{1}_{\mathsf{Q}} (Q_k^{\mathsf{PhO}}) \circ  (I-\Pi_{\mathsf{valid}})\ket{\psi_{k,H}} \ket{H} \big\|^2  \geq \frac{1}{2q(\lambda)}(1-2^{L/2}) \,.
\end{equation}

We can further write $\ket{\psi_{k,H}} = \sum_{z,d,y} \ket{z,d,y}\ket{\psi_{k,z,d,y,H}}$. Then, all in all, we get that there exists $k$, $H$ and $z,d,y$ such that
\begin{align*}
    \Pr[\mathsf{acc} =1   \, \land \,  z_i \neq d_i \cdot (x_0^{y_i} \oplus x_1^{y_i}) \oplus H(x_0^{y_i}) \oplus  &H(x_1^{y_i}) \textnormal{ for some $i$ } : \mathsf{acc} \gets Q_k^{H} \ket{z,d,y} \ket{\psi_{k,z,d,y,H}}] \\
    &\geq \frac{1}{2q(\lambda)}(1-2^{L/2})\,.
\end{align*}

This contradicts the hypothesis that $Q$ accepts invalid $z,d,y$'s only with exponentially small probability. Hence, we have shown that, for any polynomial $q$, for all large enough $\lambda \in \Lambda$, for all but a negligible fraction of $k \in \mathcal{J}_{\lambda}$,
\begin{equation}
\left \| \ket{1}\bra{1}_{\mathsf{Q}} Q_k^{\comp} (\sum_y \ket{y}\bra{y} \otimes \Pi_{y,I_{L/2}}^{(\perp)^{\frac{L}{2}}}) \ket{\psi_k} \right\|^2 \leq \frac{1}{q(\lambda)}\,.  \label{eq: 31}
\end{equation}

Let us restrict our attention to $\lambda, k$ as above.
Recall that, for $i \in \{1,\ldots,L/2\}$, we have defined $I_i := [i]$, and for $s \in \{0,1\}$, $b_{i, s} := (\perp^{i-1}, s)$.


We argue that there exists $i \in [\frac{L}{2}]$ and $s \in \{0,1\}$ such that
$$ \left \| \ket{1}\bra{1}_{\mathsf{Q}} Q_k^{\comp} (\sum_y \ket{y}\bra{y} \otimes \Pi_{y,{I_i}}^{b_{i,s}}) \ket{\psi} \right\|^2 \geq \frac{1}{p(\lambda)L^2} \,.$$

Suppose not, then a straightforward argument (which uses the definition of $p$ at the start of this proof) shows that it must be
\begin{equation}
\left \| \ket{1}\bra{1}_{\mathsf{Q}} Q_k^{\comp} (\sum_y \ket{y}\bra{y} \otimes \Pi_{y,I_{L/2}}^{(\perp)^{\frac{L}{2}}}) \ket{\psi} \right\|^2 \geq \frac{1}{p(\lambda)}\left(1-\frac{1}{L^2}\right)^{\frac{L}{2}}\,,
\end{equation}
which is greater than $\frac{1}{2p(\lambda)}$ for large enough $L$. This contradicts \eqref{eq: 31} for the choice $q = 2\cdot p$. 

\end{proof}

\section{Stronger protection against coercion before-the-fact using quantum ciphertexts}
\label{sec:quantum-ciphertexts}

In this work, we have focused on quantum encryption algorithms that output a \emph{classical} ciphertext. Here, we briefly discuss the possibility of protecting against coercion in the more general setting where (part of) the ciphertext is allowed to be a quantum state. We outline how, in this setting, it is possible to achieve an even stronger notion of protection against coercion before-the-fact. Recall from Section~\ref{sec: before-the-fact} that a perfectly unexplainable encryption scheme (like the one we presented in Section~\ref{sec: unexp scheme}) protects against coercion before-the-fact (in the formal sense of Definition~\ref{def: before-attack}). Informally, the latter notion says that there is no way for an attacker to prescribe to a sender how to encrypt (in a way that it can later verify) \emph{before} seeing the public key. This is in contrast to the classical world, where an attacker can always prescribe the input randomness, even without knowing the public key. Now, what if an attacker knows the public key? Is there any way to protect against coercion before-the-fact?

If the ciphertext is classical, then there is no hope: the attacker can just honestly compute a ciphertext, and prescribe to the sender that she should submit this ciphertext. Since the ciphertext is classical (and known to the attacker), the attacker can always verify that the sender submitted what was prescribed. However, if part of the ciphertext is allowed to be a \emph{quantum} state, then it becomes in principle possible that the attacker might not be able to verify. Here, we sketch a scheme satisfying this stronger notion of protection against coercion. The scheme is a variation on the ideas already explored in this paper. 
Let $\{(f_{k,0},f_{k,1})\}_k$ be a family of trapdoor claw-free function pairs. For simplicity, assume they are not ``noisy''.

The public key is again a trapdoor claw-free function key $k$, and the secret key is a corresponding trapdoor~$t_k$.
\begin{itemize}
    \item \emph{Encryption}: To encrypt a bit $m$, under public key $k$:
    \begin{itemize}
        \item Create the uniform superposition $\sum_{b\in\{0,1\} , x \in \{0,1\}^n} \ket{b}\ket{x}$. Then, compute $f_{k,0}$ and $f_{k,1}$ in superposition, controlled on the first qubit. The resulting state is
        $$ \sum_{b\in\{0,1\} , x \in X} \ket{b}\ket{x}\ket{f_{k,b}(x)} \,.$$
        \item Measure the image register, and let $y$ be the outcome. The leftover quantum state is
        $$ \frac{1}{\sqrt{2}} \ket{0} \ket{x_0} + \frac{1}{\sqrt{2}} \ket{1} \ket{x_1} \,,$$
        where $x_0$ and $x_1$ are the pre-images of $y$.
        Output the ciphertext $$ \left( \frac{1}{\sqrt{2}} \ket{0} \ket{x_0} + (-1)^m \frac{1}{\sqrt{2}}\ket{1} \ket{x_1}, y \right) \,,$$
        which is obtained by applying $Z^m$ to the first qubit.
    \end{itemize}
    \item \emph{Decryption}: On input $c = (\ket{\psi}, y)$, use the trapdoor $t_k$ to compute the pre-images $x_0, x_1$ of $y$. For $b \in \{0,1\}$, let $\ket{\phi_b} = \frac{1}{\sqrt{2}} \ket{0} \ket{x_0} + (-1)^b \frac{1}{\sqrt{2}} \ket{1} \ket{x_1}$. Perform the measurement 
    $$\{\ket{\phi_0}\bra{\phi_0}, \ket{\phi_1}\bra{\phi_1}, I - \ket{\phi_0}\bra{\phi_0} - \ket{\phi_1}\bra{\phi_1} \} \,.$$
    assigning outcomes $0$, $1$, and ``$\perp$'' respectively to the three projectors. Output the measurement outcome.
\end{itemize}
Note that this scheme satisfies the following property. A sender who receives from an attacker a ciphertext encrypting some bit $m$, i.e.\ a ciphertext of the form $\left( \frac{1}{\sqrt{2}} \ket{0} \ket{x_0} + (-1)^m \frac{1}{\sqrt{2}}\ket{1} \ket{x_1}, y \right)$, can turn this into a ciphertext encrypting the opposite bit $\bar{m}$ by applying a $Z$ gate to the first qubit. Crucially, the (computationally bounded) attacker cannot distinguish $\frac{1}{\sqrt{2}} \ket{0} \ket{x_0} + \frac{1}{\sqrt{2}}\ket{1}\ket{x_1}$ from $\frac{1}{\sqrt{2}} \ket{0} \ket{x_0} - \frac{1}{\sqrt{2}}\ket{x_1}$ (even though the attacker created the original state, and might have kept some side information!). This is because the ability to distinguish would imply an efficient algorithm to extract a claw. This can be seen, for example, as a consequence of the equivalence between distinguishing and ``swapping''~\cite{aaronson2020hardness}. We remark that for this reduction to hold it is crucial that part of the ciphertext that the attacker gives to the sender is \emph{classical} (namely, the string $y$).

\bibliographystyle{alpha}
\bibliography{references}

\end{document}